\documentclass[letterpaper,twocolumn,10pt]{article}
\usepackage{usenix}

\newtheorem{ass}{Assumption}
\newtheorem{theo}{Theorem}[section]
\newtheorem{lem}[theo]{Lemma}
\newtheorem{prop}[theo]{Proposition}
\newtheorem{cor}[theo]{Corollary}
\newtheorem{definition}{Definition}
\newtheorem{rem}{Remark}

\newenvironment{proof}{\noindent\textbf{Proof.}\ }{\hfill$\square$\par}
\newcommand{\defin}{\stackrel{\rm def}{=}}

\newcommand{\B}{\textcolor{black}}
\newcommand{\abs}[1]{\left|#1\right|}

\def \cA {\mathcal{A}}
\def \cS {\mathcal{S}}
\def \cO {\mathcal{O}}

\newcommand{\rSum}[3]{\sum\limits_{#1 = #2}^{#3}}

\newcommand{\bits}{\{0,1\}}
\newcommand{\Supp}[1]{\mathtt{Supp}(#1)}
\newcommand{\Ex}[1]{\mathbb{E} \left[ #1 \right ]}
\newcommand{\Exf}[2]{\mathbb{E}_{#1}\left[ #2 \right]}
\newcommand{\argmax}{\mathop{\mathrm{arg\,max}}}

\newcommand{\prdis}[2]{ \Pr\limits_{#1}\left[ #2 \right] }
\newcommand{\pr}[1]{\Pr\left[ #1 \right]}

\newcommand{\rX}{X}
\newcommand{\rY}{Y}
\newcommand{\rZ}{Z}

\newcommand{\Dist}{\mathcal{P}} % Joint distribution with observation and label
\newcommand{\bbcP}[2]{\mathbf{P}\left[#1, #2\right]}

\newcommand{\kNNclassifier}[1]{\phi^{\mathtt{NN}}_{k, #1}}
\newcommand{\MixtureD}[2]{\left[#1\right]_{#2}}
\newcommand{\Mech}{M} % Mechanism

\newcommand{\DB}{D} % Database

\newcommand{\bbe}[1]{\mathtt{BB}^{#1}}% BayBox: A Black-Box Bayesian Classification Approach for f-DP Estimation

\newcommand{\ptlr}[2]{\B{\mathtt{PLRT}}^{#1}_{#2}} % PTLR: A Perturbed Likelihood Ratio Test Algorithm for f-DP Estimation

\usepackage{usenix,epsfig,endnotes,amsmath,amssymb}
\usepackage{mdframed}
\usepackage{booktabs} % For professional-looking tables
\usepackage[letterpaper,margin=1in]{geometry}
\usepackage{longtable} % For tables spanning multiple pages (optional)
\usepackage[normalem]{ulem}
\usepackage{dirtytalk}
\begin{document}

%don't want date printed
\date{}

%make title bold and 14 pt font (Latex default is non-bold, 16 pt)
\title{\Large General-Purpose $f$-DP Estimation and Auditing in a Black-Box Setting \bf }

 \author{
     {\rm Önder Askin}$^{1}$, 
     {\rm Holger Dette}$^{1}$, 
     {\rm Martin Dunsche}$^{1,}$\thanks{Corresponding author: \texttt{martin.dunsche@rub.de}}, 
     {\rm Tim Kutta}$^{2}$, 
     {\rm Yun Lu}$^{3}$, 
     {\rm Yu Wei}$^{4}$, 
     {\rm Vassilis Zikas}$^{4}$\thanks{Authors are listed in alphabetical order.} \\
     \\
     $^1$Ruhr-University Bochum\\
     $^2$Aarhus University \\
     $^3$University of Victoria \\
     $^4$Georgia Institute of Technology\\
 }

\maketitle

% Use the following at camera-ready time to suppress page numbers.
% Comment it out when you first submit the paper for review.
\thispagestyle{empty}

\begin{abstract}
In this paper we propose new methods to statistically assess $f$-Differential Privacy ($f$-DP), a recent refinement of differential privacy (DP) that remedies certain weaknesses of standard DP (including tightness under algorithmic composition). A challenge when deploying differentially private mechanisms is that DP is hard to validate, especially in the black-box setting. This has led to numerous empirical methods for auditing standard DP, while $f$-DP remains less explored. We introduce new black-box methods for $f$-DP that, unlike existing approaches for this privacy notion, do not require prior knowledge of the investigated algorithm. Our procedure yields a complete estimate of the $f$-DP trade-off curve, with theoretical guarantees of convergence. Additionally, we propose an efficient auditing method that empirically detects $f$-DP violations with statistical certainty, merging 
techniques from non-parametric estimation and optimal classification theory. Through experiments on a range of DP mechanisms, we demonstrate the effectiveness of our estimation and auditing procedures.

\end{abstract}

\section{Introduction}

%\todo{DP and f-DP; positives of f-DP}

Differential privacy (DP) \cite{Dwork2006} is a widely-used framework to quantify and limit information leakage of a data-release {\em mechanism} $\Mech$ via privacy parameters $\varepsilon > 0$ and $\delta \in [0,1]$. Mechanisms that are differentially private for a suitable choice of $\varepsilon$ and $\delta$ mask the contribution of individuals to their output. As a consequence, DP has been adopted by companies and public institutions to ensure user privacy \cite{Erlingsson2014,Holohan2019,Abowd2018}.

Over the years,  variants and relaxations of DP have been proposed to address specific needs and challenges. Of these, the recent notion of $f$-DP~\cite{Dong2022} is one of the most notable, due to its attractive properties such as a tight composition theorem, and applications such as providing an improved, simpler analysis of privatized stochastic gradient descent (Noisy or DP-SGD), the most prominent privacy-preserving algorithm in machine learning. $f$-DP is grounded on the hypothesis testing interpretation of DP \footnote{\noindent For a rigorous introduction to hypothesis testing and $f$-DP we refer to Section \ref{sec:prelim}.} and describes the privacy of mechanism $\Mech$ in terms of a real-valued function $f$ on the unit interval $[0,1]$. Several mechanisms~\cite{Dong2022} have been shown to achieve $f$-DP. However, the process of designing privacy-preserving mechanisms and turning them into real-world implementations is susceptible to errors that can lead to so-called `privacy violations' \cite{Lyu2017,mattjj2023fixprng,Mironov2012}. Worse, checking such claims may be difficult, as some implementations may only allow for limited, {\em black-box} access.  %This is exemplified in the case of DP mechanisms (whose privacy can also be represented by $f$-DP), when the achieved privacy parameter is worse than what is claimed due to e.g., a buggy implementation~\cite{lyu2017understanding}. 
This problem has motivated the proposal of methods that assess the privacy of a mechanism $M$ with only black-box access.

Within the plethora of works on privacy validation, most approaches study mechanisms through the lens of standard DP \cite{StatDP, DP-Finder, CheckDP, DP-Sniper, Dette2022, Liu2019, Lu2024, Chadha2023, Lokna2023, Kifer2017, Kifer2019, Tschantz2011, Barthe2012,Barthe2014, Barthe2015, Barthe2016, Barthe2016_B}. In contrast, comparatively few methods examine $f$-DP~\cite{Nasr2023,Annamalai2024, Annamalai2024_B, Annamalai2024_C, Mahloujifar2024,Koskela2024}. Moreover, many of the procedures that feature $f$-DP are tailored to audit the privacy claims of a specific algorithm, namely DP-SGD \cite{Nasr2023,Annamalai2024,Annamalai2024_B}. Our goal is to devise methods that are not specific to a single mechanism, but are instead applicable to a broad class of algorithms, while only requiring black-box access. We formulate our two objectives: \\[-3.5ex]%that our methods for investigating $f$-DP should tackle:
\begin{itemize}
    \item \textbf{Estimation:} Given black-box access to a mechanism $M$, estimate its true privacy parameter (i.e., the function $f$ in $f$-DP).\\[-3.5ex]
    \item \textbf{Auditing:} Given black-box access to a mechanism $M$ and a target privacy  $f$, check whether $M$ violates the targeted privacy level (i.e., given $f$, does $M$ satisfy $f$-DP?).\\[-3.5ex]
\end{itemize}
Estimation is useful  when we do not have an initial conjecture regarding $M$'s privacy. It can thus be used as, e.g.,  preliminary exploration into the privacy of $M$. Auditing, on the other hand, can check whether an algorithm meets a specific target privacy  $f$ and is therefore designed to detect flaws or overly optimistic privacy guarantees. \footnote{\B{For a detailed discussion on the advantages of auditing $f$-DP, we refer to Section 4 in \cite{Nasr2023}.}}

\subsubsection*{Contributions}

We construct a `general-purpose' $f$-DP {\em estimator} and {\em auditor} for both objectives, where:\\[-3.5ex]
%With this work, we begin to fill the gap in the literature by constructing the first  `general-purpose' $f$-DP {\em estimator} and {\em auditor} of any mechanism, given only black-box access, where: %Specifically, we develop the first general-purpose black-box \(f\)-DP estimator and auditor, where: %\textcolor{orange}{
\begin{itemize}
        \item[(1)] The estimator approximates the entire true \(f\)-DP curve of a given mechanism $M$. \\[-3.5ex]
        \item[(2)] Given a target \(f\)-DP curve, the auditor statistically detects whether $M$ violates $f$-DP.  The auditor involves a tuneable confidence parameter to control the false detection rate. \\[-3.5ex]
\end{itemize}
%Experimentally, our black-box estimator achieves high accuracy for both estimation and auditing. \todo{TODO @Yun: Make more high-level and merge into Contributions}
A methodological advantage of our methods is that they come with strong mathematical performance guarantees (both for the estimator and the auditor). Such guarantees seem warranted when making claims about the performance and correctness of a mechanism. A practical advantage of our methods is their efficiency: Our experiments (Sec.~\ref{sec6}) demonstrate high accuracy at typical runtimes of 1-2 minutes on a standard personal device.

$ $\\[-1.5ex]
\noindent \textbf{Paper Organization} 
Preliminaries are introduced in Sec.~\ref{sec:prelim}. In Sec.~\ref{sec:overview_techniques} we give an overview of techniques. We propose our $f$-DP curve estimator in Sec.~\ref{sec:4} and auditor in Sec.~\ref{sec:goal2}. We evaluate the effectiveness of both estimator and auditor %validate and benchmark our constructions%
in Sec.~\ref{sec6} using various mechanisms from the DP literature, including DP-SGD. We delve into more detail on related work in Sec.~\ref{sec:relatedwork} and conclude in Sec.~\ref{sec:summary_discussion}. A table of notations, proofs
 and technical details can be found in the Appendix.

\section{Preliminaries}\label{sec:prelim}

In this section, we provide details on hypothesis testing, differential privacy and tools from statistics and machine learning that our methods rely on.

\subsection{Hypothesis testing} \label{sec:hyp}

We provide a brief introduction into the key concepts of hypothesis testing. We confine ourselves to the special case of sample size $1$, most relevant to $f$-DP. 
For a general introduction we refer to \cite{Bickel2001}.
Consider two probability distributions $P, Q$ on the Euclidean space $\mathbb{R}^d$ and a random variable $X$. It is unknown from which of the two distributions $X$ is drawn and the task is to decide between the two competing hypotheses
\begin{align} \label{e:hyp2}
H_0: X \sim P\quad \textnormal{vs.} \quad H_1: X \sim Q.
\end{align}
The problem is similar to a classification task (see Section \ref{sec:class} below). The key difference to classification is that\B{,} in hypothesis testing, there exists a default belief $H_0$ that is preferred over $H_1$. The user switches from $H_0$ to $H_1$ only if the data ($X$) suggests it strongly enough. In this context, a hypothesis test is a binary, potentially randomized function $g: \mathbb{R}^d \to \{0,1\}$, where $g(X)=0$ implies to stay with $H_0$, while $g(X)=1$ implies that the user should switch to $H_1$ ($H_0$ is "rejected"). Just as in classification, the decision to reject/fail to reject can be erroneous and the error rates of these decisions are called $\alpha$, the "type-I error", and $\beta$, the "type-II error". Their formal definitions are
\[
\alpha^{(g)}:= \Pr_{X \sim P}[g(X)= 1], \quad \beta^{(g)}:= \Pr_{X \sim Q}[g(X)= 0].
\]
One test $g$ is better than another $g'$, if simultaneously
\[
\alpha^{(g)} \le \alpha^{(g')} \quad \textnormal{and} \quad \beta^{(g)} \le \beta^{(g')}.
\]
This comparison of statistical tests naturally leads to the issue of optimal tests, 
%, considered in the next section.
%\subsection{Optimal tests} \label{sec:optimal}
% We consider the setup from the previous section and 
and we define the optimal level-$\alpha$-test as the argmin of
 \[
 \{\beta^{(g)}: g \,\,\textnormal{is a  test}\,\, \textnormal{with } \,\, \alpha^{(g)}\le \alpha\}.
 \]
 The minimum is achieved and the corresponding optimal test is provided by the {\em likelihood ratio (LR) test} in the Neyman-Pearson lemma, a fundamental result in statistics. In the following, we assume the two probability measures $P,Q$ in hypotheses \eqref{e:hyp2} have some probability densities $p,q$.

\begin{theo}[Neyman-Pearson Lemma~\cite{neyman1933ix}]
\label{theo: NP lemma}
For any $\alpha \in [0,1]$, the smallest  type-II error $\beta(\alpha)$ among all level-$\alpha$-tests is achieved by the \textbf{{\em likelihood ratio (LR) test}}, which is characterized by two constants $\eta \ge 0$ and $\lambda\in [0,1]$, and has the following rejection rule:
\begin{itemize}
    \item[1)] Reject $H_0$ if $q(X)/p(X) >\eta$.
    \item[2)] If $q(X)/p(X)=\eta$, flip an unfair coin with probability $\lambda$ of heads. If the outcome is heads, reject $H_0$.
\end{itemize} 
The constants $(\eta, \lambda)$ are chosen such that the type-I error is exactly $\alpha$.
%\begin{align*}
%    \beta(\alpha) = 1 - \int q \cdot \mathbb{I}\{p / q  \leq \eta\}
%\end{align*}
%where $\eta \in [0, +\infty]$ such that $\alpha = \int p \cdot \mathbb{I}\{p / q  \leq \eta\}$.
\end{theo}

\noindent \textbf{Notations.}  Neyman-Pearson  motivates the use of the following notations. For any type-I error $\alpha$ there is a corresponding (optimal) $\beta$ implied by the lemma. These constants are achieved by a pair $(\eta, \lambda)$ and we can thus write $\alpha(\eta, \lambda), \beta(\eta, \lambda)$ for them. When we are only interested in the result of the non-randomized test with $\lambda=0$, we will just write   $\alpha(\eta), \beta(\eta)$.

\subsection{($f$-)Differential Privacy (DP)} \label{sec:fdp}
%\textcolor{blue}{
DP requires that the output of mechanism $\Mech$ is similar on all {\em neighboring} datasets $\DB, \DB'$ that differ in exactly one data point (we also call $\DB, \DB'$ {\em neighbors}). \B{We use the \say{edit} notion of neighborhood, i.e., $D'$ can be obtained from $D$ by editing one of its entries, rather than deleting it.} 
%that person's data. 
%We call datasets $\DB,\DB'$ neighbors/neighboring if they differ in exactly one data point.

\begin{definition}[DP~\cite{Dwork2006}]
A mechanism $\Mech$ is $(\varepsilon,\delta)$-DP if for all neighboring datasets $\DB, \DB'$ and any set $\cS$, 
  \begin{equation*}    
\Pr(\Mech(\DB) \in \cS) \leq e^\varepsilon \, \Pr(\Mech(\DB') \in \cS) + \delta~.
\end{equation*}
\end{definition}

Informally, if $\Mech$ is $(\varepsilon, \delta)$-DP, an adversary's ability to decide whether $\Mech$ was run on $\DB$ or $\DB'$ is bounded by $\delta$ and $e^{\varepsilon}$. 
For instance, any statistical level-$\alpha$-test $g$ that aims at deciding this problem must incur a type-II-error of at least $1 - e^{\varepsilon} \, \alpha - \delta$. The notion of $f$-DP was introduced to make this observation more rigorous.
%After introducing the basics of (optimal) tests, we can now formally describe $f$-DP. 
Given a pair of neighbors $\DB$ and $\DB'$ and a sample $\rX$, consider the hypotheses:
% \begin{align*}
%     H_0: \rX \sim P,& \text{ where } P = \Mech(\DB)\\   
%     H_1: \rX \sim Q,& \text{ where } Q =\Mech(\DB').
% \end{align*}
\begin{align*}
    &H_0: \rX \sim P\quad   
    &H_1: \rX \sim Q,
\end{align*}
where $\Mech(\DB)$ and $\Mech(\DB')$ are distributed to $P, Q$, respectively. Roughly speaking, good privacy requires these two hypotheses to be hard to distinguish. That is, for any hypothesis test with type-I error $\alpha$, its type-II error $\beta$ should be large. This is captured by the trade-off function $T$ between $P$ and $Q$. 
\begin{definition}[Trade-off function~\cite{Dong2022}]
    For any two 
distributions $P$ and $Q$ on the same space, the 
trade-off function $T$ is: $$T(\alpha) := \inf \{\beta^{(g)}: g \,\,\textnormal{ test }\,\, \textnormal{ with } \,\, \alpha^{(g)}\le \alpha\}$$
\end{definition}

$\Mech$ is $f$-DP if its privacy is at least as good (its trade-off function is at least as large) as $f$, when considering all neighboring datasets.
\begin{definition}[$f$-DP~\cite{Dong2022}]
    A mechanism $\Mech$ is $f$-DP if for all neighboring datasets $\DB, \DB'$ it holds that
    $T \geq f$. Here, $T$ is the trade-off function implied by $\Mech(\DB) \sim P$ and $\Mech(\DB') \sim Q$. 
\end{definition}

%As discussed in the introduction, in line with previous black-box methods, we will confine our discussion on $f$-DP to a {\em single} pair of neighboring datasets $D,D'$.

We say $f$ is the {\em optimal/true} privacy parameter if it is the largest $f$ such that $\Mech$ is $f$-DP---such optimality is necessary to define for meaningful $f$-DP estimation, as any $\Mech$ is trivially $f$-DP for $f = 0$ (since the type-II error in hypothesis testing is always $\geq 0$).

\subsection{Kernel Density Estimation} \label{sec:kde}
Kernel density estimation (KDE) is a well-studied tool from non-parametric statistics to approximate an unknown density $p$ by an estimator $\hat{p}$.
% and $\hat{q}$. 
More concretely, in the presence of sample data $\rX_1, \dots, \rX_n \sim p$ with $\rX_i \in \mathbb{R}^d$, the KDE for $p$ is given by 
\begin{align*}
    \hat{p}(t) := \frac{1}{n b^d} \sum_{i=1}^n K\Big( \frac{t-\rX_i}{b}\Big).
\end{align*}
One can think of the KDE as a smoothed histogram where the bandwidth parameter $b >0$ corresponds to the bin size for histograms. The kernel function $K$ indicates the weight we assign each observation $X_i$ and is oftentimes taken to be the Gaussian kernel with
\begin{align*}
   K(t) = \frac{1}{(2 \pi)^{d/2}} \; \exp \left( -\frac{\vert t \vert^2}{2} \right).
\end{align*}
The appropriate choice of $b$ and $K$ can ensure the uniform convergence of $\hat p$ to the true, underlying density $p$ (as in Assumption \ref{ass2}). Higher smoothness of the density $p$ is generally associated with faster convergence rates and we refer to \cite{Jiang2017}  and \cite{Scott2015} for a rigorous definition of KDE and associated convergence results.

%\todo{TODO: @Onder Merge with Overview of Techniques}
%\textbf{Density estimation in DP} Density estimation in general and KDE in particular is 
%an important tool in the black box assessment of DP. For some examples, we refer to \cite{??}, \cite{??} and \cite{..}\todo{please add the references you had in mind}. The reason is that DP can typically be expressed as some transformation of the density ratio $p/q$ - this is true for standard DP (a supremum), Rényi DP (an integral) and, as we exploit in this paper, $f$-DP via the Neyman-Pearson test.
%, i.e., 
%\begin{align*}
%    \mathbb{P} \left(\sup_{t} |\hat{p}(t)-p(t)|>a_n \right)=o(1)
%\end{align*}
%for a sequence $(a_n)_n$ with $a_n >0$ that converges to $0$. 
%We point the interested reader to \cite{Jiang2017}  and \cite{Scott2015} for a suitable choice of $K$ and $b$ and the construction of KDEs in general.

\subsection{Machine Learning Classifiers} \label{sec:class}
\textbf{Binary classifiers} are the final addition to our technical toolbox. We begin with some notations: We denote a generic classifier on the Euclidean space $\mathbb{R}^d$ by $\phi$. Formally, a {\em classifier} is not that different from a statistical test: It is a (potentially random) binary function $\phi: \mathbb{R}^d \to \{0,1\}$. However, its interpretation is different from hypothesis testing, because we do not have a default belief in a label $0$ or $1$. 
%both binary labels $0$ and $1$ are equal and we do not have a default belief. 
Let us now consider a probability distribution $\mathcal{P}$ on the combined space of inputs and outputs $\mathbb{R}^d \times \{0,1\}$. A classification error has occurred for a pair $(x,y) \in \mathbb{R}^d \times \{0,1\}$, whenever $\phi(x) \neq y$. If $(x,y)$ are randomly drawn from $\mathcal{P}$, we define the risk of the classifier $\phi$ w.r.t. to $\mathcal{P}$ as
    \begin{align*}
        R(\B{\phi}) = \Pr\limits_{(x,y)\sim \Dist}[\phi(x) \neq y].
    \end{align*}
\smallskip
\noindent \textbf{Bayes Classification Problem.} The Bayes classification problem refers to a setup to generate the distribution $\mathcal{P}$, where a Bernoulli random variable $\rY \in \{0,1\}$ is drawn and then a second variable $X$ with
\begin{align*}
    (\rX|\rY=0) \sim P, \qquad (\rX|\rY=1) \sim Q.
\end{align*}
% and the task is to minimize the risk. 
In our work, we specifically consider the case where $\rY$ is drawn from a fair coin flip (i.e., $\Pr[\rY=0] = \Pr[\rY=1] = \frac{1}{2}$), and we denote this setup by $\bbcP{P}{Q}$.

\smallskip
% \noindent \textbf{Optimal classifiers} The \textit{Bayes optimal classifier} $\phi^*$ is the one that has minimal risk in the Bayes classification problem. In practice, $\phi^*$ is typically unknown, because it depends on the (unknown) $P,Q$. However, it is possible to approximate $\phi^*$ using the feasible nearest neighborhood classifier. More concretely, one train a kNN classifier $\kNNclassifier{n}$ with $n$ samples by simply storing $n$ independent samples from distribution $\mathcal{P}$. To predict the label of an observation $o \in \cO,$ $\kNNclassifier{n}$ returns the label taking a majority vote of the class labels of its $k$ nearest neighbors (Euclidean distance in our context) in the stored training points. 

\noindent \textbf{Bayes (Optimal) classifiers.} $\phi^*$ minimizes the risk in the Bayes classification problem. However, $\phi^*$ is usually unknown in practice because it depends on the (unknown) $P$ and $Q$. To approximate $\phi^*$, one can use a feasible nearest-neighbor classifier~\cite{altman1992introduction}. Specifically, a $k$-nearest neighbors ($k$-NN) classifier, denoted as $\kNNclassifier{n}$, assigns a label to an observation $o \in \cO$ by identifying its $k$ closest neighbors\footnote{In our context, closeness is measured using Euclidean distance} from the size $n$ training set. The label is then determined by a majority vote among these $k$ neighbors.

% Specifically, a $k$-nearest neighbors ($k$-NN) classifier, denoted as $\kNNclassifier{n}$, can be trained using $n$ independent samples drawn from the distribution $\mathcal{P}$. To predict the label of an observation $o \in \cO$, $\kNNclassifier{n}$ returns a label based on a majority vote of the class labels of the $k$ nearest neighbors (measured using Euclidean distance in our context) among the stored training samples.

% \textcolor{orange}{\textbf{Yu:} Write (very very shortly) what the kNN classifier is - it is not defined. Do not use additional notations, besides $\kNNclassifier{n}$. Then state the below theorem with the new, \underline{pruned notations} from the above paragraph! Adjust notations in 5.1 accoridngly. We need to reduce notations!!}

The following convergence result for $k$-NN gauges how close the true risk $R(\kNNclassifier{n})$ of the $k$-NN classifier $\kNNclassifier{n}$ is to the risk of the optimal classifier, $R(\phi^{*})$.

\begin{theo} [\textbf{Convergence of $k$-NN Classifier}~\cite{Books:DGL96}]
\label{thm:covergence of kNN}
Let $\Dist$ be a joint distribution with  support  $\cO \times \mathcal{Y}.$ If the conditional distribution $\Dist|\mathcal{Y}$ has a density, $\cO \subseteq \mathbb{R}^d,$ and $k = \sqrt{n},$ then for every $\epsilon >0$ there is an $n_0$ such that for $n>n_0,$ 
{\small
    \begin{align*}
        \Pr[|R(\kNNclassifier{n}) - R(\phi^{*})| > \epsilon] \leq 2e^{-n\epsilon^2/(72c^2_d)},
    \end{align*}
}
where $c_d$\footnote{By Lemma 5.5 of~\cite{Books:DGL96}, $c_d$ satisfies $c_d \leq (1+{2}/{\sqrt{2-\sqrt{3}}})^d - 1$.} is the minimal number of cones centered at the origin of angle $\pi/6$ that cover $\mathbb{R}^d.$ Note that if the number of dimensions $d$ is constant, then $c_d$ is also a constant.
\end{theo}

\section{Overview of Techniques}\label{sec:overview_techniques}
Our goal is to provide an estimation and auditing procedure for the optimal privacy curve $f$ of a mechanism $\Mech$. This task can be broken down into two parts: (1) Selecting datasets $\DB,\DB'$ that cause the largest difference in $\Mech$'s output distributions and (2) Developing an estimator/auditor for the trade-off curve given that choice of $\DB, \DB'$. %for the trade-off function $T = T(M(D),M(D'))$. 
In line with previous works on black-box estimation/auditing, we focus on task (2). The selection of $\DB,\DB'$ has been studied in the black-box setting and can typically be guided by simple heuristics \cite{StatDP, DP-Sniper, Lokna2023}.% We will therefore focus on part (2) and rely on output samples gathered from $M$ on $D,D'$ for this.

Our proposed estimator of a trade-off curve relies on KDEs. Density estimation in general and KDE in particular is 
an important tool in the black box assessment of DP. For some examples, we refer to \cite{Liu2019}, \cite{Dette2022} and \cite{Kutta2024}. The reason is that DP can typically be expressed as some transformation of the density ratio $p/q$ -- this is true for standard DP (a supremum), Rényi DP (an integral) and, as we exploit in this paper, $f$-DP via the Neyman-Pearson test. A feature of our new approach is that we do not simply plug in our estimators in the definition of $f$-DP, but rather use them to make a novel, approximately optimal test. This test is not only easier to analyze than the standard likelihood ratio (LR) test but also retains similar properties (see the next section for details).

Our second goal (Sec.~\ref{sec:audit}) is to audit whether a mechanism~$\Mech$ satisfies a claimed trade-off $f$, given datasets $\DB$ and $\DB'$. At a high level, we address this task by identifying and studying the \emph{most vulnerable point} on the trade-off curve $T$ of $\Mech$ --- the point most likely to violate $f$-DP. We begin by using our $f$-DP estimator to compute a value $\eta$ (from the Neyman-Pearson framework in Sec.~\ref{sec:hyp}), which defines a point $\bigl(\alpha(\eta), \beta(\eta)\bigr)$ on the true privacy curve $T$ of the mechanism~$\Mech$. $\eta$ is chosen such that $\bigl(\alpha(\eta), \beta(\eta)\bigr)$ has the largest distance from the claimed trade-off curve~$f$ asymptotically, which we prove in Prop.~\ref{prop1}.
%We prove that it converges to the point in the curve with the largest difference when subtracting from the claimed trade-off curve $f$ (Prop.~\ref{prop1}). Intuitively, this is the {\em most vulnerable point} that is  most likely to violate $f$-DP. 
Next, by extending a technique proposed in~\cite{Lu2024}, we express $\bigl(\alpha(\eta), \beta(\eta)\bigr)$ in terms of the Bayes risk of a carefully constructed Bayesian classification problem, and approximate that Bayes risk using a feasible binary classifier (e.g., $k$-nearest neighbors). By deploying the $k$-NN classifier, we obtain a confidence interval that contains our vulnerable point $(\alpha,\beta)$ with high probability \B{(in the Appendix, we provide a brief explanation for choosing confidence intervals over the credible intervals used in other works \cite{zanella2023bayesian, Nasr2023})}.
% We then examine this $\eta$ value of interest further, by extending a technique used in~\cite{Lu2024}. Specifically, our auditor employs a Bayes (optimal) binary classifier (which can be approximated by, e.g., $k$-nearest neighbor) to approximate the corresponding true $(\alpha, \beta)$ value. 
Finally, our auditor decides whether to reject (or fail to reject) the claimed $f$ curve by checking whether the corresponding point $(\alpha, \beta')$ on $f$ with $f(\alpha) = \beta'$ is contained in this interval or not.
% We then compare our estimate of $(\alpha, \beta)$ to the corresponding point $(\alpha, \beta')$ on the given claimed trade-off curve $f$---note that we will compare using the same $\alpha$. %on the $x$-axis.  In other words, we check whether $f$-DP is satisfied, by checking this 'worst' point. 
Leveraging the convergence properties of $k$-NN, our auditor provides a provable and tuneable confidence region that depends on sample size. We also note that the connection between Bayes classifiers and $f$-DP that underpins our auditor may be of independent interest, as it offers a new interpretation of $f$-DP by framing it in terms of Bayesian classification problems.

\section{Goal 1: $f$-DP Estimation} \label{sec:4}

In this section, we develop a new method for the approximation of the entire optimal trade-off curve. The trade-off curve results from a study of the Neyman-Pearson test, where any type-I error $\alpha$ is associated with the smallest possible type-II error $\beta$ (see Section \ref{sec:hyp} for details). Understood as a function in $\alpha$, we denote the type-II error by $T:[0,1] \to [0,1]$ and call it a trade-off curve. We note that any trade-off curve is continuous, non-increasing and convex (see \cite{Dong2022}).

\subsection{Estimation of the $f$-DP curve}
Our approach is based on the perturbed likelihood ratio (LR) test which mimics the properties of the optimal Neyman-Pearson test, but requires less knowledge about the distributions involved. In the following, we denote by $P,Q$ the output distributions of $M(D), M(D')$ respectively. The corresponding probability densities are denoted by $p,q$.\\
\textbf{The perturbed LR test.} The optimal test for the hypotheses pair 
\[
H_0: X \sim p\quad \textnormal{vs.} \quad H_1: X \sim q
\]
is the Neyman-Pearson test described in Section \ref{sec:hyp}.  It is also called a \textit{likelihood ratio} (LR) test, because it rejects $H_0$ if the density ratio satisfies $q(X)/p(X)>\eta$ for some threshold $\eta$. If $q(X)/p(X)=\eta$ the test rejects randomly with probability $\lambda.$
In a black-box scenario\B{,} this process is difficult to mimic, even if two good estimators, say $\hat p, \hat q$ of $p,q$ are available. Even if $\hat p \approx p$ and $\hat q \approx q$, it will usually be the case that
\[
q(x)/p(x) = \eta  \quad \textnormal{does not imply} \quad \hat q/ \hat p =\eta
\]
(it may hold that $\hat p/ \hat q \approx \eta $, but typically not exact equality). In principle, one could cope with this problem by modifying the condition $\hat q/ \hat p =\eta$ to $\approx \eta$ to mimic the optimal test. Yet, the implementation of this approach turns out to be difficult. In particular, it would involve two tuneable arguments $(\eta, \lambda)$, as well as further parameters (to specify "$\approx$"), making approximations costly and unstable. A simpler and more robust approach is to focus on a different test rather than the optimal one - a test that is close to optimal but does not require the knowledge of when $q/p$ is constant. For this purpose, we introduce here the novel \textit{perturbed LR test} \B{(PLRT)}. 
We define it as follows: Let $U \in [-1/2, 1/2]$ be uniformly distributed and $h>0$ a (small) number. Then we make the decision
\begin{align} \label{e:PLR}
"\textnormal{reject $\,\,H_0\,\,$ if } \quad q(X)/p(X)>\eta + h U".
\end{align}
Just as the Neyman-Pearson test, the perturbed LR test is randomized. Instead of flipping a coin when $q/p=\eta$, the threshold $\eta$ is perturbed with a small, random noise term. Obviously the perturbed LR test does not require knowledge of the level sets $\{q/p=\eta\}$, making it more practical for our purposes.
To formulate a theoretical result for this test, we impose two natural assumptions.
\begin{ass} \label{ass1} $ $
    \begin{itemize}
        \item[i)] The densities $p,q$ are continuous.
        \item[ii)] There exists only a finite number of values $\eta \ge 0$ where the set $\{q/p=\eta\}$ has positive mass.
       % \item[iii)] The optimal trade-off curve $T$ implied by $p,q$ is continuous.
    \end{itemize}
\end{ass}
The second assumption is met for all density models that the authors are aware of and in particular for all mechanisms commonly used in DP. %The third assumption facilitates the presentation of our below results, as is guarantees convergence of the estimated trade-off curve to $T$ uniformly; i.e. simultaneously for all arguments $\alpha$. This condition can be dropped at the expense of a more technical formulation, where uniform convergence is replaced by convergence in the Skorohod metric (see \cite{billingsley:1999}). We do not pursue such generalizations here, since they are not practically relevant.\\
Let us denote the $f$-DP curve of the perturbed LR test by $T_h$. The next Lemma shows that for small values of $h$ the perturbed LR test performs as the optimal LR test.
\begin{lem} \label{lem1}
Under Assumption \ref{ass1} it holds that
\[
\lim_{h \downarrow 0} \sup_{\alpha \in [0,1]}|T(\alpha)-T_h(\alpha)|=0.
\]
\end{lem}
 \textbf{Approximating $T_h$.} The Lemma shows that to create an estimator of the optimal trade-off curve $T$, it is sufficient to approximate the curve $T_h$ of the perturbed LR test for some small $h$. This is an easier task, since we do not need to know the level sets $\{q/p=\eta\}$ for all $\eta$. Indeed, suppose we have two estimators $\hat p, \hat q$\B{. Then} we can run a perturbed LR test with them, just as in equation \eqref{e:PLR}. A short theoretical derivation (found in the appendix) then shows that running the perturbed LR test for $\hat p, \hat q$ and some threshold $\eta$, yields the following type-I and type-II errors:
\begin{align}
\hat \alpha_h(\eta) := &\quad\,\,\,\,\, \; \, \int_{x \in [-h/2,h/2]} \frac{1}{h}\int_{\hat q /\hat p  > \eta +x} \hat p , \\\hat \beta_h(\eta) := & \; \, 1-\int_{x \in [-h/2,h/2]} \frac{1}{h}  \int_{\hat q /\hat p  > \eta +x} \hat q .
\end{align}
The entire trade-off-curve for the perturbed LR test with $(\hat p, \hat q)$ is then given by $\hat T_h$ with
\begin{align} \label{e:def:Th}
\hat T_h(\alpha) = \hat \beta_h(\eta) \quad \Leftrightarrow \quad \alpha = \hat \alpha_h(\eta).
\end{align}
For the curve estimate $\hat T_h$ to be close to $T_h$ (and thus $T$), the involved density estimators need to be adequately precise. We hence impose the following regularity condition on them. In the condition, $n$ is the sample size used to create the estimators.
\begin{ass} \label{ass2}
    The density estimators $\hat p, \hat q$ are themselves continuous probability densities that decay to $0$ at $\pm \infty$ 
    (see eq. \eqref{e:decay} for a precise definition)
    . For a null-sequence of non-negative numbers $(a_n)_{n \in \mathbb{N}}$ they satisfy
    \begin{align*}
     & \Pr[\sup_{x } |\hat p(x)-p(x)|>a_n]=o(1)\\
    and \quad &\Pr[\sup_{x } |\hat q(x)-q(x)|>a_n]=o(1). 
    \end{align*} 
\end{ass}
The above assumption is in particular satisfied by KDE (see Section \ref{sec:kde}), where the convergence speed $a_n$ depends on the smoothness of the underlying densities. However, in principle other estimation techniques than KDE could be used, as long as they produce continuous estimators. The next result formally proves the consistency of $\hat T_h$. The notation of "$o_P(1)$" refers to a sequence of random variables converging to $0$ in probability.

\begin{theo} \label{theo:1}
    Suppose that Assumptions \ref{ass1} and \ref{ass2} hold, and that $h=h_n$ is a positive number depending on $n$ with $h_n \to 0$ and $h_n/a_n \to \infty$. Then, as $n \to \infty$ it follows that
    \[
    \sup_{\alpha \in [0,1]}|\hat T_h(\alpha)-T(\alpha)|=o_P(1).
    \]
\end{theo}
 $ $\\[-2ex]
The above result proves that simultaneously for all $\alpha$, the curve $\hat T_h$ approximates the optimal trade-off function $T$. Thus, we have achieved the first goal of this work. The (very favorable) empirical properties of $\hat T_h$ will be studied in Section \ref{sec6}. We have also incorporated Algorithm \ref{alg:pointwise_KDE_estimator} for an overview of the procedure in the appendix. 
% \begin{algorithm}[h]
% \footnotesize
% \algorithmicrequire \; \parbox[t]{\dimexpr0.9\linewidth-\algorithmicindent}{Black-box access to $M$; Threshold vector $\eta > 0$; Sample size $n$.}\\[0.1cm]
% \algorithmicensure \, An estimate $(\hat{\alpha}(\eta), \hat{\beta}(\eta))$ of $(\alpha(\eta), \beta(\eta))$ for tuple $(P, Q)$.
% \begin{algorithmic}[1]
%     \State Set parameter $h$. 
%     \State Set the density estimation algorithm $\cA$. By default, use the KDE algorithm.
%     \Function{\textnormal{PTLR Estimatior} $\ptlr{h}{\cA}(M, \eta,n)$}{}
%     \State Compute the estimated densities $\hat{p}, \hat{q}$ based on outputs of $M$ by running $\cA$ with a sample size of $n$.
%     \State Compute $\hat{\alpha}(\eta_i) \leftarrow \int_{x \in [-h/2,h/2]} \frac{1}{h}\int_{\hat q /\hat p  > \eta_i +x} \hat p$ for all $\eta_i\in \eta$
%     \State Compute $\hat{\beta}(\eta_i) \leftarrow \int_{x \in [-h/2,h/2]} \frac{1}{h}  \int_{\hat q /\hat p  > \eta_i +x} \hat q$ for all $\eta_i\in \eta$ 
%     \State Return vector $(\hat{\alpha}(\eta), \hat{\beta}(\eta))$
%     \EndFunction
% \end{algorithmic}
% \caption{PTLR: A Perturbed Likelihood Ratio Test Algorithm for $f$-DP Estimation}
% \label{alg:pointwise_KDE_estimator}
% \end{algorithm}

\subsection{Finding maximum vulnerabilities} We conclude this section by some preparations for the second goal - auditing $f$-DP. The precise problem of auditing is described in Section \ref{sec:audit}. Here, we only mention that the task of auditing is to check (in some sense) whether $f$-DP holds for a claimed trade-off curve, say $f=T^{(0)}$.
As an initial step, to check $T^{(0)}$-DP, we create the estimator $\hat T_h$ for the optimal curve $T$. If $T^{(0)}$-DP holds, this means that
\begin{align} \label{e:H0fDP}
     T(\alpha)\ge  T^{(0)}(\alpha)\quad \forall \alpha \in [0,1].
\end{align}
A priori, we cannot say whether this is true or not. However, by comparing our estimator $\hat T_h$ with  $T^{(0)}$ we can gather some evidence. For example, if $\hat T_h(\alpha)$ is much smaller than $ T^{(0)}(\alpha)$ for some $\alpha$, it then seems that the claim in \eqref{e:H0fDP} is probably false. We will develop a rigorous criterion for what "much smaller" means in the next section. For now, we will confine ourselves to identifying a point where privacy seems most likely to be broken. We therefore define 
\begin{align} \label{e:def:eta}
\hat \eta^* \in \textnormal{argmax} \big\{T^{(0)}(\hat \alpha_h(\eta))-\hat T_h(\hat \alpha_h(\eta)): \eta\ge 0\big\} 
\end{align}
and the next result shows that the discrepancy between $T^{(0)}$ and $T$ is indeed maximized in $\hat \eta^* $ for large $n$.\\
%\todo{Question: Does this not also have the issue to directly using the estimator, that we don't know the rate of convergence to the true most-vulnerable point?} \textcolor{orange}{Yes.}\todo{To discuss: how to describe our auditor's accuracy (esp in the intro, given that we try to differentiate our estimator/auditor results)}
\begin{prop} \label{prop1}
    Suppose that the assumptions of Theorem \ref{theo:1} hold. Then, it follows that 
    \begin{align*}
    &T^{(0)}(\hat \alpha_h(\hat \eta^*)) - T(\hat \alpha_h(\hat \eta^*)) \\
    =&\sup_{\alpha \in [0,1]}\big[ T^{(0)}(\alpha)-T(\alpha)\big]+o_P(1).
    \end{align*}
\end{prop}
The threshold $\hat \eta^*$ demarcates the greatest weakness of the $T^{(0)}$-privacy claim and it is therefore ideally suited as a starting point for our auditing approach in Section \ref{sec:audit}.

%\begin{itemize}
%    \item Theorem/lemmas for Goal 1 using KDE
%    \item Lemma where KDE can also give you the $\eta$ with the biggest violation and point to next section on how to use that for Goal 2.
%\end{itemize}
\section{Goal 2: Auditing $f$-DP} \label{sec:goal2}

In this section, we develop methods for uncertainty quantification in our assessment of $T$. We begin with Section \ref{sec:conf}, where we derive (two dimensional) confidence regions for a pair of type-I and type-II errors. Our approach relies on the approximation of Bayes optimal classifiers using the $k$-nearest neighbor ($k$-NN) method. The resulting confidence regions are used in Section \ref{sec:audit} as a subroutine of a general-purpose $f$-DP auditor that combines the estimators from KDE and the confidence regions from $k$-NN.

%\begin{itemize}
%    \item Using kNN to quantify  the maximum violation at one point (given an $\eta$). Emphasize that kNN gives point-wise confidence level that depends on the number of samples.
%    \item Use the worst-case $\eta$ found by KDE from previous section to detect violation of $f$-DP given some $f$.
%\end{itemize}

\subsection{Pointwise confidence regions} \label{sec:conf}

In this section, we introduce the BayBox estimator, an algorithm designed to provide point-wise estimates of the trade-off curve $T$ with theoretical guarantees. Specifically, for a given threshold $\eta > 0$, the BayBox estimator outputs an estimate of the trade-off point $(\alpha(\eta), \beta(\eta))$. This estimate is guaranteed to be within a small additive error of the true trade-off point, with high probability.

The BayBox estimator is backed up by the observation that the quantity $\alpha(\eta)$ (also $\beta(\eta)$) can be expressed as the Bayes risk of a carefully constructed Bayesian classification problem. For instance, to compute $\alpha(\eta)$ when $\eta \geq 1$, a theoretical derivation (provided in the appendix) shows that this computation is equivalent to computing the Bayes risk for the Bayesian classification problem $\bbcP{\MixtureD{P}{\eta}}{Q}$\footnote{Refer to Section \ref{sec:class} for the notation and setup of the Bayesian classification problem.}. The mixture distribution $\MixtureD{P}{\eta}$ is formally defined in the following.

% We begin with the definition of mixed distributions.  Recall that we denote by $p$ the probability density of the distribution $P$. Roughly speaking the mixed distribution $\MixtureD{P}{\eta}$ lets us sample data from the improper density $p/\eta$, which again is related to the Neyman-Pearson test, where we reject if
% \[
% p(X)/q(X)> \eta \quad \Leftrightarrow \quad p(X)/\eta>q(X).
% \]
% Recently mixed distributions were related to approximate DP by~\cite{Lu2024}, and our related results provide a significant extension, establishing a link between the theory of optimal classification and $f$-DP.

\begin{definition}[Mixture Distribution]
Let $P$ be a distribution and $\eta \in [1, +\infty)$. The mixture distribution $\MixtureD{P}{\eta}$ is defined as:
\begin{align*}
    \MixtureD{P}{\eta} =
    \begin{cases}
        P & \text{with probability } \frac{1}{\eta}, \\
        \bot & \text{with probability } 1 - \frac{1}{\eta}.
    \end{cases}
\end{align*}
\end{definition}

We note that recent work \cite{Lu2024} showed that the parameters of approximate DP can be expressed in terms of the Bayes risk of carefully constructed Bayesian classification problems. They further showed how to construct such classification problems using mixture distributions. Building on this foundation, our results significantly extend their approach by establishing a direct link between the theory of optimal classification and $f$-DP.

\begin{algorithm}[!htp]
\footnotesize
\algorithmicrequire \; \parbox[t]{\dimexpr0.9\linewidth-\algorithmicindent}{Black-box access to $\Mech$; Threshold $\eta > 0$; Sample size $n$.}\\[0.1cm]
\algorithmicensure \, An estimate $(\tilde{\alpha}(\eta), \tilde{\beta}(\eta))$ of $(\alpha(\eta), \beta(\eta))$ for tuple $(P, Q)$, where $\Mech(\DB)$ and $\Mech(\DB')$ are distributed according to $P, Q$, respectively.
\begin{algorithmic}[1]
    \State Set the classifier $\phi$ for the Bayesian classification problem $\bbcP{\MixtureD{P}{\eta}}{Q}$ if $\eta \geq 1$; otherwise, set $\phi$ for the problem $\bbcP{P}{\MixtureD{Q}{1/\eta}}$. By default, use the $k$-NN classifier $\kNNclassifier{n}$ with $k = \sqrt{n}$.
    \Function{\textnormal{BayBox \B{Estimator}} $\bbe{\phi}(M, \DB, \DB', \eta,n)$}{}
    \State Set $cnt_{\alpha} \leftarrow 0$ and $cnt_{\beta} \leftarrow 0$
    \For{$i \in [n]$}
        \State $x \leftarrow \Mech(\DB)$; $x' \leftarrow \Mech(\DB')$
        \State If $\phi(x) = 1$ then $cnt_{\alpha} \leftarrow cnt_{\alpha} + 1$
        \State If $\phi(x') = 1$ then $cnt_{\beta} \leftarrow cnt_{\beta} + 1$
    \EndFor
    \State Return $(\tilde{\alpha}(\eta), \tilde{\beta}(\eta)) \leftarrow (\frac{cnt_{\alpha}}{n}, 1 - \frac{cnt_{\beta}}{n})$
    \EndFunction
\end{algorithmic}
\caption{BayBox: A Black-Box Bayesian Classification Algorithm for $f$-DP Estimation}
\label{alg: general BayBox estimator}
\end{algorithm}

\B{The key insight to connect classification and $f$-DP is that the trade-off point $(\alpha(\eta), \beta(\eta))$ can be expressed as the expected classification error of the Bayes optimal classifier. We propose a simple Monte Carlo estimator for the expected classification error and an implementation is given by the BayBox estimator in Algorithm~\ref{alg: general BayBox estimator}. In theory, if the Bayes optimal classifier $\phi^*$ were known and used as input of the BayBox algorithm, the output of BayBox would be an unbiased estimator for $(\alpha(\eta), \beta(\eta))$ that has a small error with high probability. A formal statement is provided in Lemma \ref{lemma: accuracy stat of general BayBox estimator} of the appendix.
In practice $\phi^*$ is unknown, but can be approximated using a k-NN classifier. A statement of the theoretical approximation properties is given in the next theorem. The notation "$\mathbb{E}$" refers to the expectation of a random variable, conditional on the threshold $\eta$.}

% Theorem \ref{lemma: accuracy stat of general BayBox estimator} establishes simultaneous confidence intervals for the parameters $\alpha(\eta)$ and $\beta(\eta)$ with a user-determined failure probability of $\gamma$ (it can be made arbitrarily small). Of course, in practice, the Bayes classifier $h^*$ is unknown and has to be approximated. A concrete approximation is provided by the kNN discussed below. However, Theorem \ref{lemma: accuracy stat of general BayBox estimator} is of independent interest, as it implies that our method is (in some sense) agnostic w.r.t. to the approximating classifier.
% We now derive an analogous result for the feasible kNN classifier.
%Theorem~\ref{thm: accuracy stat of kNN BayBox estimator} provides an analogous result for the feasible $k$-NN classifier. This is achieved by replacing the Bayes classifier $\phi^*$ with a concrete approximation provided by the $k$-NN classifier.
\B{
\begin{theo}%[Proof is in Appendix~\ref{proof: for thm: accuracy stat of kNN BayBox estimator}]
    \label{thm: accuracy stat of kNN BayBox estimator}
    Suppose that Assumption \ref{ass1} holds. Let $\eta$, $(\alpha(\eta), \beta(\eta))$, $(\tilde{\alpha}(\eta), \tilde{\beta}(\eta))$, and $\phi$ be defined as in Algorithm~\ref{alg: general BayBox estimator}. Set $\phi$ to the $k$-NN classifier $\kNNclassifier{n}$, with $k = \sqrt{n}$, for the corresponding Bayesian classification problem
   :\\[0.6ex]
    \textnormal{1)}  Then, for any $\gamma \in (0,1)$ and any $n \ge 2$ it holds with  probability $\ge 1-\gamma$ that
      \begin{align*}
    &\max\big\{|\tilde{\alpha}(\eta) - \mathbb{E}[\tilde{\alpha}(\eta)]| ,|\tilde{\beta}(\eta) - \mathbb{E}[\tilde{\beta}(\eta)]|\big\} \leq w(\gamma),
    \end{align*}
    \textnormal{2)} Moreover, for any $\gamma \in (0,1)$ and for all $n$ sufficiently large it holds with probability $\ge 1-\gamma$ that
     \begin{align*}
    &\max\big\{|\tilde{\alpha}(\eta) - \alpha(\eta)|, |\tilde{\beta}(\eta) - \beta(\eta)|\big\} \leq (25 c_d) w(\gamma)~,
    \end{align*} 
Here,  $c_d$ is a constant depending on the dimension $d$ with $c_d \le 4.9^d$ and the bound $w(\gamma)$ is defined as
      \begin{align}
    w(\gamma):=\sqrt{\ln(4/\gamma)/(2n)}~. \label{e:wgamma}
    \end{align}
\end{theo}}
\B{The two statements in the above theorem are distinct and have different interpretations. Part 1) shows that the output of the BayBox algorithm $(\tilde \alpha(\eta), \tilde \beta(\eta))$ is randomly fluctuating in a narrow region of width $w(\gamma)$ around its expectation $(\mathbb{E} \tilde \alpha(\eta),\mathbb{E}  \tilde \beta(\eta))$. The expectation $(\mathbb{E} \tilde \alpha(\eta),\mathbb{E}  \tilde \beta(\eta))$ can be shown to always lie on or above the optimal trade-off curve (Remark \ref{lem:bias}, in the appendix) 
and in this sense the output  $(\tilde \alpha(\eta), \tilde \beta(\eta))$ can be slightly biased (it may overstate privacy).
Part 2) of the theorem entails the stronger statement that $(\tilde \alpha(\eta), \tilde \beta(\eta))$ is actually close to the true value $(\alpha(\eta), \beta(\eta))$. The price is a looser bound by a factor of $(25 c_d) $, which arises from bounding the distances $|\mathbb{E}\tilde \alpha(\eta)-\alpha(\eta)|$ and $|\mathbb{E}\tilde \beta(\eta)-\beta(\eta)|$.
In principle, auditing mechanisms for $f$-DP can be based on either part 1) or part 2) of the theorem, and in the next section we give details. Practically, using part 1) yields better results (more accurate detection for lower sample sizes) and will be used in our below methodology. Notice that it is also a finite sample bound and non-asymptotic.}

\subsection{Auditing $f$-DP} \label{sec:audit}

\textbf{Outline} In the remainder of this section, we present an $f$-DP auditor that fuses the localization of maximum vulnerabilities (by the KDE method) with the confidence guarantees (afforded by the $k$-NN method). We can describe the problem as follows: Usually, when a DP mechanism $M$ is developed it comes with a privacy guarantee for users. In the case of standard DP this takes the form of a single parameter $\varepsilon_0$. In the case of $f$-DP a privacy guarantee is associated with a continuous trade-off curve $T^{(0)}$. Essentially the developer promises that the mechanism will afford at least $T^{(0)}$-DP. The task of the auditor is to empirically and reliably check this claim.

\noindent\textbf{The auditor} We proceed in two steps. Since we do not want to force the two steps to depend on the same sample size parameters, we introduce two (potentially different) sample sizes $n_1, n_2$. First, using the KDE method, we find an estimated value of maximum vulnerability $\hat \eta^*$ (based on a sample of size $n_1)$. This is possible according to  Proposition \ref{prop1}. Second, we apply the BayBox algorithm with input $\hat \eta^*$ and sample size $n_2$, giving us outputs $(\tilde \alpha(\hat \eta^*), \tilde \beta(\hat \eta^*))$. \B{Then, we draw on Theorem \ref{thm: accuracy stat of kNN BayBox estimator}  to check the $T^{(0)}$-DP claim. More precisely, recall that the pair $(\mathbb{E} \tilde \alpha(\hat \eta^*),\mathbb{E}  \tilde \beta(\hat \eta^*))$ lies on or above the optimal, unknown trade-off curve $T$. This is intuitively clear, because any classifier will have a worse (at best equal) performance as the Bayes optimal classifier, which lies exactly on the curve $T$. Now, from Theorem \ref{thm: accuracy stat of kNN BayBox estimator} part 1) we know that the pair $(\mathbb{E} \tilde \alpha(\hat \eta^*),\mathbb{E}  \tilde \beta(\hat \eta^*))$ is included with high probability inside the box }
%According to Theorem \ref{thm: accuracy stat of kNN BayBox estimator} it holds with high probability ($1-\gamma$) that the values $(\alpha(\hat \eta^*), \beta(\hat \eta^*))$ of the optimal test are contained inside the square
\begin{align} \label{e:defsq}
\square_\gamma:= \,\big[\tilde{\alpha}(\hat \eta^*) - w(\gamma),\tilde{\alpha}(\hat \eta^*) + w(\gamma)\big] &\\
\qquad\times  \big[\tilde{\beta}(\hat \eta^*) - w(\gamma),\tilde{\beta}(\hat \eta^*) + w(\gamma)\big]. \nonumber & 
\end{align}
\B{This means that $\square_\gamma$ includes points that are above the optimal trade-off curve $T$ with high probability.  
Now there are two cases: First, if the claim of $T^{(0)}$-DP is true (i.e. $T^{(0)} \le T$), then some points in $\square_\gamma$ must also be above $T^{(0)}$. For our auditor this means that if there are points in $\square_\gamma$ that are above $T^{(0)}$, we will detect "no privacy violation" (see for an illustration Figure \ref{fig:not_faulty_sgd_gauss}). In other words, our findings are compatible with  $T^{(0)}$-DP. 
Conversely, if we observe that the entire box $\square_\gamma$ is below $T^{(0)}$, then our auditor will detect a "privacy violation" and our findings are at odds with  $T^{(0)}$-DP (see for an illustration Figure \ref{fig:faulty_sgd_gauss}).}
%Now, the claim of $T^{(0)}$-DP is true (i.e. $T_0 \le T$), then some points in $\square_\gamma$ must also be above $T^{(0)}$. 
%Put differently, after running the BayBox algorithm, the only plausible values for $(\alpha(\hat \eta^*), \beta(\hat \eta^*))$ are inside $\square_\gamma$. \\
%Now, since $(\alpha(\hat \eta^*), \beta(\hat \eta^*))$ is a pair of errors associated with the optimal test, it corresponds to a point on the optimal trade-off curve. If this point were below the curve $T^{(0)}$, the claim of $T^{(0)}$-DP would be wrong. We do not know the exact value of $(\alpha(\hat \eta^*), \beta(\hat \eta^*))$, but we do know (with high certainty) that it is inside the very small box $\square_\gamma$. If the entirety of this box is below $T^{(0)}$, there seems no plausible way that  $T^{(0)}$-DP is satisfied and the auditor will detect a privacy violation. If, on the other hand, some or all of the values in $\square_\gamma$ are on or above  $T^{(0)}$, our %data suggests that $M$ plausibly satisfies $T^{(0)}$-DP and the 
%auditor does not detect a violation. 
Algorithm \ref{auditor} summarizes the procedure we have just described. It uses a small geometrical argument to check more easily whether the entire box is below $T^{(0)}$ or not (see lines $6-7$ of the algorithm).

\begin{algorithm}[h]
\footnotesize
\algorithmicrequire \; \parbox[t]{\dimexpr0.9\linewidth-\algorithmicindent}{Mechanism $\Mech$, neighboring databases $\DB, \DB'$, sample sizes $n_1, n_2$, confidence level $\gamma$, threshold vector $\eta$, claimed curve $T^{(0)}$.}\\[0.1cm]
\algorithmicensure \, "Violation" or "No Violation".
\begin{algorithmic}[1]
    \Function{\textnormal{Auditor}$(\Mech, \DB, \DB', n_1, n_2, \gamma,\eta, T^{(0)})$}{}
        \State Compute $\hat{T}_h$ using $\ptlr{h}{\cA}(M, \DB, \DB', \eta,n_1)$ for all $\eta_i \in \eta$.
        \State Compute $\hat{\eta}^* \in \arg\max \left\{T^{(0)}(\hat{\alpha}_h(\eta)) - \hat{T}_h(\hat{\alpha}_h(\eta)) : \eta \ge 0\right\}$.
        \State Run the $k$-NN BayBox estimator $\bbe{\kNNclassifier{n_2}}(M, \DB, \DB', \hat \eta^*,n_2)$ to obtain $(\tilde{\alpha}(\hat{\eta}^*), \tilde{\beta}(\hat{\eta}^*))$.
        \State Calculate the threshold $w(\gamma)$ from eq. \eqref{e:wgamma}
        \State Calculate $i^*$ as the solution to $T^{(0)}(i^*) = \tilde{\beta}(\hat{\eta}^*) + w(\gamma)$.
        \If{$i^* > \tilde{\alpha}(\hat{\eta}^*) + w(\gamma)$}
          \State \Return "Violation".
        \Else
           \State \Return  "No Violation".
        \EndIf
    \EndFunction
\end{algorithmic}
\caption{Privacy Violation Detection Algorithm}\label{alg:auditor}
\label{auditor}
\end{algorithm}

% \begin{algorithm}
%     \caption{Privacy Violation Detection Algorithm}
%     \label{auditor}
%     \begin{algorithmic}[1]
%         \State \textbf{Input:} Mechanism $M$, neighboring databases $D, D'$, sample sizes $n_1, n_2$, confidence level $\gamma$
%         \Function{Auditor}{$M, D, D', n, \gamma$}
%             \State Compute $\hat{T}_h$ (defined in eq. \eqref{e:def:Th}) with sample size $n_1$
%             \State Compute $\hat{\eta}^* \in \arg\max \left\{T^{(0)}(\hat{\alpha}_h(\eta)) - \hat{T}_h(\hat{\alpha}_h(\eta)) : \eta \ge 0\right\}$
%             \State Run the BayBox algorithm for KNN classifiers with input $\eta = \hat{\eta}^*$ and sample size $n_1$. Obtain $(\tilde{\alpha}(\hat{\eta}^*), \tilde{\beta}(\hat{\eta}^*))$
%             \State Calculate the threshold $w(\gamma)$ from eq. \eqref{e:wgamma}
%             \State Calculate $i^*$ as the solution to $T_0(i^*) = \tilde{\beta}(\hat{\eta}^*) + w(\gamma)$
%             \If{$i^* > \tilde{\alpha}(\hat{\eta}^*) + w(\gamma)$}
%               \State \Return "Violation"
%             \Else
%                \State \Return  "No Violation"
%             \EndIf
%         \EndFunction
%     \end{algorithmic}
%     \textcolor{orange}{\textbf{Yu:} I have created a pseudocode (of sorts) here. Cann you bring it into your Pseudocode notation? Also, do we need a Pseudocode for KDE that we can call?}
% \end{algorithm}

\noindent\textbf{Theoretical analysis} To provide theoretical guarantees for the algorithm, we add a mathematical assumption on the trade-off curve of $p \sim M(D), q \sim M(D')$.

\begin{ass} \label{ass3}
    The optimal trade-off curve $T$ corresponding to the output densities $p,q$ is strictly convex.
\end{ass}
We can now formulate the main theoretical result for the auditor. 

\begin{theo} \label{theo:auditor}
    Suppose that Assumptions \ref{ass1} and \ref{ass2} hold, let $\gamma \in (0,1)$ be user-determined and denote the output of \B{Auditor($M, D, D', n_1, n_2, \gamma$,$\eta$, $T^{(0)}$)} by $A$.
    \begin{itemize}
        \item[1)] Then, if $T^{(0)}(\alpha) \le T(\alpha)$ for all $\alpha \in [0,1]$ (no violation of $T^{(0)}$-DP), it follows for any $n_1,n_2 \ge 2$ 
    \[
    \Pr\Big[ A = "\textnormal{No Violation}"\Big]\ge 1-\gamma.
    \]
 
    \item[2)] Suppose that additionally Assumption \ref{ass3} holds. Then, if $T^{(0)}(\alpha^*) > T(\alpha^*)$ for some $\alpha^* \in [0,1]$ (a violation of $T^{(0)}$-DP), it follows that   
    \[
    \lim_{n_1 \to \infty} \,\,\liminf_{n_2 \to \infty} \,\,\Pr\Big[ A = "\textnormal{Violation}"\Big]= 1.
    \]
    \end{itemize}
\end{theo}
Part 1) of the theorem states that the risk of falsely detecting a violation can be made arbitrarily small ($\le \gamma$) by the user. On the other hand, if some violation exists, part 2) assures that it will be reliably detected for large enough sample sizes. We note that for smaller values of $\gamma$ larger sample sizes are typically needed to detect violations. This follows from the definition of the box $\square_\gamma$ in \eqref{e:defsq}. \\\B{The theoretical Assumptions \ref{ass1}-\ref{ass2} of the theorem are comparable to related works \cite{Kutta2024,Lu2024} and require smoothness of the output distributions $p,q$. Such assumptions are required to avoid known impossibility results in privacy estimation (see \cite{gorla2023impossibility}). Assumption \ref{ass3} of a strictly convex trade-off function is often satisfied (e.g. for Gaussian type mechanisms), but can be further relaxed. A simple to prove but fairly general relaxation is that $T$ is only strictly convex in a sufficiently small, open neighborhood of the set $argmax(T^{(0)}-T)$. We do not include it here, to avoid making the results even more technical. } 
\begin{rem}
The auditor in Algorithm \ref{alg:auditor} uses the threshold $\hat \eta^*$ (see eq.
\ref{e:def:eta}), to locate the maximum vulnerability. We point out that any other method to find vulnerabilities would still enjoy the guarantee from part 1) of Theorem \ref{theo:auditor} (it is a property of $k$-NN), but not necessarily of part 2). It might be an interesting subject of future work to consider other ways of choosing $\hat \eta^*$ (e.g. based on the two dimensional Euclidean distance between $T^{(0)}$ and $\hat T_h$ rather than the supremum distance).
\end{rem}

\B{
\begin{rem}
Our black-box algorithms for estimation and auditing face computational limitations due to sample size requirements and the curse of dimensionality. These challenges arise from the black-box setting itself, where we require larger amounts of data samples to make up for missing information and knowledge with regard to algorithm structure. In addition, higher dimensional algorithm output only increases the need for larger sampling efforts. This makes the auditing of machine learning on large, real-world datasets challenging. Here, white-box methods that aim at minimizing the amount of trained models needed for a privacy audit can be of help (see e.g.\cite{Nasr2023}). We discuss how modifications of our algorithms could help with these limitations in Section \ref{sec:summary_discussion}.
\end{rem}
}

\section{Experiments} \label{sec6}
We investigate the empirical performance of our new procedures in various experiments to demonstrate their effectiveness.
%To demonstrate the effectiveness of our new procedures, we investigate their empirical performance in the following experiments. 
Recall that our procedures are developed for two distinct goals, namely estimation of the optimal trade-off curve $T$ (see Section \ref{sec:4}) and auditing a privacy claim $T^{(0)}$ (see Section \ref{sec:goal2}). We will run experiments for both of these objectives. \\
%These goals correspond to Sections \ref{sec:4} and \ref{sec:goal2} respectively. \\
%This section aims to validate the theoretical results presented in Section~\todo{cite section} and Section~\todo{cite section}. \\
\textbf{Experiment Setting:} 
%We have outlined two distinct objectives along with their corresponding methodologies:
%\begin{description}
 %   \item[\textbf{Goal 1: Uniform Estimation of the Privacy Curve $T$}]  
 %   The first objective is to uniformly estimate an unknown privacy curve $T$, as stated in Theorem~\ref{theo:1}. To validate not only the theoretical correctness but also the practical effectiveness of this estimation approach, we conducted a simulation study on all four mechanisms. The results of this study are presented in \todo{Table~\ref{tab:estimation_f_curves} and Figure~\ref{fig:todo}.}
 %   \item[\textbf{Goal 2: Detection of Privacy Violations}] 
 %   The second objective is inferential in nature. As formulated in Theorem~\ref{theo:auditor}, the goal is to detect privacy violations for a predefined false positive rate. To demonstrate the effectiveness of this methodology, we constructed faulty algorithms and analyzed their behavior. The results of this analysis are depicted in Figure~\ref{fig:todo}.
%\end{description}
Throughout the experiments, we consider databases $\DB,\DB' \in [0,1]^r$, where the participant number is always $r=10$. As discussed in Section \ref{sec:overview_techniques}, we first choose a pair of neighboring datasets such that there is a large difference in the output distributions of $\Mech(D)$ and $\Mech(D')$. We can achieve this by simply choosing $D$ and $D'$ to be as far apart as possible (while still remaining neighbors) and we settle on the choice 
%As typical in the privacy validation literature, we consider two neighboring databases that are far apart. On the $r$-dimensional cube $[0,1]^r$ we make the natural choice of
\begin{equation}\label{eq_databases}
    \DB=(0,\hdots, 0)\quad \textnormal{and} \quad \DB'=(1,0,\hdots, 0)
\end{equation}
for all our experiments.
%and notice that similar results as the ones below hold for other pairs of databases. %Our methods do however work just as well for other data bases $D$ and $D'$.
%Additionally, for data lying in the unit cube, this choice is natural, as these two databases are far apart on the unit cube.

\subsection{Mechanisms}\label{sec:algorithms}
In this section, we test our methods on two frequently encountered mechanisms from the auditing literature: the Gaussian mechanism and differentially private Stochastic Gradient Descent (DP-SGD). We study two other prominent DP algorithms -- the Laplace and Subsampling mechanism -- in Appendix \ref{AppB}. \\
%We apply our methods to four mechanisms frequently encountered in the privacy literature: the Gaussian mechanism, the Laplace mechanism, the Subsampling mechanism, and, most notably, the Noisy Stochastic Gradient Descent (DP-SGD) mechanism. These algorithms are quite  heterogeneous and hence collectively form a good benchmark to evaluate our methods. We quickly review these mechanisms and specify parameter settings. \\

\noindent \textbf{Gaussian mechanism.}
We consider the summary statistic $S(x)= \sum_{i=1}^{10} x_i$ and the mechanism
\begin{equation*}
    M(x):= S(x)+Y~,
\end{equation*}
where $Y\sim \mathcal N (0, \sigma^2)$. The statistic $S(x)$ is privatized by the random noise $Y$ if the variance $\sigma^2$ of the Normal distribution is appropriately scaled. We choose $\sigma = 1$ for our experiments and note that - in our setting - the optimal trade-off curve is given by 
\begin{align*}
     T_{Gauss}(\alpha)= \Phi(\Phi^{-1}(1-\alpha)- \mu)
\end{align*}
with $\mu = 1$. We point the reader to \cite{Dong2022} for more details. \\

\noindent \textbf{DP-SGD.} The DP-SGD mechanism is designed to (privately) approximate a solution for the empirical risk minimization problem
\begin{equation*}
\theta^*=argmin_{\theta\in \Theta} \mathcal L_x(\theta) \quad \text{with} \quad \mathcal L_x(\theta)=\frac{1}{r}\sum_{i=1}^{r} \ell(\theta, x_i)~.
\end{equation*}
Here, $\ell$ denotes a loss function, $\Theta$ a closed convex set and $\theta^*\in \Theta$ the unique optimizer. For sake of brevity, we provide a description of DP-SGD in the appendix (see Algorithm \ref{alg:noisy_sgd}). In our setting, we consider the loss function $\ell(\theta, x_i)=\frac{1}{2} (\theta-x_i)^2$, initial model $\theta_0=0$ and $\Theta=\mathbb{R}$. The remaining parameters are fixed as $\sigma=0.2, \rho = 0.2, \tau = 10, m=5$. In order to have a theoretical benchmark for our subsequent empirical findings, we also derive the theoretical trade-off curve $T_{SGD}$ analytically for our setting and choice of databases (see Appendix \ref{AppB} for details). Our calculations yield
%For the choice of databases as in equation \eqref{eq_databases}, one can compute the trade-off curve $T_{SGD}$ analytically: 
\begin{equation*}
    T_{SGD}(\alpha)=\sum_{I\subset \{1,\hdots, \tau \}} \frac{1}{2^{\tau}}\Phi\Big(\Phi^{-1} (1-\alpha)-\frac{\mu_I}{\bar\sigma}\Big)~,
\end{equation*}
where $\mu_I$ is chosen as in \eqref{mu_I} and $\bar{\sigma}$ as in \eqref{sigma_bar}.

\subsection{Simulations}
We begin by outlining the parameter settings of our KDE and $k$-NN methods for our simulations. We then discuss the metrics employed to validate our theoretical findings and, in a last step, present and analyze our simulation results.\\
\textbf{Parameter settings:}
%For the subsequent simulations we always use the same parameters across all algorithms, acknowledging the black-box setting. 
For the KDEs, we consider different sample sizes of $n_1=10^2,10^3,10^4,10^5,10^6$ and we fix the perturbation parameter at $h=0.1$. For the bandwidth parameter $b$ (see Sec. \ref{sec:kde}), we use the method of \cite{bandwidth}. To approximate the optimal trade-off curve, we use $1000$ equidistant values for $\eta$ between $0$ and $15$ (see Algorithm \ref{alg:pointwise_KDE_estimator} for details on the procedure). For the $k$-NN, we set the training sample size to \B{$n_2=10^3,10^4,10^5$} and testing sample size to $10^3,10^4$ and $10^5$. \\
%\todo{Yu: are you sure that this is sufficient?}\\

\noindent \textbf{Estimation}
The first goal of this work is estimation of the optimal trade-off curve $T$. In our experiments, we want to illustrate the uniform convergence of the estimator $\hat T_h$ to the optimal curve $T$, derived in Theorem \ref{theo:1}. Therefore, we consider increasing sample sizes $n_1$ to study the decreasing error. The distance of $\hat T_h$ and $T$ in each simulation run is measured by the  uniform distance\footnote{Of course, one cannot practically maximize over all (infinitely many) arguments $\alpha \in [0,1]$. The estimator $\hat T_h$ is made for a grid of values for $\eta$ (see our parameter settings above) and we maximize over all gridpoints.} %maximum distance on a grid $G$ 
%We repeatedly estimate the respective trade-off curves of the four mechanism introduced in Section \ref{sec:algorithms} and computed 
\[
    Error_T:=\sup_{\alpha \in [0,1]}|\hat T_h(\alpha)-T(\alpha)|.
\]
%on a grid $G$ of $[0,1]$. 
%In our setting, we defined $G$ as the grid given by the KDE. Since, we choose $1000$ $\eta$ equidistant, we will get $1000$ $\alpha$ values. However, they do not have to be equidistant nor unique. 
To study not only the distance in one simulation run, but across many, we calculate $Error_T$ in $1000$ independent runs and take the (empirical) mean squared error
\begin{equation}\label{eq:mse}
    MSE(Error_T):= \Ex{Error_T^2}.
    %\mathbb{E}\mathbb Var(Error_G)+\mathbb E[Error_G]^2~.
\end{equation}
The results are depicted in Figure \ref{fig:estimation_mse} for the DP algorithms described in this section and the appendix. On top of that, we also construct figures that upper and lower bound the worst case errors for the Gaussian mechanism and DP-SGD over the $1000$ simulation runs. These plots visually show how the error of the estimator $\hat T_h$ shrinks as $n_1$ grows. 
% for the sample size $n_1=1000$. 
%For that, we computed the worst estimation point wise on an equidistant discretization of $[0,1]$ and interpolated the curves linearly. 
The results are summarized in Figures \ref{fig:gaussian}-\ref{fig:sgd}.\\
\begin{figure}
\centering\includegraphics[width=0.75\linewidth]{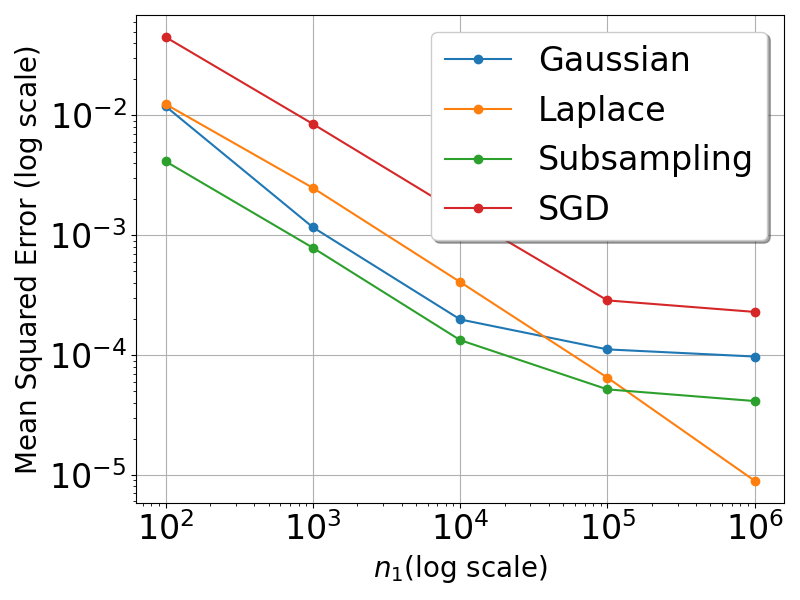}
    \caption{\centering
    MSE defined in \eqref{eq:mse} to empirically validate Theorem \ref{theo:1} for varying sample sizes $n_1$ and over $1000$ simulation runs each.}\label{fig:estimation_mse}
\end{figure}
\begin{figure*}[h!]
    \centering
    \subfloat[$n_1=10^3$]{\includegraphics[width=0.3\textwidth]{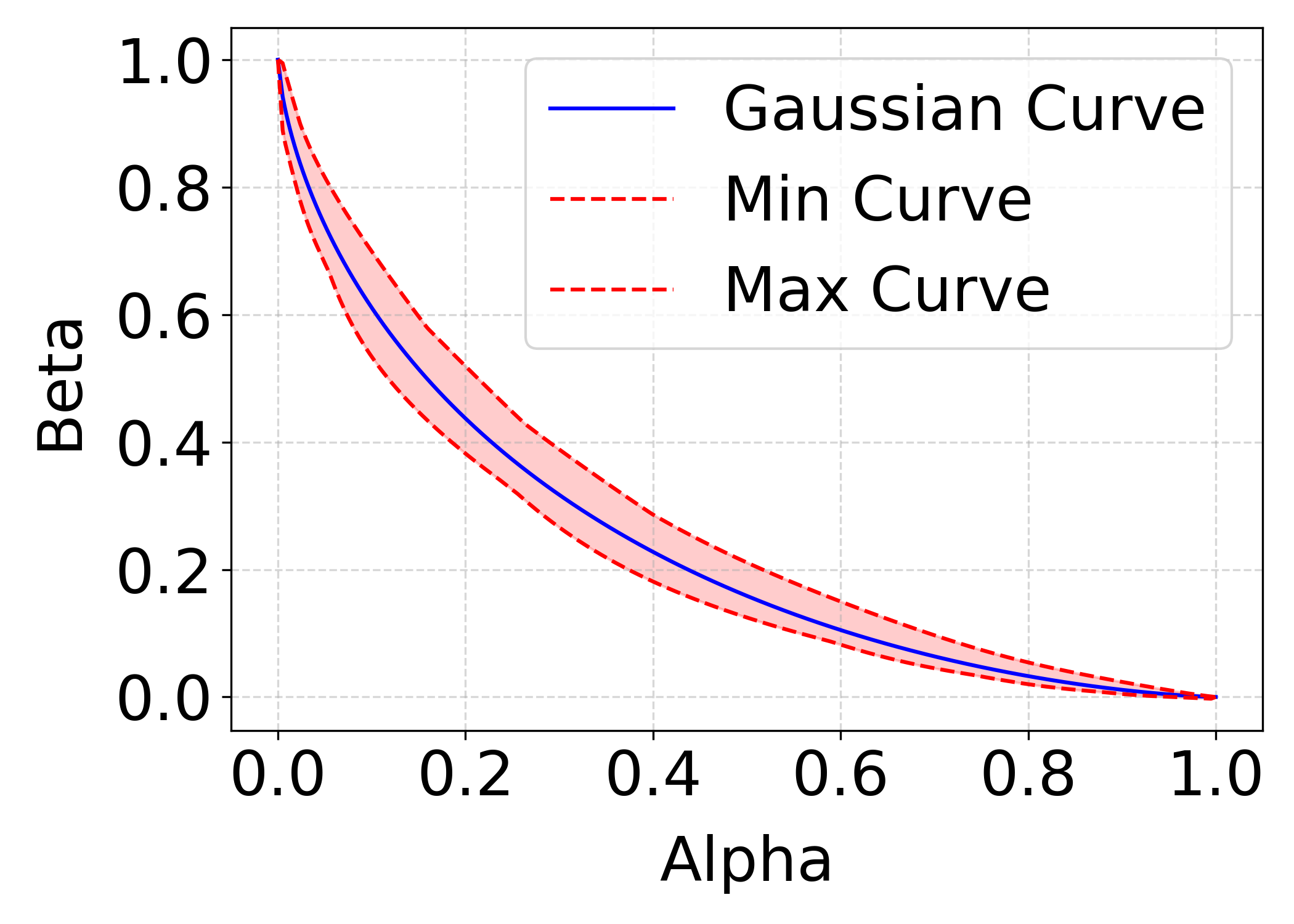}}
    \hfill
    \subfloat[$n_1=10^4$]{\includegraphics[width=0.3\textwidth]{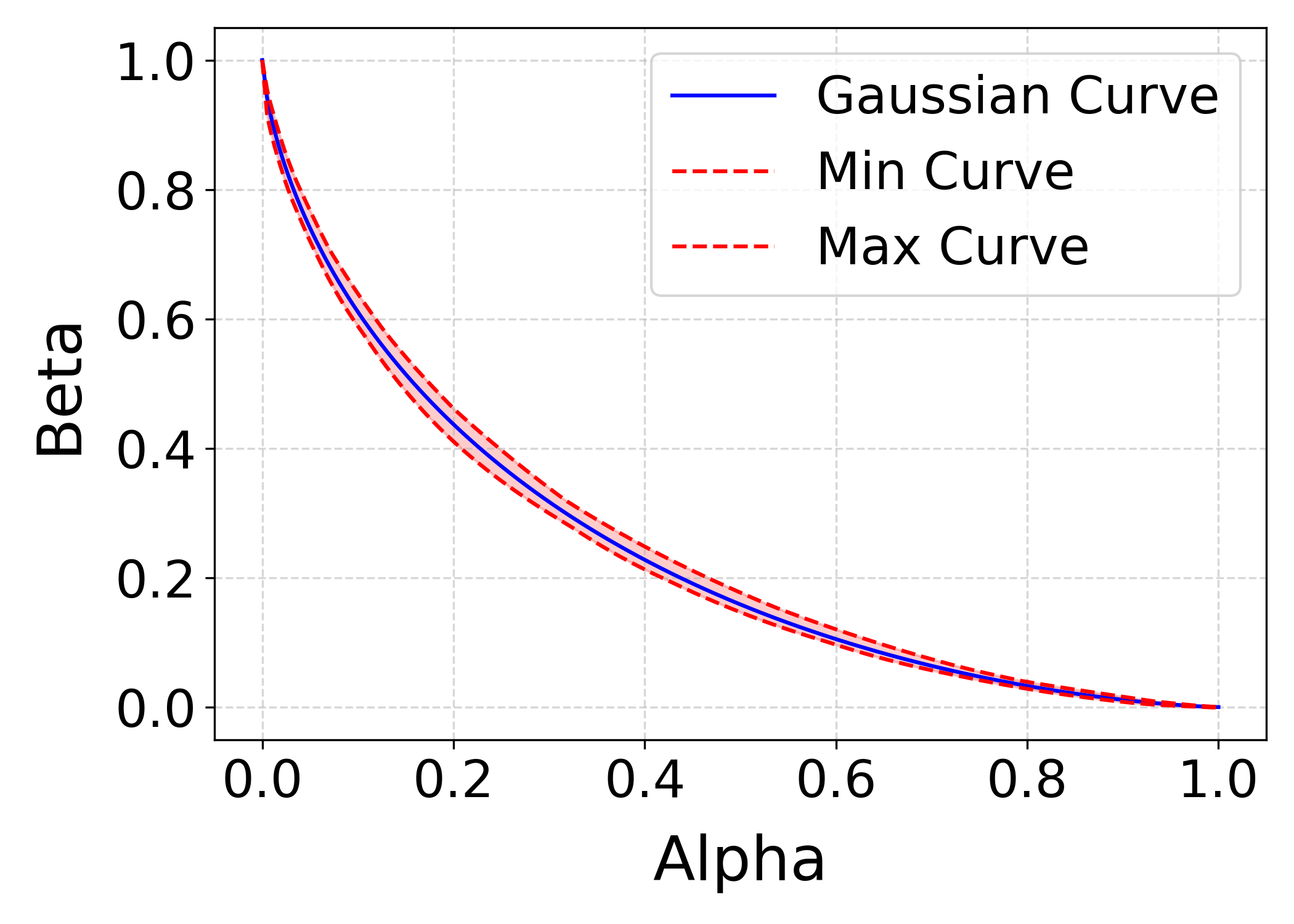}}
    \hfill
    \vspace{-0.2cm}
    \subfloat[$n_1=10^5$]{\includegraphics[width=0.3\textwidth]{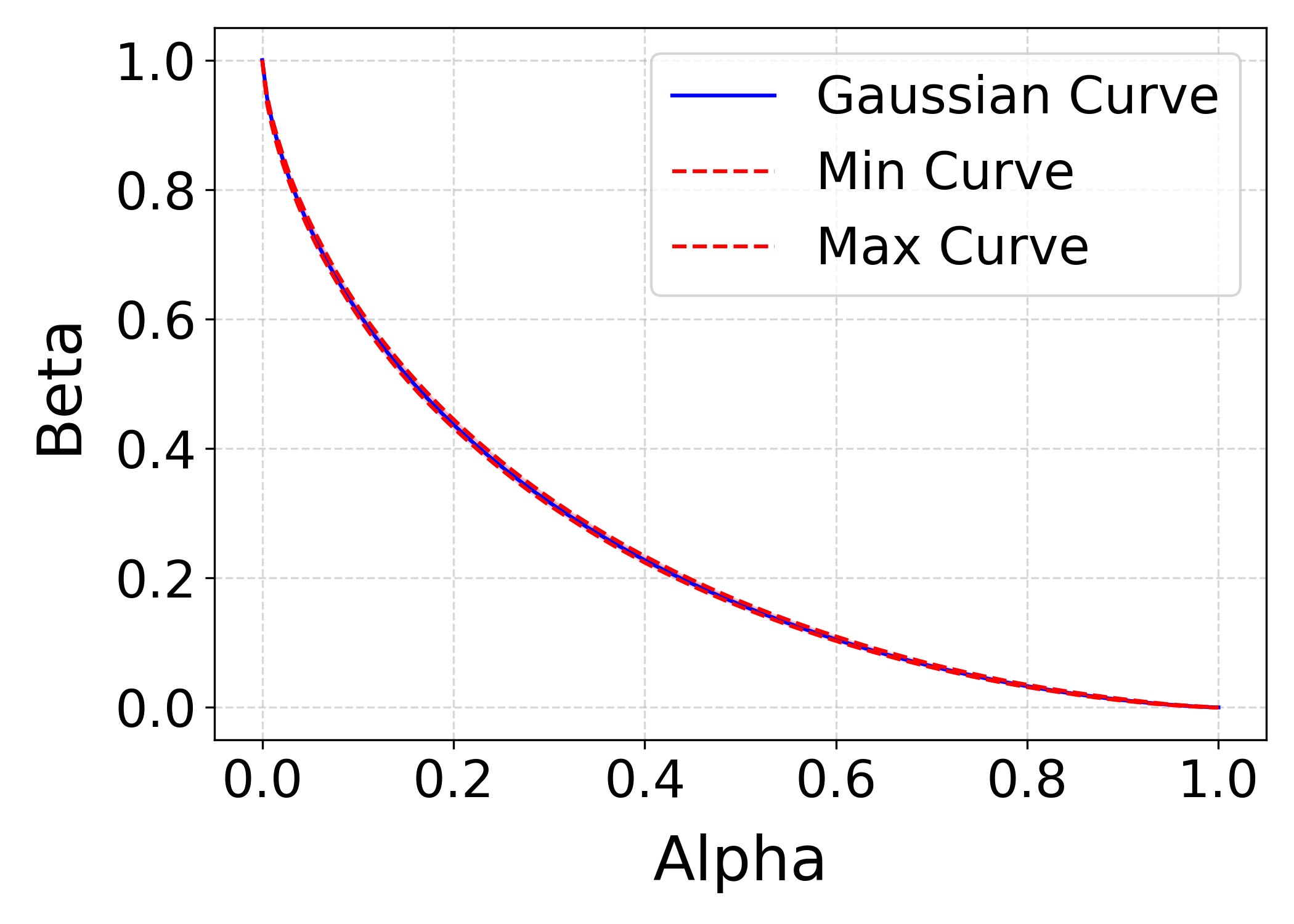}}
    \caption{Estimation of the Gaussian Trade-off curve $T_{Gauss}$ for varying sample sizes and $\mu=1$. Min- and Max Curve lower- and upper bound the worst point-wise deviation from the true curve $T_{Gauss}$ over $1000$ simulations.}
    \label{fig:gaussian}
\vspace{-0.1cm}
\centering
    \subfloat[$n_1=10^3$]{\includegraphics[width=0.3\textwidth]{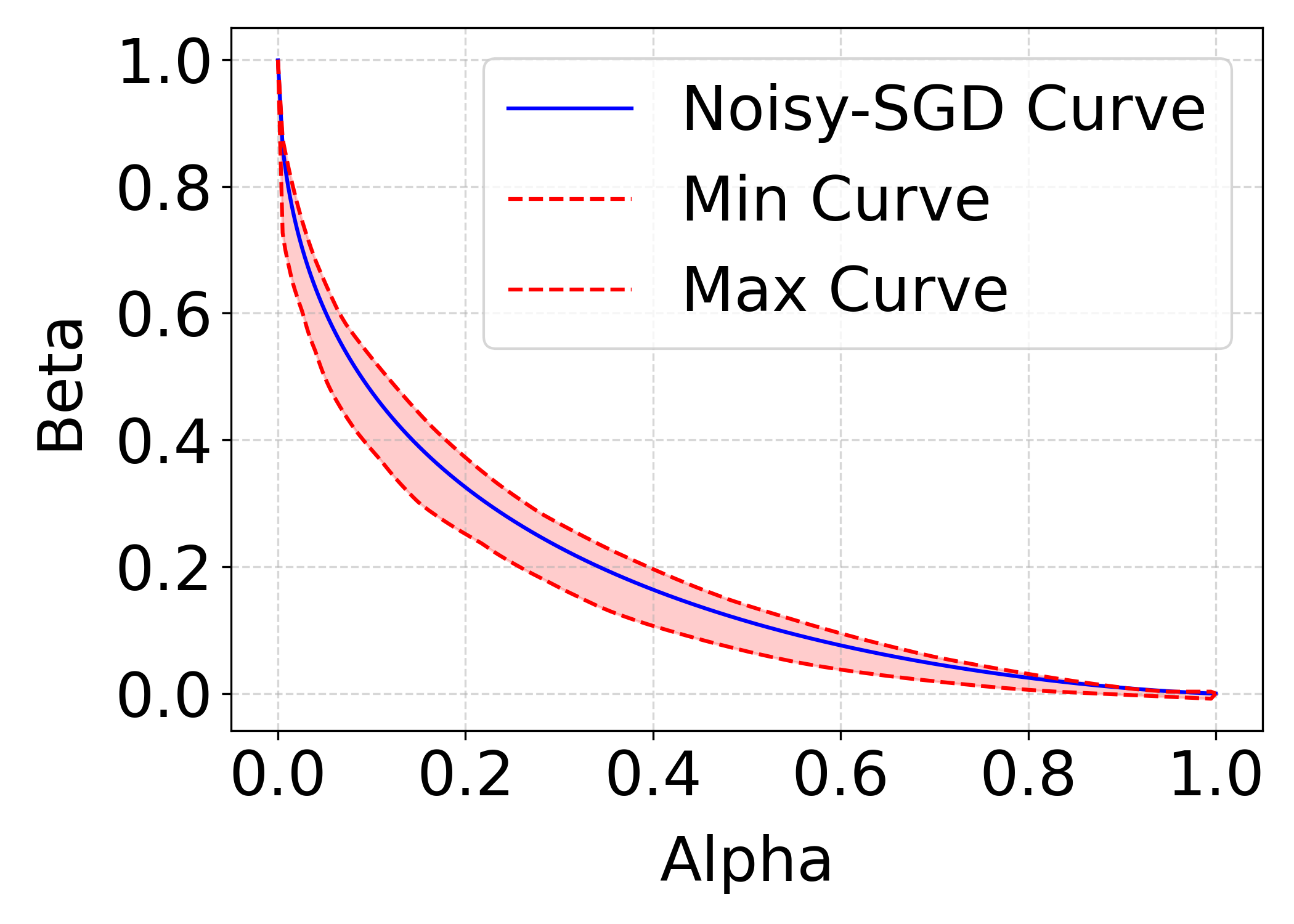}}
    \hfill
    \subfloat[$n_1=10^4$]{\includegraphics[width=0.3\textwidth]{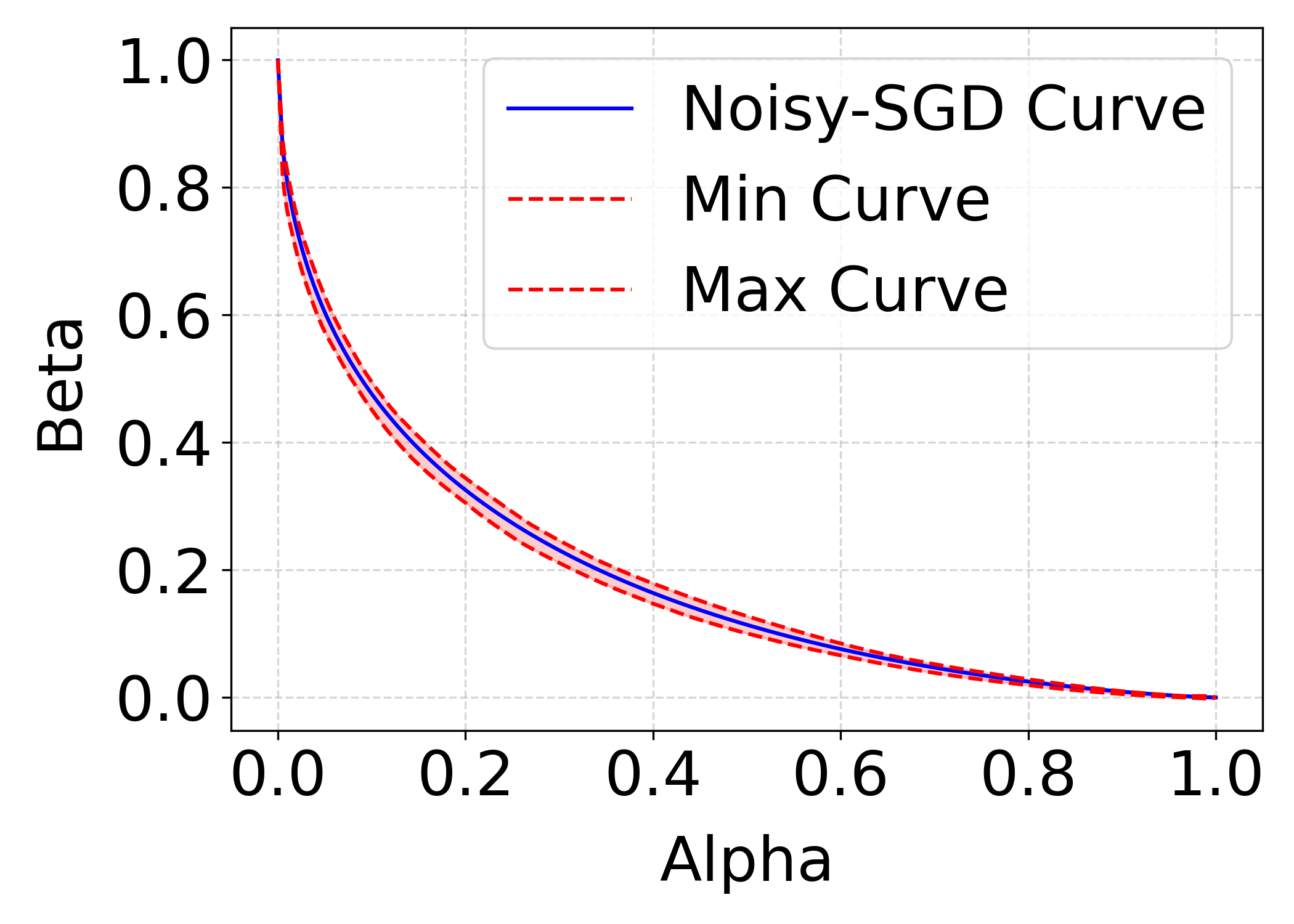}}
    \hfill
    \vspace{-0.2cm}
    \subfloat[$n_1=10^5$]{\includegraphics[width=0.3\textwidth]{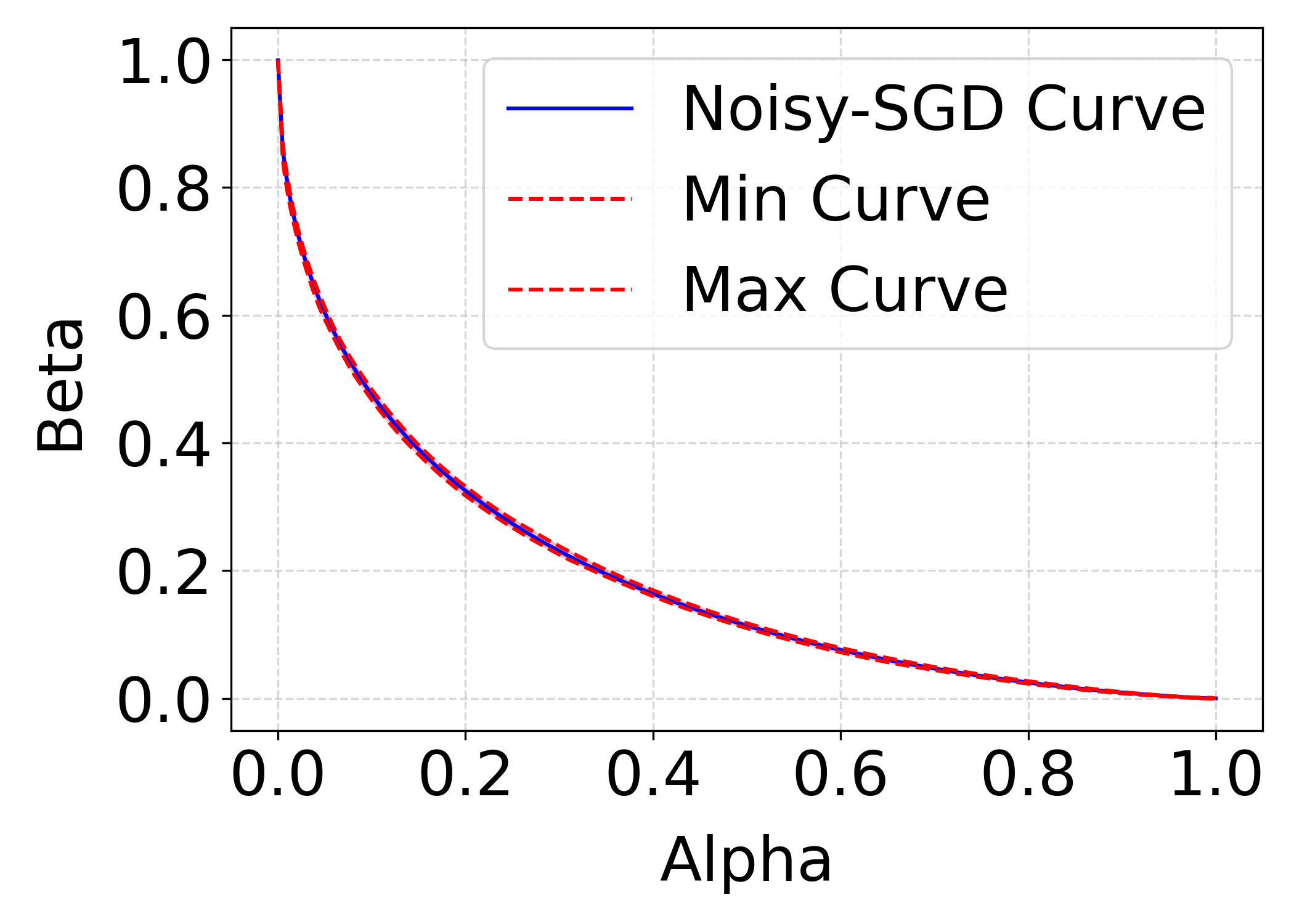}}
    \caption{Estimation of the DP-SGD Trade-off curve $T_{SGD}$ for varying sample sizes. Min- and Max Curve lower- and upper bound the worst point-wise deviation from the true curve $T_{SGD}$ over $1000$ simulations.}
    \label{fig:sgd}
\end{figure*}

\noindent {\textbf{Inference}\label{Inference}}
Next, we turn to the second goal of this work: Auditing a $T^{(0)}$-DP claim for a postulated trade-off curve $T^{(0)}$. 
The theoretical foundations of our auditor can be found in Theorem \ref{theo:auditor}. The theorem makes two guarantees: First, that for a mechanism $M$ satisfying $T^{(0)}$-DP the auditor will (correctly) not detect a violation, except with low, user-determined probability $\gamma$. Second, if $M$ violates  $T^{(0)}$-DP, the auditor will (correctly) detect the violation for sufficiently large sample sizes $n_1,n_2$. Together, these results mean that if a violation of $T^{(0)}$-DP is detected by the auditor, the user can have high confidence that $M$ does indeed not satisfy $T^{(0)}$-DP. 
%To begin, we examine the first result, which ensures, informally speaking that the auditor will not generate more than $\gamma>0$ false positives, when auditing a mechanism $M$. The second result guarantees that if the claimed privacy does not hold, the auditor will eventually identify this with probability $1$ as the sample size increases.
For the first part, we consider a scenario, where the claimed trade-off curve $T^{(0)}$ is the correct one $T^{(0)}=T$ ($M$ does not violate $T^{(0)}$-DP). For the second part, we choose a function $T^{(0)}$ above the true curve $T$ ($M$ violates $T^{(0)}$-DP). We will consider both scenarios for the Gaussian mechanism and DP-SGD.
%We will use two of the four mechanism for illustration: First, the standard Gaussian mechanism, as an example of an additive noise mechanism and second the DP-SGD mechanism, as an example of a machine learning mechanism.
%We start with auditing correctly claimed curves $T^{(0)}$. For that purpose, 
We run our auditor (Algorithm \ref{alg:auditor}) with parameters $n_1=10^4$ and $\gamma=0.05$ fixed. The choice of $\gamma=0.05$ is standard for confidence regions in statistics and we further explore the impact of $n_1$ and $\gamma$ in additional experiments in Appendix \ref{AppB}. Here, we focus on the most impactful parameter, the sample size $n_2$ and study values of  \B{$n_2 = 10^3,10^4,10^5$}. \\
Technically, the auditor only outputs a binary response that indicates whether a violation is detected or not. However, in our below experiments, we depict the inner workings of the auditor and geometrically illustrate how a decision is reached. More precisely, in Figure \ref{fig:not_faulty_sgd_gauss} we depict the claimed trade-off curve $T^{(0)}$ as a blue line. The auditor makes an estimate for the true trade-of curve $T$, namely $\hat T_h$ depicted as the orange line. The location, where the orange line (estimated DP) and the blue line (claimed DP) are the furthest apart is indicated by the vertical, dashed green line. This position is associated with the threshold $\hat \eta^*$ in Algorithm \ref{alg:auditor}. As a second step, $\hat \eta^*$ is used in the $k$NN method to make a confidence region, depicted as a purple square (this is $\square_\gamma$ from \eqref{e:defsq}). If the square is fully below the claimed curve $T^{(0)}$, a violation is detected (Figure \ref{fig:faulty_sgd_gauss}) and if not, then no violation is detected (Figures \ref{fig:gaussian} and \ref{fig:sgd}). As we can see, detecting violations requires $n_2$ to be large enough, especially when $T^{(0)}$ and $T$ are close to each other. \\
For the incorrect $T^{(0)}$-DP claims, we have done the following: For the Gaussian case (Figure \ref{fig:faulty_sgd_gauss}), we have used a trade-off curve with parameter $\mu=0.5$ instead of the true $\mu=1$. For DP-SGD, we have used the trade-off curve corresponding to $\tau = 5$ instead of the true $\tau = 10$ iterations (Figure \ref{fig:faulty_sgd_gauss}). 

\begin{figure*}[t]
    \centering
    \subfloat[\centering \B{$n_2=10^3$},\textbf{Ground Truth:} No Violation; \newline \textbf{Decision:} \textcolor{green}{"No Violation"}{\textcolor{green}{\scalebox{1.5}{\ding{51}}}}]{\includegraphics[width=0.3\textwidth]{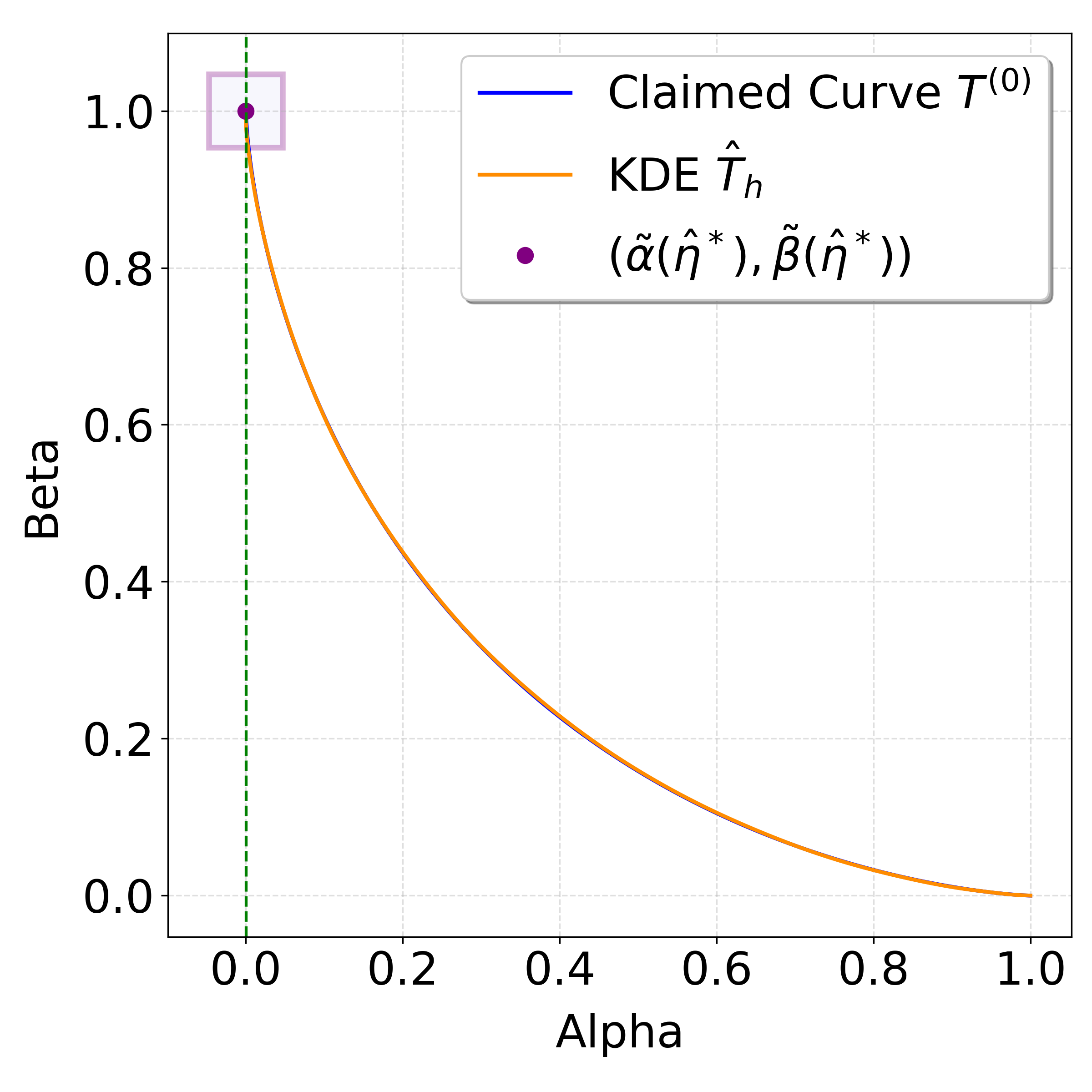}}
    \hfill
    \subfloat[\centering  \B{$n_2=10^4$},\textbf{ Ground truth:} No Violation; \newline \textbf{Decision:} \textcolor{green}{"No Violation"}{\textcolor{green}{\scalebox{1.5}{\ding{51}}}}]{\includegraphics[width=0.3\textwidth]{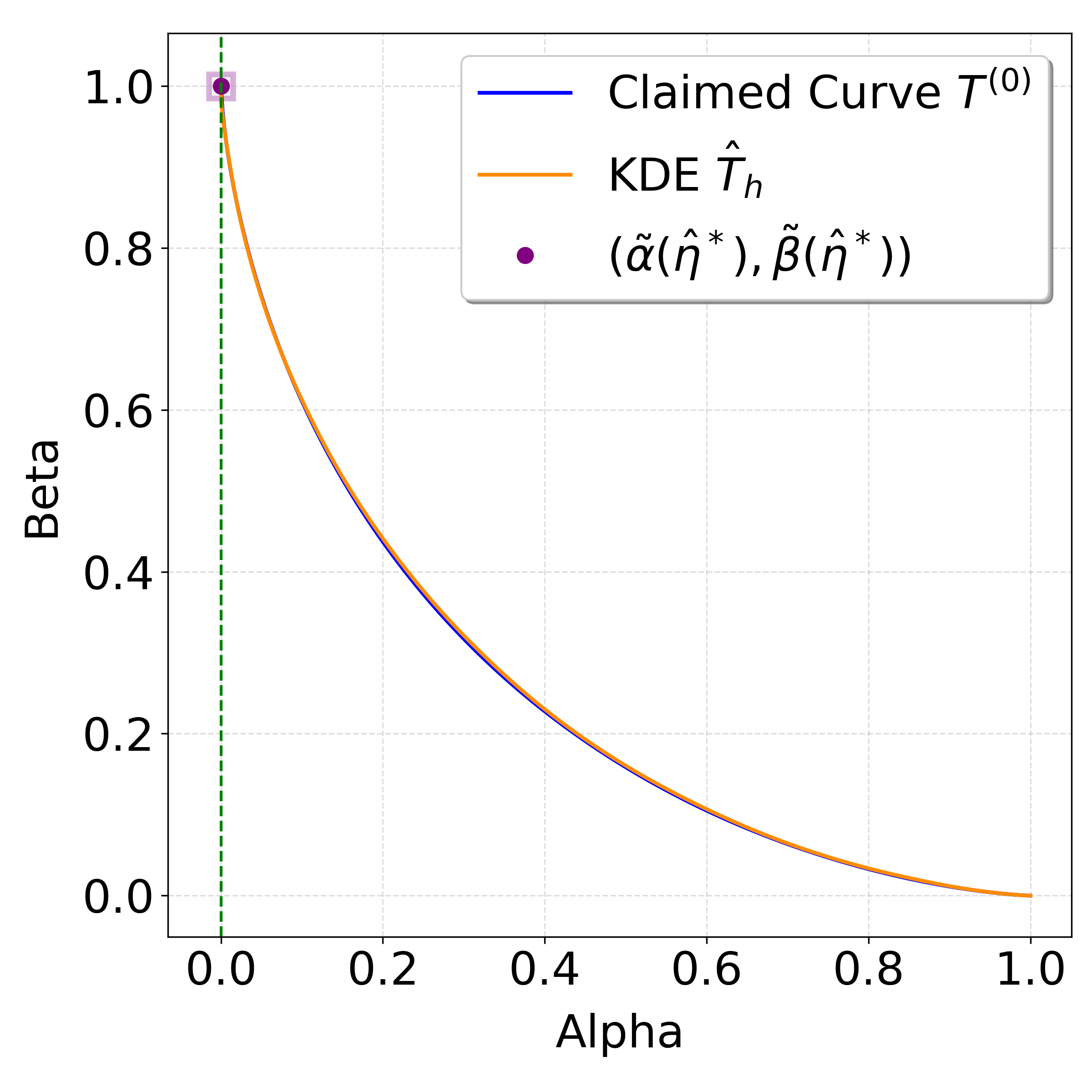}}
    \hfill
    \subfloat[\centering  \B{$n_2=10^5$}, \textbf{ Ground truth:}No Violation; \newline \textbf{Decision:} \textcolor{green}{"No Violation"}{\textcolor{green}{\scalebox{1.5}{\ding{51}}}}]{\includegraphics[width=0.3\textwidth]{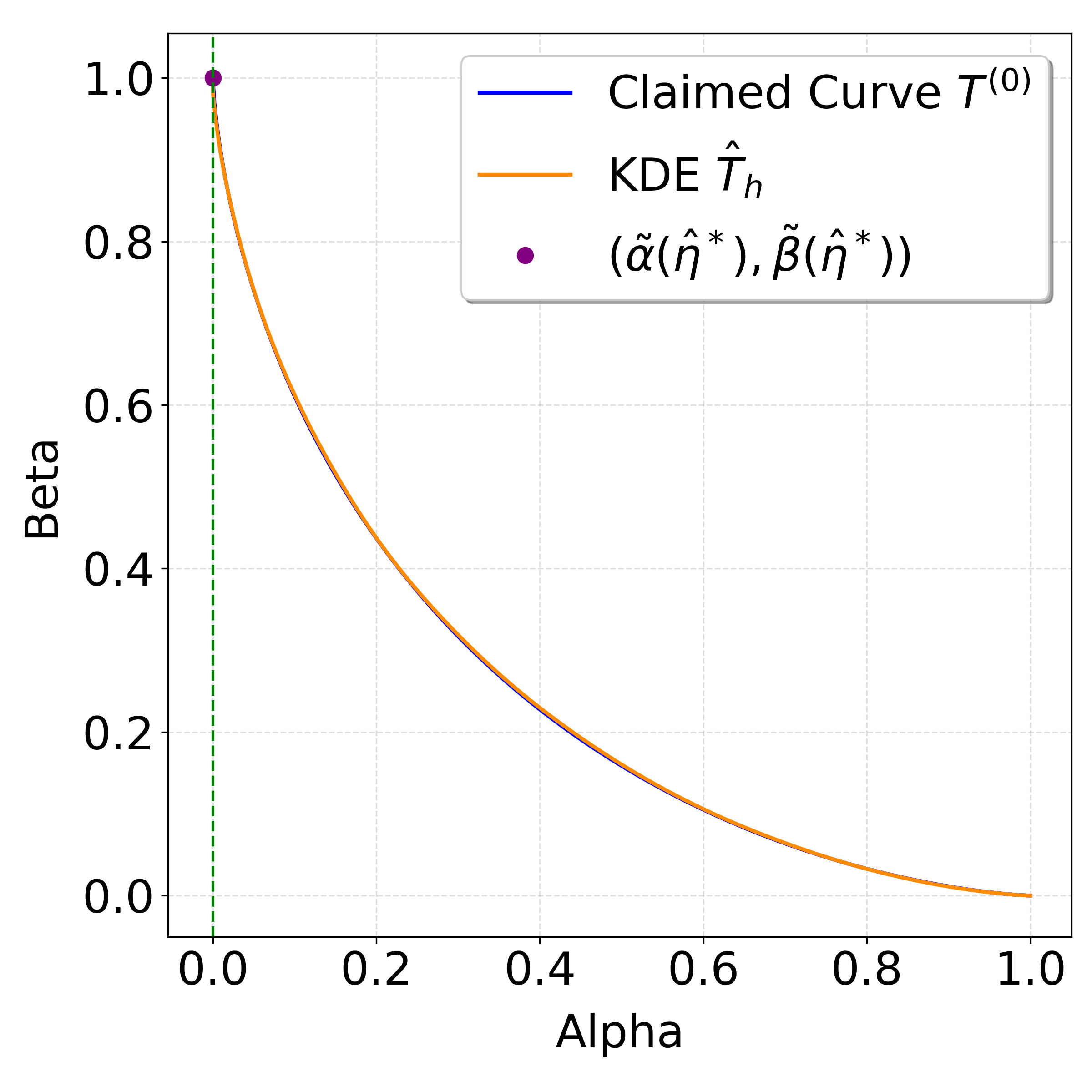}}
   \vspace{-1em}
    \subfloat[\centering  \B{$n_2=10^3$}, \textbf{Ground truth:} No Violation; \newline \textbf{Decision:} \textcolor{green}{"No Violation"}{\textcolor{green}{\scalebox{1.5}{\ding{51}}}}]{\includegraphics[width=0.3\textwidth]{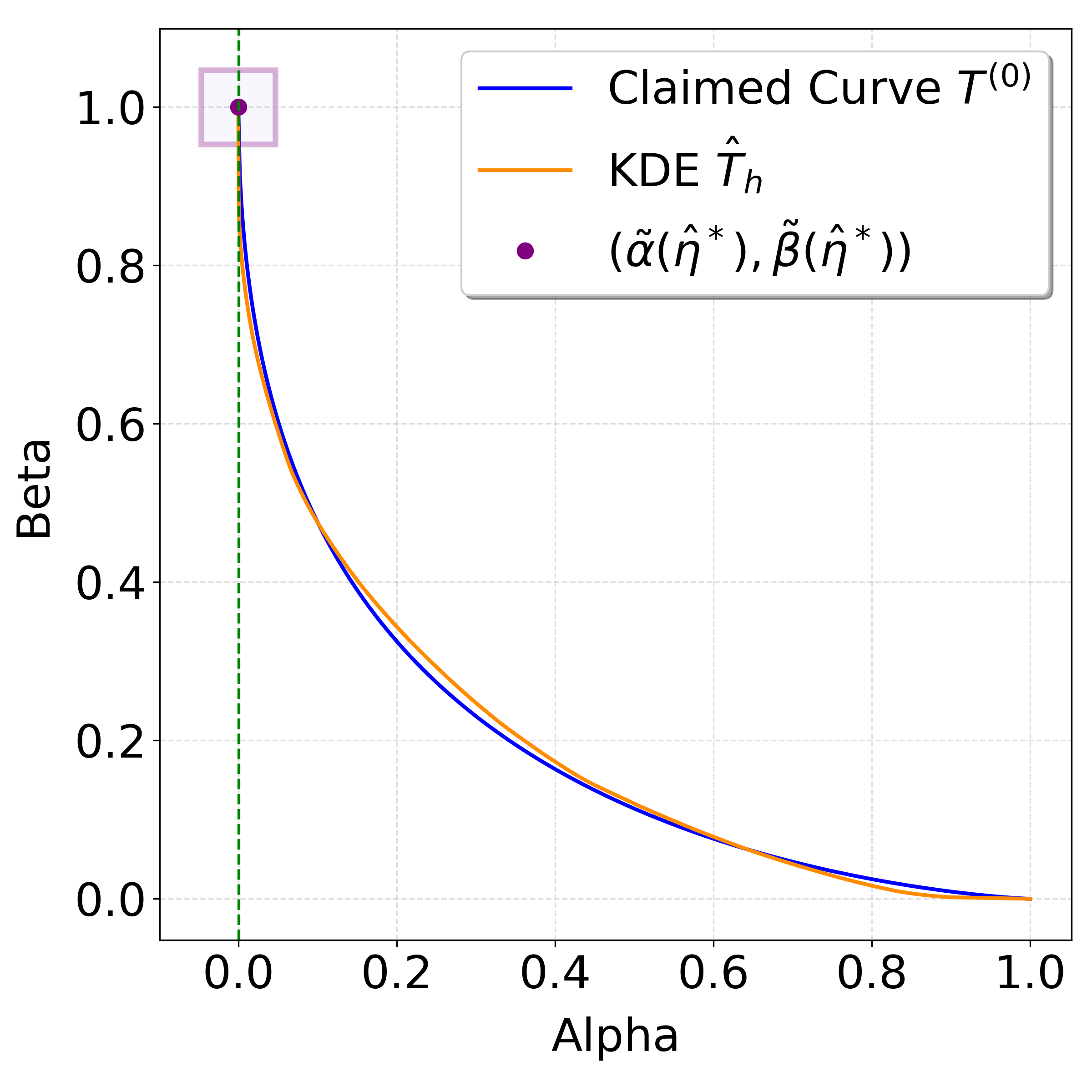}}
    \hfill
    \subfloat[\centering  \B{$n_2=10^4$}, \textbf{Ground truth:} No Violation; \newline \textbf{Decision:} \textcolor{green}{"No Violation"}{\textcolor{green}{\scalebox{1.5}{\ding{51}}}}]{\includegraphics[width=0.3\textwidth]{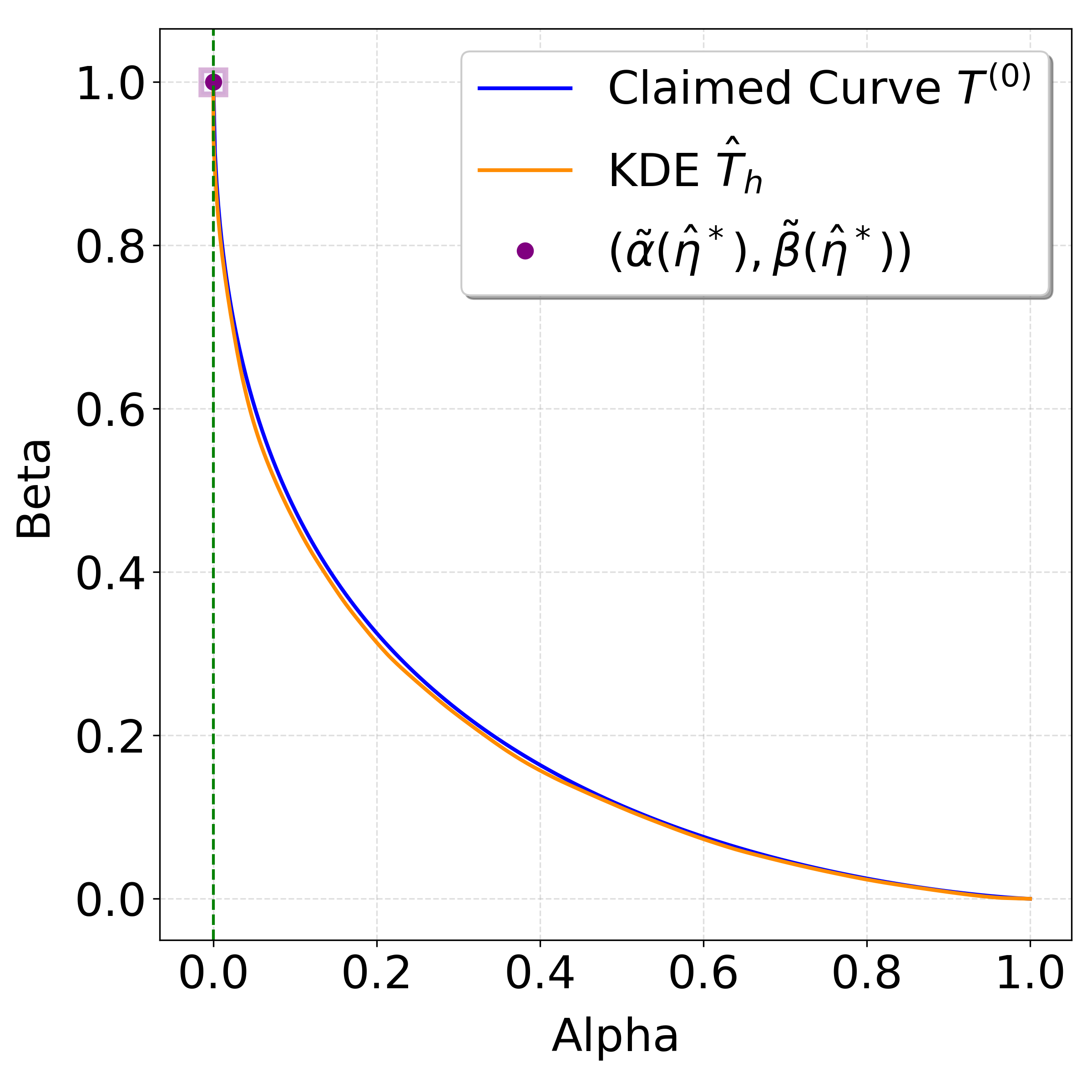}}
    \hfill
    \subfloat[\centering  \B{$n_2=10^5$}, \textbf{Ground truth:} No Violation; \newline \textbf{Decision:} \textcolor{green}{"No Violation"}{\textcolor{green}{\scalebox{1.5}{\ding{51}}}}]{\includegraphics[width=0.3\textwidth]{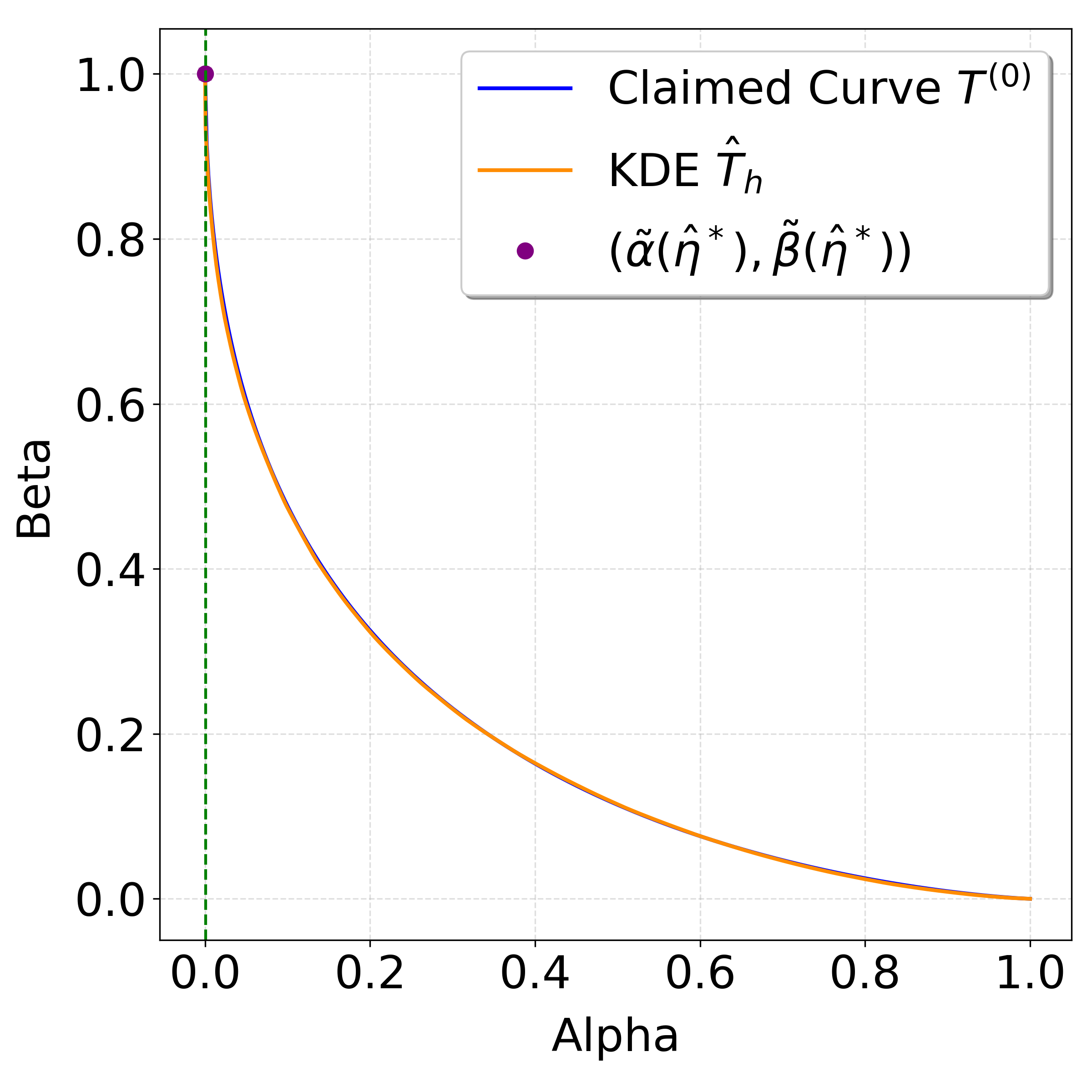}} \caption{\textbf{Auditing a correct Mechanism:} Claimed curve $\textcolor{blue}{T^{(0)}} = T_{Gauss}$ (a,b,c) and $\textcolor{blue}{T^{(0)}} = T_{SGD}$ (d,e,f). We depict the critical vertical line (obtained with step 3 in Algorithm \ref{alg:auditor}) with intercept $(\hat\alpha(\hat\eta^*), \hat \beta(\hat \eta^*))$, the $k$-NN point estimator {\textcolor{purple}{\ding{108}}} $(\tilde\alpha(\hat\eta^*), \tilde \beta(\hat\eta^*))$ and the confidence region $\textcolor{purple}{\square}$. The sample size for the KDE is $n_1=10^4$ and the confidence parameter is $\gamma=0.05$.}
    \label{fig:not_faulty_sgd_gauss}
\end{figure*}
\begin{figure*}[t]
    \centering
    \subfloat[\centering \B{$n_2=10^3$}, \textbf{Ground truth:} Violation; \newline \textbf{Decision:} \textcolor{red}{"No Violation"}{\textcolor{red}{\scalebox{1.5}{\ding{55}}}}]{\includegraphics[width=0.3\textwidth]{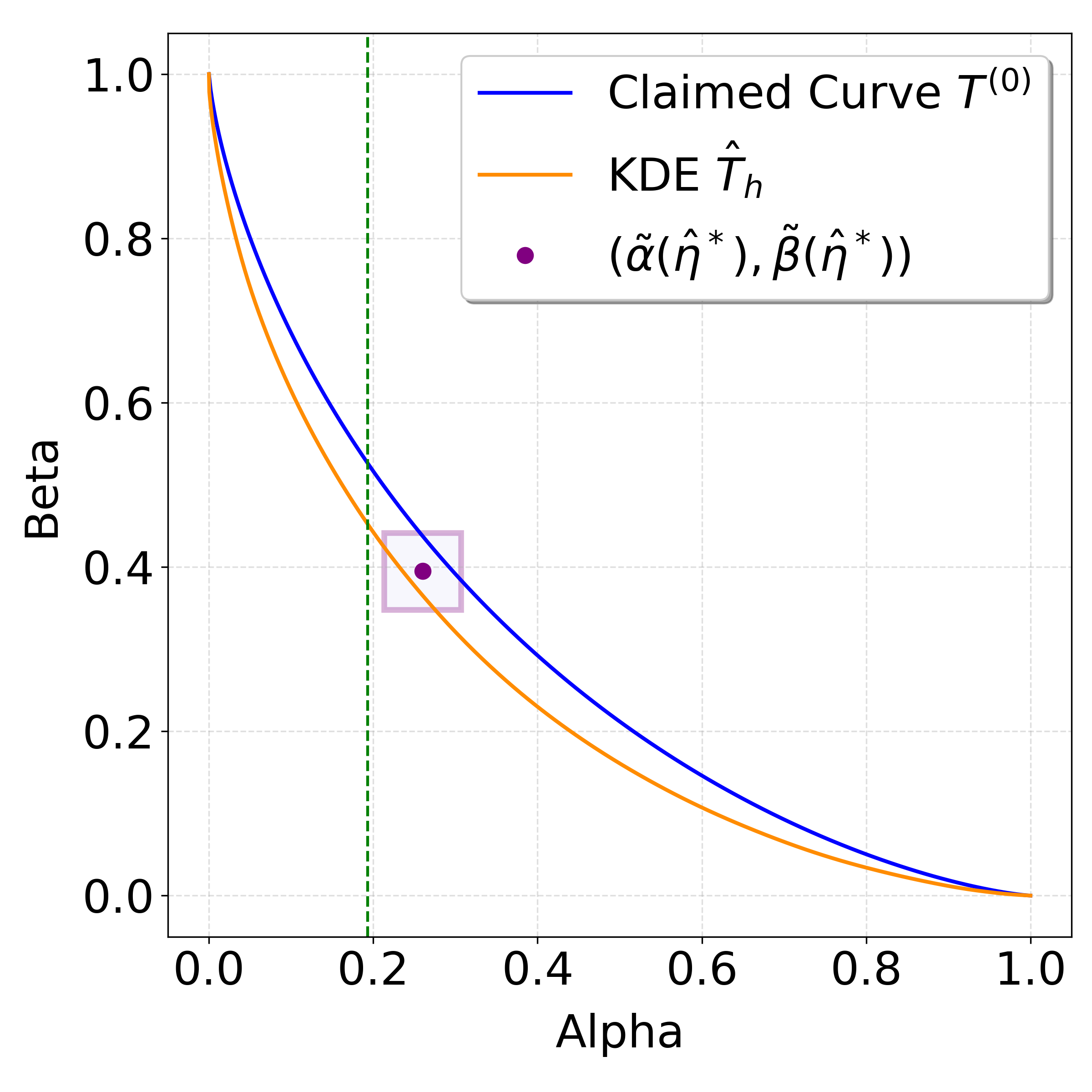}}
    \hfill
    \subfloat[\centering \B{$n_2=10^4$}, \textbf{Ground truth:} Violation; \newline \textbf{Decision:} \textcolor{green}{"Violation"}{\textcolor{green}{\scalebox{1.5}{\ding{51}}}}]{\includegraphics[width=0.3\textwidth]{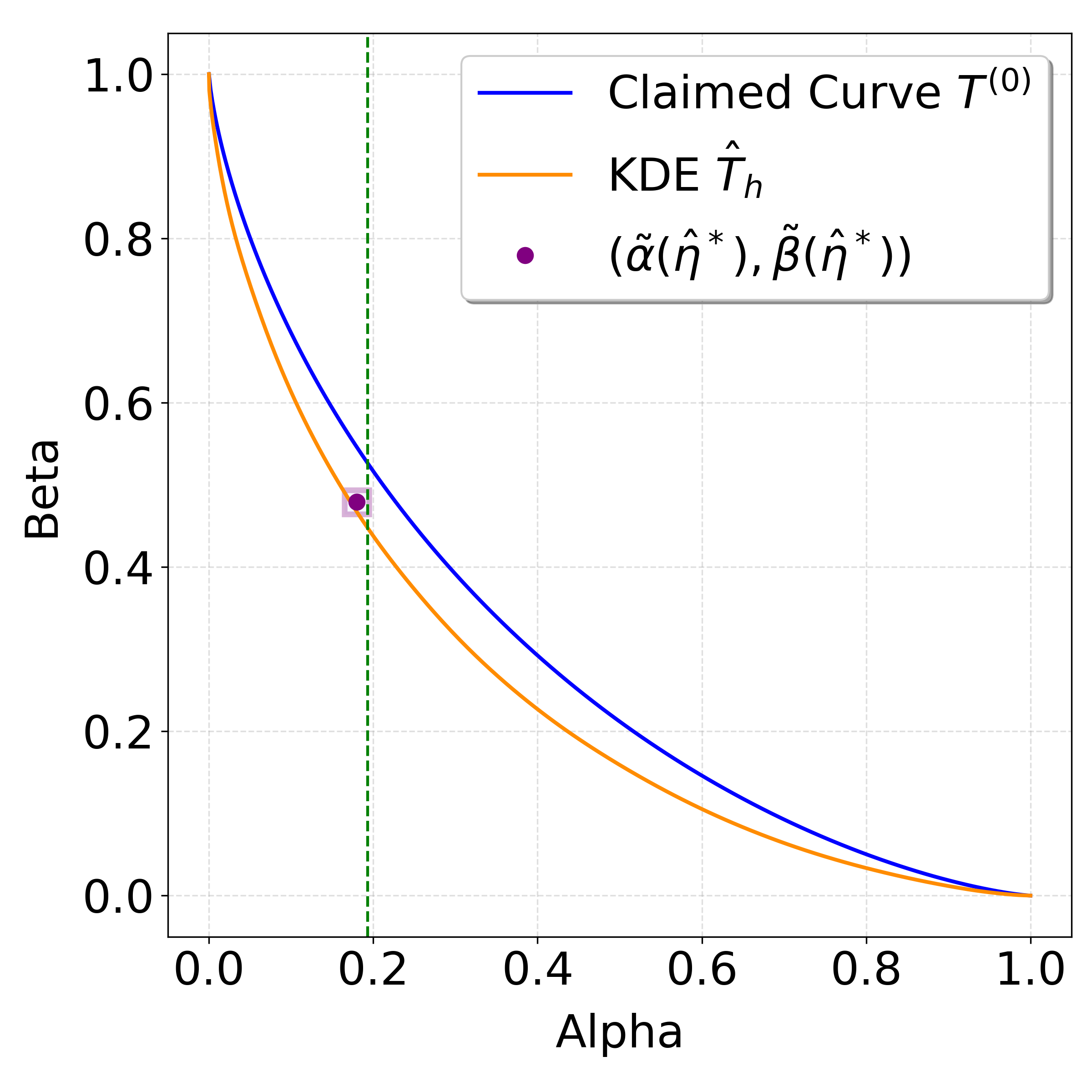}}
    \hfill
    \subfloat[\centering  \B{$n_2=10^5$}, \textbf{Ground truth:} Violation; \newline \textbf{Decision:} \textcolor{green}{"Violation"}{\textcolor{green}{\scalebox{1.5}{\ding{51}}}}]{\includegraphics[width=0.3\textwidth]{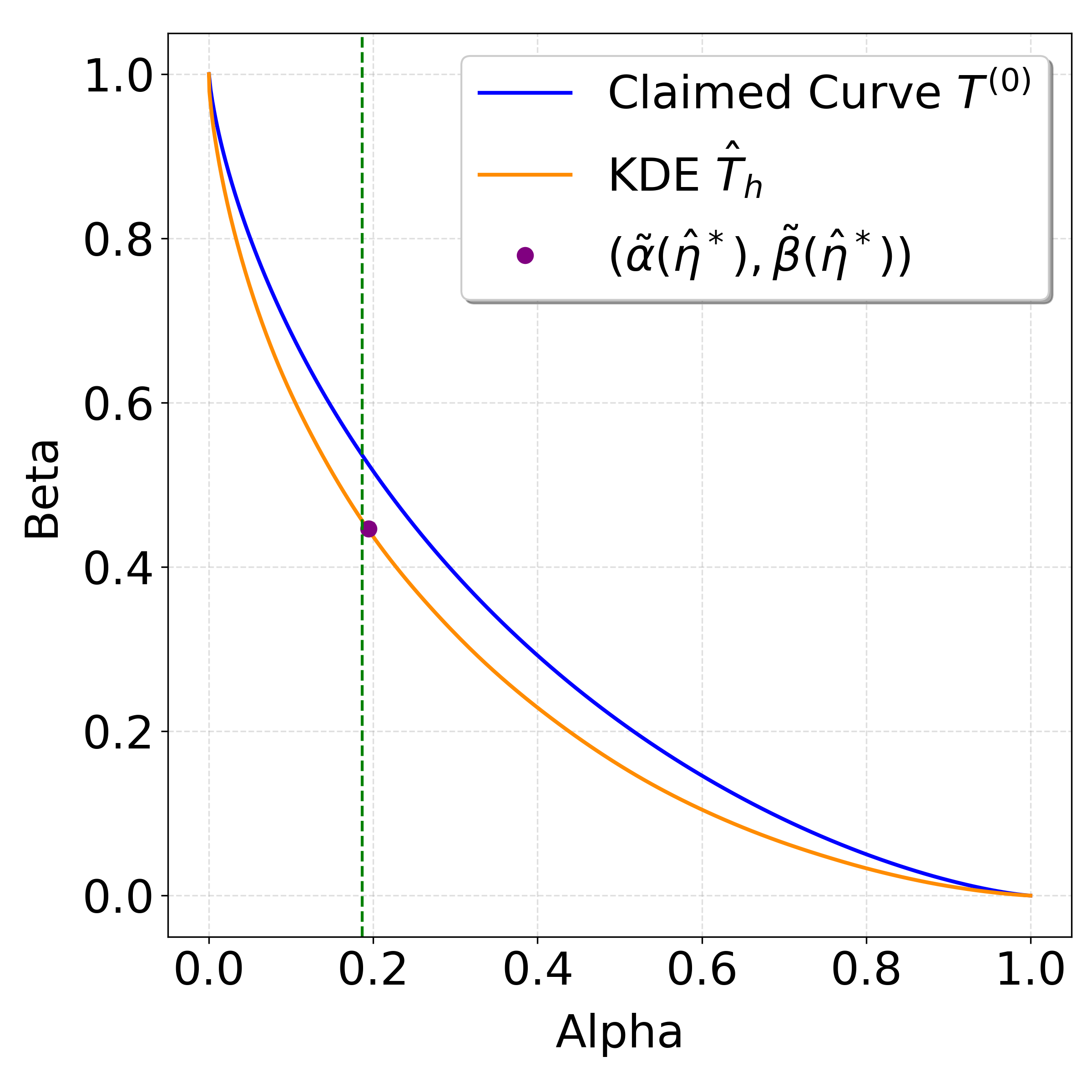}}
    \vspace{-1em}
    \subfloat[\centering \B{$n_2=10^3$}, \textbf{Ground truth:} Violation; \newline \textbf{Decision:} \textcolor{red}{"No Violation"}{\textcolor{red}{\scalebox{1.5}{\ding{55}}}}]
    {\includegraphics[width=0.3\textwidth]{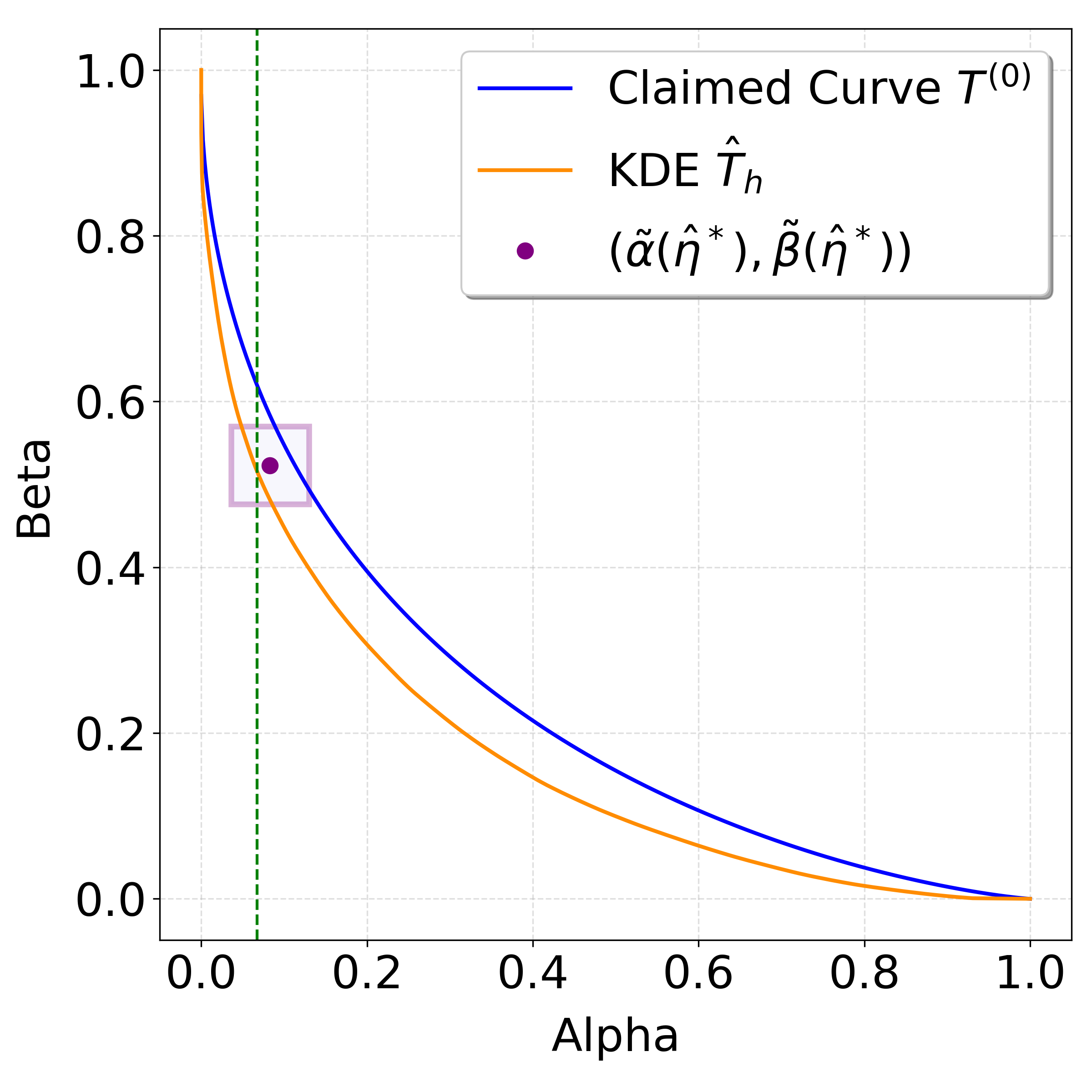}}
    \hfill
    \subfloat[\centering \B{$n_2=10^4$}, \textbf{Ground truth:} Violation; \newline \textbf{Decision:} \textcolor{green}{"Violation"}{\textcolor{green}{\scalebox{1.5}{\ding{51}}}}]{\includegraphics[width=0.3\textwidth]{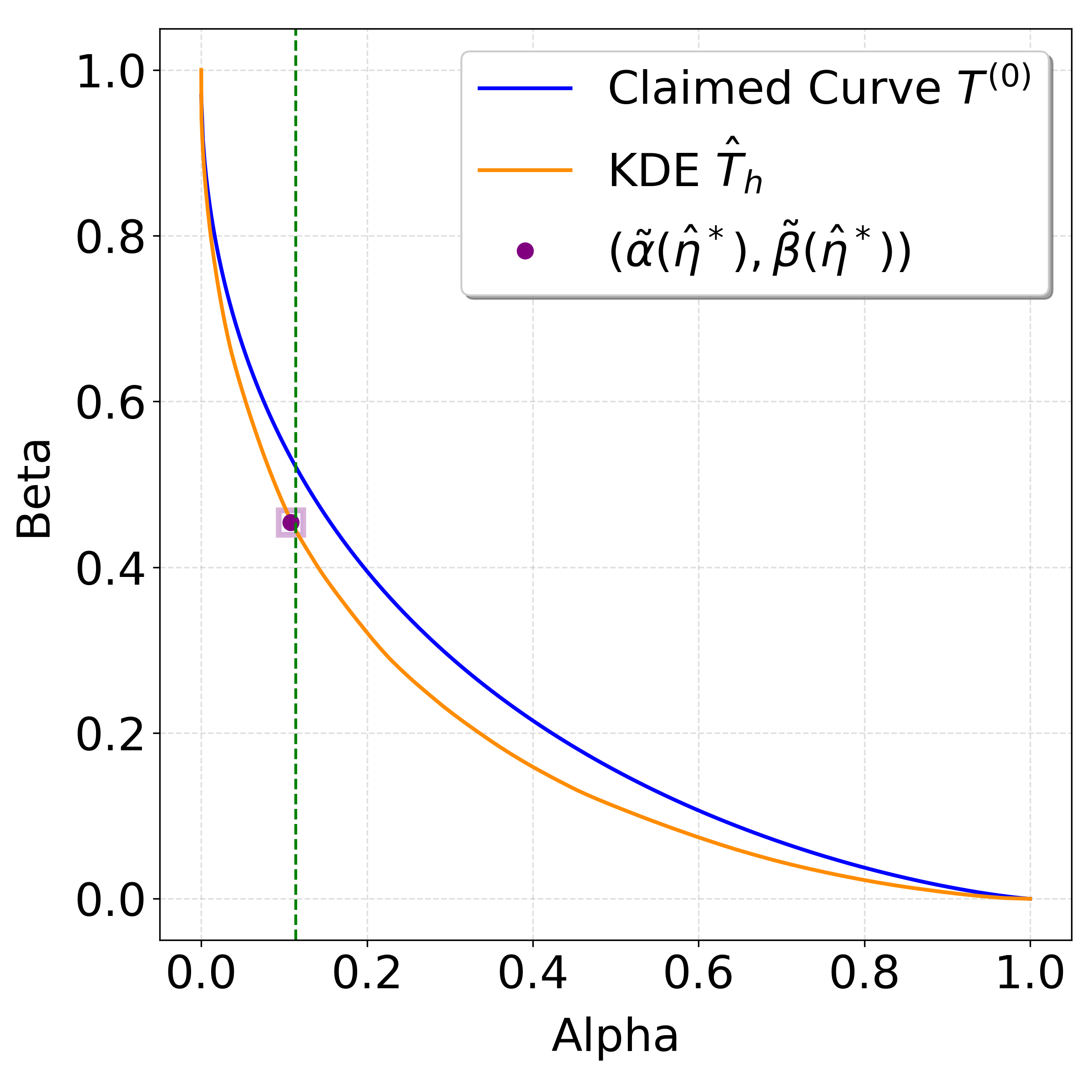}}
    \hfill
    \subfloat[\centering \B{$n_2=10^5$}, \textbf{Ground truth:} Violation; \newline \textbf{Decision:} \textcolor{green}{"Violation"}{\textcolor{green}{\scalebox{1.5}{\ding{51}}}}]{\includegraphics[width=0.3\textwidth]{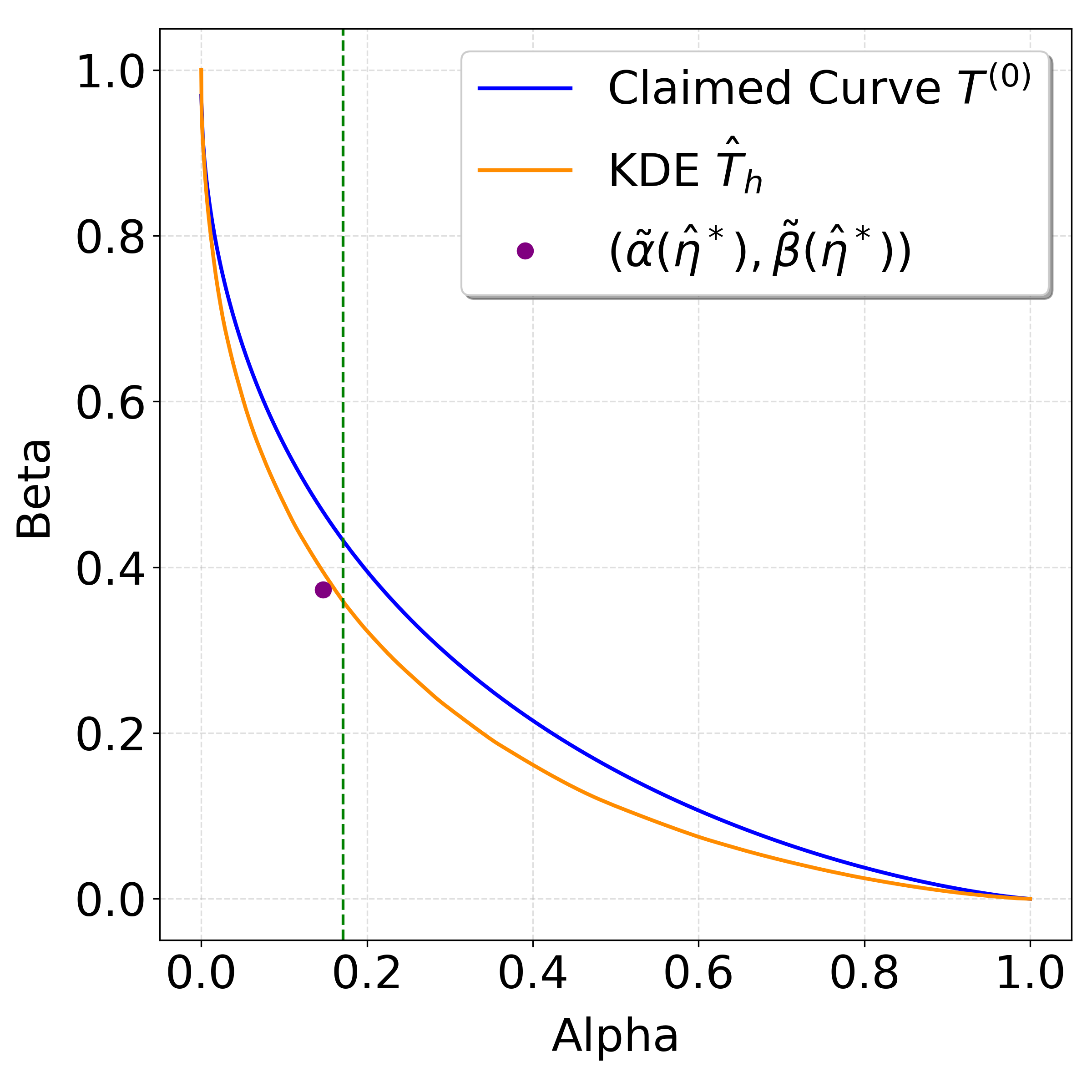}}
     \caption{\textbf{Auditing a faulty Mechanism:} Claimed Curve $\textcolor{blue}{T^{(0)}} = T_{Gauss}$ (a,b,c) with $\mu=0.5$ and $\textcolor{blue}{T^{(0)}} = T_{SGD}$ (d,e,f) with $\tau =5$. Both mechanisms assume stronger privacy ($\mu=0.5<1$ and $\tau =5<10$). We depict the critical vertical line (obtained with step 3 in Algorithm \ref{alg:auditor}) with intercept $(\hat\alpha(\hat\eta^*), \hat \beta(\hat \eta^*))$, the $k$-NN point estimator {\textcolor{purple}{\ding{108}}} $(\tilde\alpha(\hat\eta^*), \tilde \beta(\hat\eta^*))$ and the confidence region $\textcolor{purple}{\square}$. The sample size for KDE is $n_1=10^4$ and the confidence parameter is $\gamma=0.05$.}
    \label{fig:faulty_sgd_gauss}
\end{figure*}

\noindent\textbf{Implementation Details} The implementation is done using python and R. \footnote{\scriptsize{\url{https://github.com/stoneboat/fdp-estimation}}}. 
For the simulations, we have used a local device and a server. All runtimes were collected on a local device with an Intel Core i5-1135G7 processor (2.40 GHz), 16 GB of memory, and running Ubuntu 22.04.5, averaged over $10$ simulations. Thus, we demonstrate reasonable runtimes even on a standard personal computer (see Appendix \ref{app:add_simulations}). 
Additionally, we used a server with four AMD EPYC 7763 64-Core (3.5 GHz) processors and 2 TB of memory running Ubuntu 22.04.4 for repetitive simulations. For python, we have used Python 3.10.12 and the libraries "numpy" \cite{2020NumPy-Array}, "scikit-learn" \cite{pedregosa2011scikit} and "scipy" \cite{2020SciPy-NMeth}. For R, we used R version 4.3.1 and the libraries "fdrtool" \cite{fdrtool} and "Kernsmooth" \cite{Kernsmooth}.

\subsection{Interpretation of the results}
%refine our understanding of certain details of our methods. 
For Goal 1 (estimation), we see in Figure \ref{fig:estimation_mse} the fast decay of the estimation error of $\hat T_h$ for the optimal trade-off curve. The estimation error decays quickly in $n_1$, regardless of whether there are plateau values in the sense of Assumption \ref{ass1} (e.g. Laplace Mechanism) or not (e.g. Gaussian Mechanism).
These quantitative results are supplemented by the visualizations in  
Figures~\ref{fig:gaussian}--\ref{fig:sgd}, where we depict the largest distance of $\hat T_h$ and $T$ in $1000$ simulation runs (captured by the red band). Even for the modest sample size of $n_1 = 10^3$, this band is fairly tight and for $n_1 = 10^5$ the estimation error is almost too minute to plot. We find this convergence astonishingly fast. It may be partly explained by the estimator $\hat T_h$ being structurally similar to $T$ -  after all $\hat T_h$ is also designed to be a trade-off curve for an almost optimal LR test.
The approximation over the entire unit interval corresponds to the uniform convergence guarantee in Theorem~\ref{theo:1}. 

For Goal 2 (inference), we recall that a  $T^{(0)}$-DP violation is detected if the box $\square_\gamma$ (purple) lies completely below the postulated curve $T^{(0)}$ (blue). In Figure \ref{fig:not_faulty_sgd_gauss}, we consider the case of no violation where $T=T^{(0)}$, and we expect not to detect a violation. This is indeed what happens, since $\square_\gamma$ intersects with the curve $T^{(0)}$ in all considered cases. Interestingly, we observe that $\square_\gamma$ has a center close to $\alpha=0$ in the cases where no violation occurs (such a behavior might give additional visual evidence to users that no violation occurs).
%In principal, one would reject the privacy curve, whenever the purple square $\textcolor{purple}{\square}$ is disjoint from the blue \textcolor{blue}{curve}, i.e.
%\begin{equation*}
%    \textcolor{purple}{\square} \cap \textcolor{blue}{\textnormal{curve}}=\emptyset~.
%\end{equation*}
%In Figure \ref{fig:not_faulty_sgd_gauss} and \ref{fig:not_faulty_sgd_gauss} we have displayed the case where the claimed privacy indeed holds, so we expect to not detect a violation. For both mechanism, we can observe that for any sample size, we correctly do not reject that claim, as the $\textcolor{purple}{\square}$ and the $\textcolor{blue}{curve}$ are not disjoint. 
In Figure \ref{fig:faulty_sgd_gauss}, we display the case of faulty claims, where the privacy breach is caused by a smaller variance for both mechanisms under investigation. In accordance with Theorem \ref{theo:auditor}, we expect a detection of a violation if $n_2$ is large enough. This is indeed what happens, at a sample size of \B{$n_2=10^4$} for the Gaussian mechanism and at \B{$n_2=10^4$} for DP-SGD. \B{Note} that larger samples $n_2$ are needed to expose claims $T^{(0)}$ that are closer to the truth $T$ (as for DP-SGD in our example). For larger $n_2$ the square $\square_\gamma$ shrinks (see eq. \eqref{e:defsq}), leading to a higher resolution of the auditor. 

\subsection{\B{Real-World Example - CIFAR-10}}
\B{We consider the CIFAR-10 dataset and train a small private convolutional neural network using Opacus\footnote{\B{\scriptsize{\url{https://github.com/pytorch/opacus}}}} \cite{Opacus}, a standard library for training PyTorch models with DP.\\
\textbf{Parameter settings:}
To simulate a pair of neighboring databases, we take the same subset of size 1,000 of the CIFAR-10 training data and replace the index 0 by differing images. In $D$, we use an all-black (all-zero) synthetic image, while in group $D'$, we use an all-white (all-255) image. Both are labeled arbitrarily as "airplane". We train a 4-layer convolutional neural network on 1,000 images, using a batch size of $m=512$ (so in total 2 batches) for 1,10,15,20 and 25 epochs. 
We set the learning rate to $\rho=0.1$, the clipping parameter to $1.0$ and use $\sigma = 1.0$ for the noise multiplier. We train the model 1,000 times on $D$ and $D'$ respectively, which took approximately 7 hours on a machine with 64 CPU cores while training in parallel.\\
%For the stochastic gradient descent we use a learning rate of $\rho=0.1$. For the privacy parameter and the clipping parameter, we use the standard settings, i.e. noise multiplier of $\sigma = 1.0$ and clipping at $1.0$. We have trained the model 1,000 times on $D$ and $D'$ respectively. For the experiments, we train this 2,000 models in parallel in approximately 7 hours on a machine with 64 CPU cores. This setup can be scaled to the full CIFAR-10 training dataset (50,000 samples) with a linear increase in runtime or by using more powerful hardware.\\
\textbf{Score:} %Using the above setup, it is unclear on what our estimator is build on.% 
For classical neural network structures, taking a reasonable score function will be essential for proper auditing. In image classification, assigning the correct label to the distinct image can yield reasonable results. In our black-box setting, we tried different scores, namely the logits, the cross-entropy (CE), with which the model is trained, and the Kullback Leibler divergence (KL). It is also crucial for our methodology that the score function is indeed one dimensional, as the KDEs performance heavily degrades for higher dimensions. %Consequently, a simple comparisons of the parameter vectors of the models is infeasible. %
In Figure \ref{fig:estimation_cnn} and Figure \ref{fig:auditing_cnn}, we depict the performance of estimation and auditing after 25 epochs. For auditing we have used the KL loss. In Figure \ref{fig:lower_bound_cnn}, we illustrate how tight the lower bounds are over various epochs, which yield different theoretical $\varepsilon$.\\}
\begin{figure}[t]
\centering\includegraphics[width=0.7\linewidth]{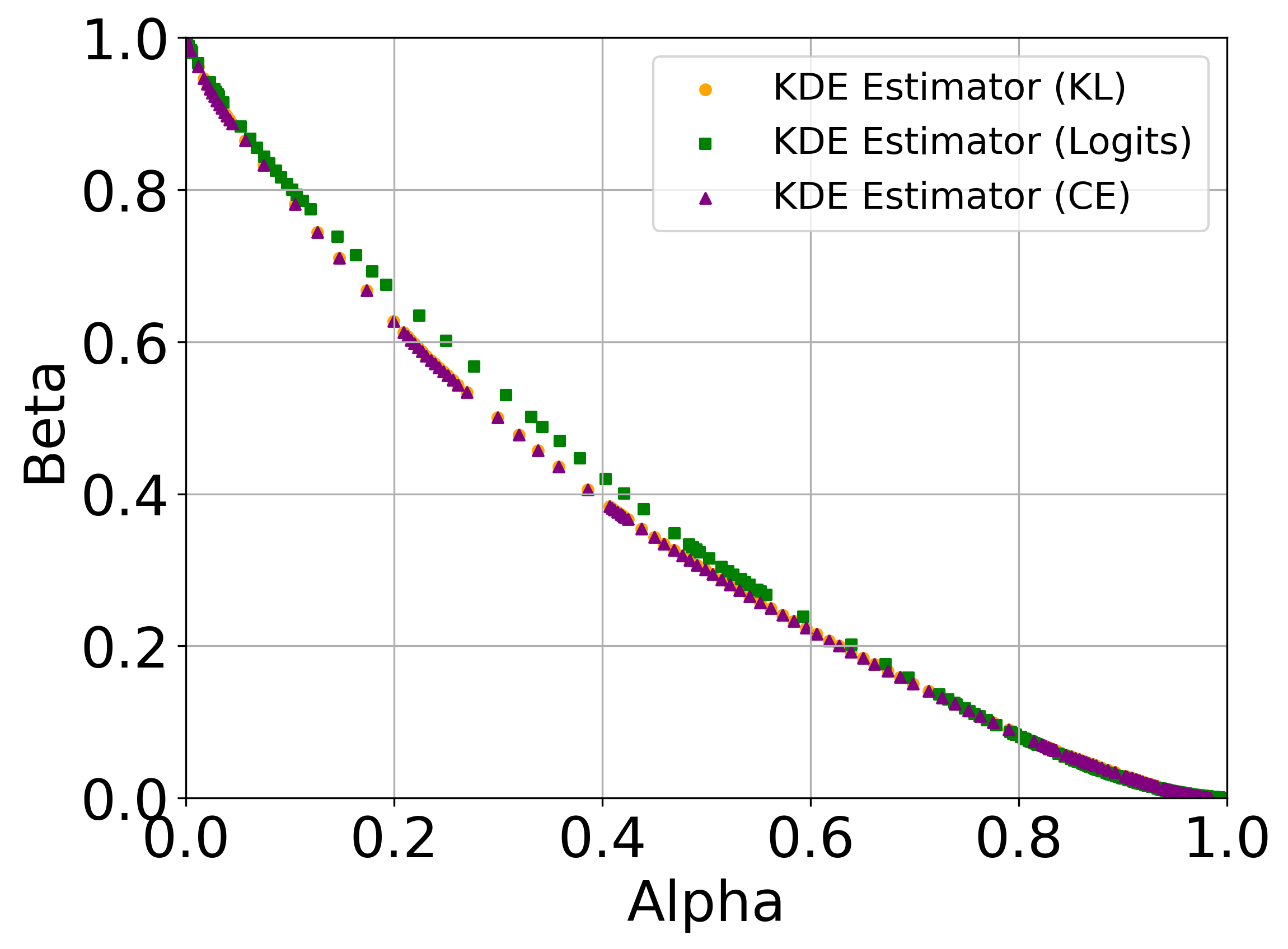}
    \caption{\centering \B{Algorithm \ref{alg:pointwise_KDE_estimator} based on 1,000 models for different loss functions. KL and CE are overlapping. }}\label{fig:estimation_cnn}
\end{figure}
\begin{figure}[t]
\centering\includegraphics[width=0.7\linewidth]{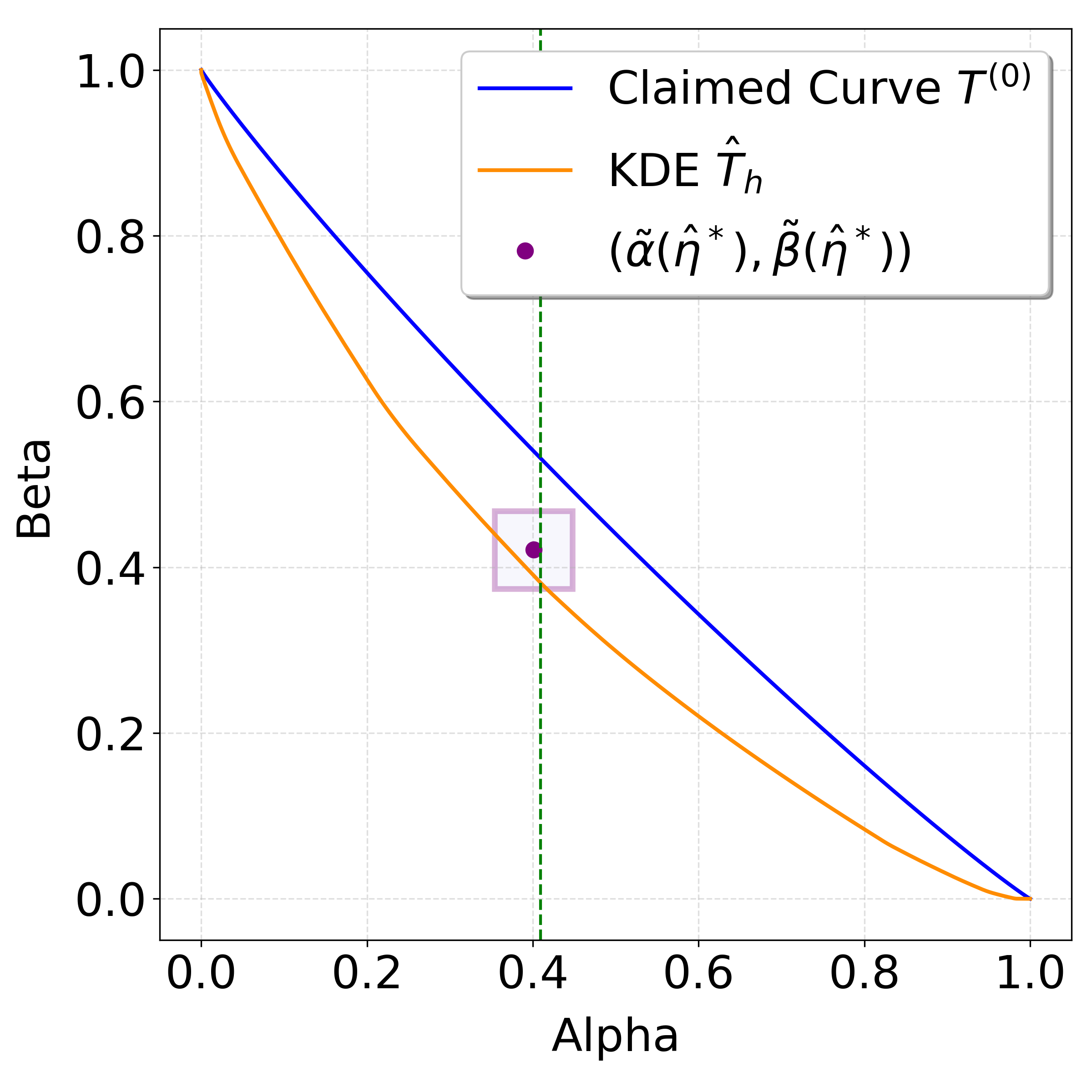}
    \caption{\centering \B{Algorithm \ref{alg:auditor} based on 1,000 models using the KL. $T^{(0)}$ is a gaussian trade-off curve with $\mu=0.15$.  \textbf{Decision:} "Violation".}}\label{fig:auditing_cnn}
\end{figure}
\begin{figure}
\centering\includegraphics[width=0.7\linewidth]{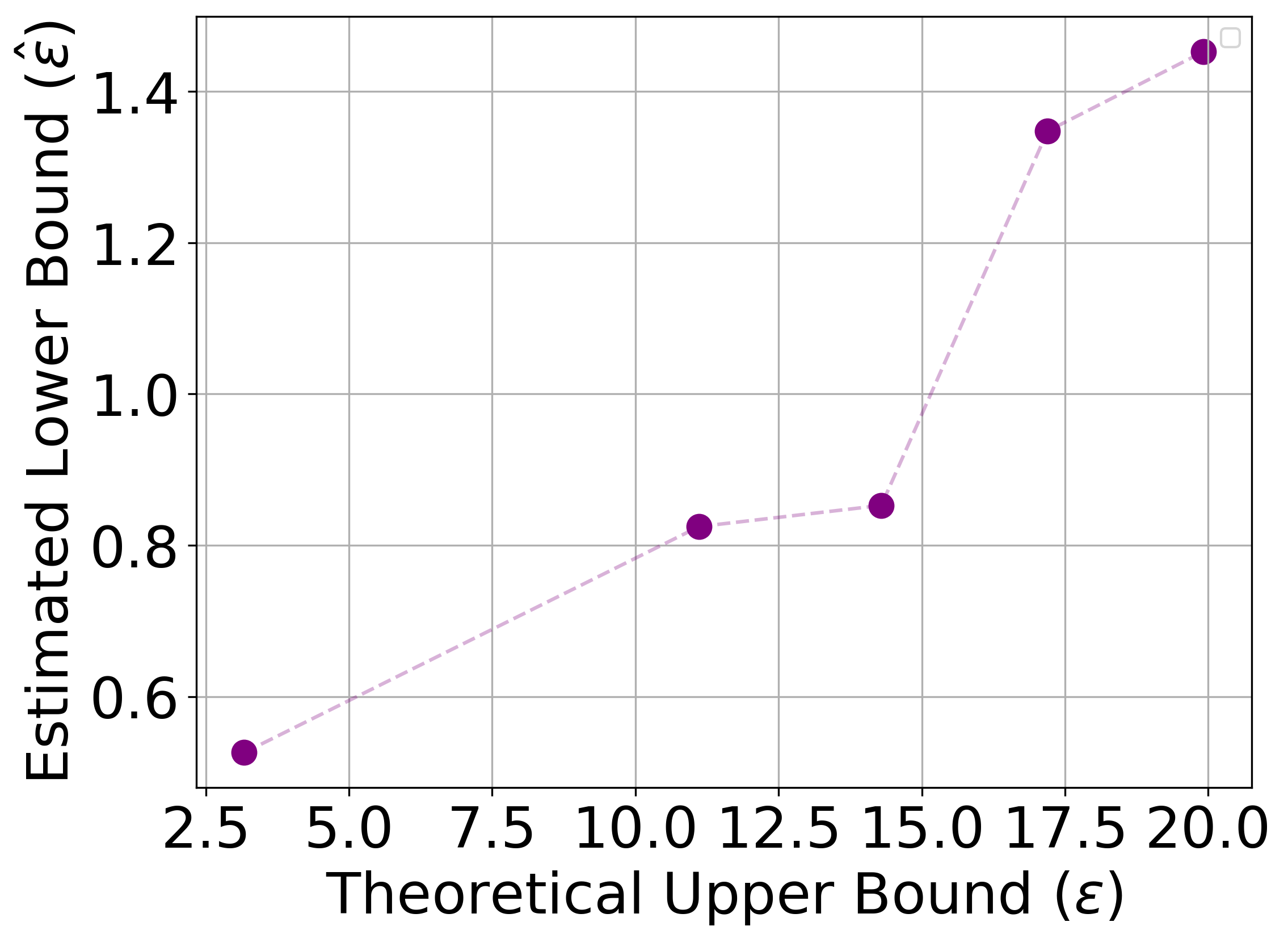}
    \caption{\centering \B{Comparison of theoretical $\varepsilon$ and empirical lower bound for $\delta=0.001$.}}\label{fig:lower_bound_cnn}
\end{figure}
\noindent\B{\textbf{Interpretation of results:} In Figure \ref{fig:estimation_cnn}, we compare several scoring functions for estimation. As expected, %Cross-Entropy% 
CE performs best, since the model is trained to minimize it. Nonetheless, related losses (e.g., KL) perform similarly, underscoring the black-box nature of our method. In principle, one could try multiple scores and select the most promising. For auditing (see Figure \ref{fig:auditing_cnn}), we consider a misspecified training process, %e.g., the blue curve $T^{(0)}$%
which can result from incorrect noise calibration, excessive training epochs, or misconfigured clipping. %To wrap things up,%
We also %consider%
compare the empirical lower bound $\hat{\varepsilon}$ and theoretical $\varepsilon$ obtained using Opacus. We compute the empirical lower bound from the estimated trade-off curve. Specifically, we use binary search to find the Gaussian trade-off function that is closest to our estimated curve in terms of $\ell_1$ distance. We then leverage the known relationship between approximate differential privacy and the Gaussian trade-off function to convert this into a lower bound on $\varepsilon$. %For the empirical lower bound, we used an ad-hoc method, by minimizing the $\ell_1$ distance to a $\mu-$GDP curve and obtaining the $\varepsilon$ through that relation. 
Similar to \cite{Nasr2023}, we see that the estimates are loose (see Figure \ref{fig:lower_bound_cnn}). Consequently, using the proposed method for auditing is only possible for relatively strong privacy violations. More accurate models, on the other hand, allow for tighter estimates an better auditing, as already highlighted in \cite{Nasr2023}.}

\section{Related Work} \label{sec:relatedwork}

%{\color{orange}We discuss the techniques used in related work, and why they were not sufficient in our goal to construct a general $f$-DP estimator and auditor. 
%First, given a known connection between $f$-DP and DP, one possible avenue to studying $f$-DP guarantees is to make use of an existing DP estimator or auditor. 
%Concretely~\cite{Dong2022}, a $(\epsilon,\delta)$-DP 
% algorithm $M$ is also $f_{\varepsilon, \delta}$-DP with %trade-off function 
%\begin{align} \label{f_epsilon_delta}
%    f_{\varepsilon, \delta}(\alpha) := \max \left\{ 0, 1 - \delta - e^{\varepsilon} \, \alpha, e^{- \varepsilon} (1 -\delta - \alpha) \right\}.
%\end{align}

%}

%{In this section, we provide a more detailed overview of and comparison with related works that focus on the empirical assessment of $f$-DP.}

\B{In this section, we provide a more detailed overview of related works on auditing privacy guarantees and f-DP.} One avenue to assessing $f$-DP is to resort to a method that provides estimates for the $(\varepsilon,\delta)$-parameter of $M$ and to subsequently exploit the link between standard and $f$-differential privacy to obtain an estimate of $f$. To be more concrete, an algorithm that is $(\varepsilon,\delta)$-DP is also $f_{\varepsilon, \delta}$-DP (see \cite{Dong2022}) with trade-off function 
\begin{align} \label{f_epsilon_delta}
    f_{\varepsilon, \delta}(\alpha) := \max \left\{ 0, 1 - \delta - e^{\varepsilon} \, \alpha, e^{- \varepsilon} (1 -\delta - \alpha) \right\}.
\end{align} 
Thus, an estimator for $(\varepsilon, \delta)$ could, in principle, also provide an estimate for the trade-off curve \B{$f$} of $M$. \B{This approach is pursued in \cite{Koskela2024} with the help of a black-box estimator of $\delta$ for fixed $\varepsilon > 0$. The maximum over the $f_{\varepsilon, \delta}(\cdot)$-estimates then provides an approximation of the trade-off curve of $M$. In contrast, our auditing procedure is based on a single, direct estimate of the trade-off curve of $M$, which makes our approach more expedient.} In fact, the runtimes reported for estimation of $f$ in \B{Appendix \ref{app:add_simulations}} confirm the efficacy of our approach. Moreover, from an auditing perspective, results with regard to convergence and reliability in \cite{Koskela2024} are only obtained for the \B{$\delta(\epsilon)$}-estimate in the standard DP framework. Our work, on the other hand, provides formal statistical guarantees for the inference of the trade-off $f$.

Interestingly, the relation between standard and $f$-DP can also be exploited in the opposite direction, that is, to use estimates of the trade-off curve $f$ to obtain estimates for $(\varepsilon, \delta)$. \B{This approach was first taken in \cite{Nasr2023} and subsequently adopted in other works \cite{Annamalai2024, Annamalai2024_B,Mahloujifar2024} for the purpose of auditing the privacy claims of DP-SGD, a cornerstone of differentially private machine learning. Essentially, these works aim at auditing DP by converting confidence intervals for the type-I and type-II error of a distinguishing attack into estimates of $(\varepsilon,\delta)$. And while previous works \cite{Jagielski_2020, zanella2023bayesian, Steinke_2023} investigated the type-I and type-II errors of distinguishing attacks in the standard DP model to obtain such estimates, \cite{Nasr2023} was the first to exploit the tight characterization of these errors in the Gaussian DP model to obtain even tighter lower bounds for $\varepsilon$.}
\B{Auditing in \cite{Nasr2023} focuses on the white-box scenario, where the auditor does not only have access to the training datasets $D$ and $D'$, but can also examine all intermediate model updates that go into computing the corresponding final models $\theta$ and $\theta'$. This setting is further enhanced by allowing the auditor to actively intervene in the training process via self-crafted gradients or datasets that can be inserted into the computations that yield the final model outputs. \say{Opening the black-box} in this manner results in tighter empirical estimates with fewer observations.} \B{A black-box scenario discussed in \cite{Nasr2023} restricts the available information about the instance of DP-SGD under investigation. In this setting, the auditor's access is limited to the training datasets, the corresponding final models and the specific loss function $\ell$ that the training algorithm uses. Important parameters and features such as model weights, noise scales, sampling rates, learning rates and etc. remain hidden from the auditor. This scenario can thus be considered \say{parameter black-box}. Though far more restrictive than the above white-box setting, this scenario is also specific to the DP-SGD mechanism, which is characterized by the repeated use of Gaussian noise. It therefore allows for auditing procedures that work with the class of (subsampled) Gaussian trade-off curves as in \cite{Nasr2023} or for distinguishing attacks such as LiRA \cite{Carlini_2022}, which models the distribution of losses as Gaussians.} \B{Subsequent works have adopted the approach in \cite{Nasr2023} to study differentially private synthetic data generation  \cite{Annamalai2024_C}, the impact of shuffling on the privacy of DP-SGD \cite{Annamalai2024_B}, how to obtain tighter audits in the black-box setting of \cite{Nasr2023} with specially crafted worst-case parameters \cite{Annamalai2024}, or to improve the analysis of distinguishing attacks \cite{Mahloujifar2024}. }
\B{In our work, we further tighten the parameter black-box scenario in \cite{Nasr2023} by avoiding assumptions that a specific mechanism or family of trade-off curves are under investigation. 
Our approach can thus be deemed \say{mechanism black-box} and is more aligned with a number of existing works that study other common variants of DP through this lens \cite{DP-Sniper, Dette2022, gorla2023impossibility, Lu2024, Kutta2024,Kong2024}. With no knowledge of the underlying privacy mechanism, we cannot assume normally distributed output data in this setting. Hence, we developed tools like the perturbed LR test and BayBox estimator, which do not require modeling distributions as Gaussians and instead rely on mild regularity assumptions (like smoothness) that are common in the mechanism black-box setting \cite{Dette2022,gorla2023impossibility,Lu2024}.}
\B{Even though it was designed for the more restrictive mechanism black-box scenario, our method's performance is similar to that of the black-box approach from \cite{Nasr2023} in our experiments on more realistic datasets. Moreover, our approach compares favorably to other works that operate in the same mechanism black-box setting. Here, the required number of output samples to achieve reasonable levels of accuracy usually surpasses the maximum sample size $n = 10^5$ in our experiments on auditing \cite{Dette2022,gorla2023impossibility, Kong2024} and can even reach into the millions \cite{DP-Sniper, Kutta2024, Lu2024}. Thus, our estimation and auditing methods are effective and flexible tools that add to the existing literature on DP validation.}

\section{Conclusion \B{and Future Work}}\label{sec:summary_discussion}

In our work, we construct the first general-purpose $f$-DP estimator and auditor \B{for the mechanism} black-box setting, by \B{combining} techniques from statistics and classification theory. Our constructions enjoy not only formal guarantees---convergence of our estimator and a tuneable confidence region for our auditor---but also \B{perform well in experiments on standard algorithms from the DP literature}. \B{Our approach has limitations as well. In practice, training even a single machine learning model can be costly, making it impractical to sample thousands. This is a problem in black-box settings, which require larger samples sizes to achieve desired levels of accuracy.}

\B{However, our approach benefits from a plug-and-play design, allowing these limitations to be mitigated by substituting alternative estimators -- such as neural network-based estimators -- that might be more effective in high-dimensional settings.
%for both the perturbed LR test and the $k$-NN classifier. %
Replacing the $k$-NN
classifier does not affect the theoretical guarantees provided in Theorem~\ref{theo:auditor}, although it may 
%cause larger %
increase the chance of failing to reject a false claim if the alternative 
%classifier %
is not well chosen. For the PLRT component, adopting a density estimator that satisfies Assumption~\ref{ass2} preserves all guarantees established in Theorem~\ref{theo:1}; otherwise, the guarantees may no longer hold. In general, designing improved estimators and classification tools for the black-box scenario is a promising direction for future work.}

 \section{Acknowledgments}
 \noindent This work was funded by the Deutsche Forschungsgemeinschaft (DFG, German Research Foundation) under Germany’s
 Excellence Strategy - EXC 2092 CASA - 390781972. Tim Kutta was partially funded by the AUFF Nova Grant 47222. Yun Lu is supported by NSERC (RGPIN-03642-2023). Vassilis Zikas and Yu Wei are supported in part by Sunday Group, Incorporated. 
Holger Dette was partially supported by the DFG Research unit 5381 {\it Mathematical Statistics in the Information Age}, project number 460867398. 
\newpage
 \section*{Ethics considerations}

 We attest that the algorithms our privacy auditor/estimator identify as having privacy vulnerabilities are already recognized for (or specifically created with) these issues. Thus, this work does not expose new vulnerabilities. Additionally, we do not make use of private datasets in our experiments. Thus, our experiment results do not introduce privacy leaks. We attest that our auditor/estimator do not have any biases stemming from a conflict of interest. While our auditor can be used on other existing mechanisms, our implementation mitigates risks by ensuring our output is limited to alerts to potential privacy violation, and does not leak any additional information about the dataset(s) tested.

 \section*{Open science}

 To comply with the Open Science Policy, we will make all artifacts publicly available\footnote{\scriptsize{\url{https://doi.org/10.5281/zenodo.15599462}}}. In our experiments section, we  ensure the transparency and reproducibility of our methodology by describing the dataset used, the specification of our machines, and citing all algorithms tested.

{\footnotesize
\bibliographystyle{acm}

}

\appendix
\section{Appendix}
%The appendix is dedicated to the proofs and technical details of our results. Throughout our proofs we will use the notation $R= o_P(1)$ for a remainder $R$ that satisfies $R \overset{P}{\to}0$ (convergence in probability).
The appendix is dedicated to the technical details of our results. The proofs can be found in an extended arxiv version.
\begin{table}[ht]
\centering
\caption{Overview of Notation Used in the Paper}
\label{tab:notation}
\begin{tabular}{ll}
\toprule
\textbf{Notation} & \textbf{Description} \\
\midrule
\( D,D'\) & Pair of adjacent databases \\
\( M \) & ($f$-)DP Mechanism \\
\( \Pr{}[ ], \Ex{ } \) & Probability, Expectation\\
\( P,Q \) & Output distributions of $M(D), M(D')$ \\
\(\MixtureD{P}{\eta}\) & Mixture distribution with parameter $\eta$\\
\( p,q\) & Probability densities of $P,Q$  \\
\( \alpha, \beta \) & type-I \& type-II errors \\
& (typically of the Neyman-Pearson test) \\
\(\hat \alpha_h, \hat \beta_h\) & Estimated errors using KDE \\
\(\tilde \alpha, \tilde \beta\) & Estimated errors using $k$-NN \\
& (typically of the Neyman-Pearson test) \\
\( T\) & optimal trade-off curve for $p,q$\\
\( T^{(0)}\) & trade-off curve that is audited\\
\( T_h\) & trade-off curve of perturbed LR test \\
\( \hat T_h\) & estimated trade-off curve using KDE \\
\(  \eta\) & threshold in LR tests \\
& vulnerability \\
\( \hat \eta^*\) & estimated threshold of maximum \\
& vulnerability \\
\(  \lambda\) & randomization parameter in \\
 & Neyman-Pearson test \\
 \(  h\) & randomization parameter in \\
 & perturbed LR test \\
\( \phi, \kNNclassifier{n} \) & generic classifier, $k$-NN classifier \\
\( \phi^* \) & Bayes optimal classifier \\
\( \gamma, w(\gamma) \) & confidence level \& 
margin of error\\
$\square_\gamma$ & confidence region for \\
 & type-I-type-II errors\\
 $n, n_1, n_2$ & sample size parameters \\
\bottomrule
\end{tabular}
\end{table}
\section{Additional Experiments and Details} \label{AppB}

In this section, we provide some additional details on our experiments and implementations.

\subsection{Implementation details}

Algorithm \ref{alg:pointwise_KDE_estimator} gives a pseudo-code of our trade-off curve estimator $\hat T_h$, presented in Section \ref{sec:4}. 

\begin{algorithm}[h]
\footnotesize
\algorithmicrequire \; \parbox[t]{\dimexpr0.9\linewidth-\algorithmicindent}{Black-box access to $\Mech$; Threshold $\eta > 0$; Sample size $n$, databases $\DB, \DB'$.}\\[0.1cm]
\algorithmicensure \, An estimate $(\hat{\alpha}(\eta), \hat{\beta}(\eta))$ of $(\alpha(\eta), \beta(\eta))$ for tuple $(P, Q)$, where $\Mech(\DB)$ and $\Mech(\DB')$ are distributed according to $P, Q$, respectively.
\begin{algorithmic}[1]
    \State Choose perturbation parameter $h$. 
    \State Set the density estimation algorithm $\cA$. By default, use the KDE algorithm.
    \Function{\textnormal{\B{PLRT Estimator}} $\ptlr{h}{\cA}(M, \DB, \DB', \eta,n)$}{}
    \State Compute the estimated densities $\hat{p}$ and $\hat{q}$ by running $\mathcal{A}$ on $n$ independent copys of $\Mech(\DB)$ and $\Mech(\DB')$, respectively.
    % based on outputs of $M$ by running $\cA$ with a sample size of $n$.
    \State Compute $\hat{\alpha}(\eta) \leftarrow \int_{x \in [-h/2,h/2]} \frac{1}{h}\int_{\hat q /\hat p  > \eta +x} \hat p$ 
    \State Compute $\hat{\beta}(\eta) \leftarrow 1 -\int_{x \in [-h/2,h/2]} \frac{1}{h}  \int_{\hat q /\hat p  > \eta +x} \hat q$ 
    \State Return $(\hat{\alpha}(\eta), \hat{\beta}(\eta))$
    \EndFunction
\end{algorithmic}
\caption{\B{PLRT}: A Perturbed Likelihood Ratio Test Algorithm for $f$-DP Estimation}
\label{alg:pointwise_KDE_estimator}
\end{algorithm}

Next, we turn to the DP-SGD algorithm from our Experiments section. The pseudocode for that algorithm can be found in Algorithm \ref{alg:noisy_sgd} below. Note that we add Gaussian noise $Z_t \sim \mathcal{N}(0, \sigma^2)$ to the parameter $\theta_t$ at each iteration of DP-SGD. The operator $\Pi_{\Theta}$ projects the averaged and perturbed gradient onto the space $\Theta$ and is thus similar to clipping that gradient. We can derive the exact trade-off function of this algorithm for our choice of databases in \eqref{eq_databases} and our specifications from Section \ref{sec:algorithms}. More concretely, we first consider the distribution of DP-SGD on $D = (0, \dots, 0)$ and note that 
\begin{align*}
    \theta_{t+1} = \theta_t - \rho \, (\theta_t + Z_{t+1}) 
\end{align*}
for each $t \in \{0, \dots, \tau\}$. Some calculations then yield that $\Theta_{\tau} \sim \mathcal{N}(0, \bar{\sigma}^2)$ with 
\begin{align} \label{sigma_bar}
    \bar{\sigma}^2 = \rho^2 \, \sigma^2 \, \frac{1 - (1 - \rho)^{2 \tau}}{1 - (1 - \rho)^{2}}.
\end{align}
Similarly, we have for $D' = (1, 0, \dots, 0)$ that 
\begin{align*}
     \theta_{t+1} = (1 - \rho) \, \theta_t + \rho \,  Z_{t+1} 
\end{align*}
for each $t \in \{0, \dots, \tau\}$. Here, $Z_t$ is a Gaussian mixture with 
\begin{align*}
    Z_t \sim \frac{1}{2} \, \mathcal{N}\left(0,\sigma^2\right) + \frac{1}{2} \, \mathcal{N}\left(\frac{1}{m},\sigma^2\right).
\end{align*}
We can then see that $ \theta_{\tau} = \tilde{Z}_1 + \dots + \tilde{Z}_{\tau} $
where the $\tilde{Z}_t$ are independent Gaussian mixtures with
\begin{align*}
    \tilde{Z}_t & \sim \frac{1}{2} \, \mathcal{N}\left(0, \rho^2 \, (1- \rho)^{2 (\tau - t)} \, \sigma^2\right) \\ & + \frac{1}{2} \, \mathcal{N}\left(\frac{\rho (1 - \rho)^{\tau - t}}{m}, \rho^2 \, (1- \rho)^{2 (\tau - t)} \, \sigma^2\right).
\end{align*}
By defining 
\begin{align} \label{mu_I}
    \mu_I := \sum\limits_{t \in I} \frac{\rho (1 - \rho)^{\tau - t}}{m}
\end{align}
and choosing $\bar{\sigma}$ as in \eqref{sigma_bar}, we get that
\begin{align*}
    \theta_{\tau} \sim \sum\limits_{t \subset \{1, \dots, \tau\}} \frac{1}{2^{\tau}} \mathcal{N}(\mu_I, \bar{\sigma}^2).
\end{align*}
Having derived the distribution of $M(D)$ and $M(D')$, we take a look at the corresponding LR-test $g$ and note that it can be expressed as
\begin{align*}
    g(x) = \begin{cases}
       1  & x > c \\
       0  & x \leq c \\
    \end{cases}
\end{align*}
for some threshold $c$. A few calculations then yield the trade-off curve
\begin{align*}
    T_{SGD}(\alpha)=\sum_{I\subset \{1,\hdots, \tau\}} \frac{1}{2^{\tau}}\Phi\Big(\Phi^{-1} (1-\alpha)-\frac{\mu_I}{\bar\sigma}\Big)~.
\end{align*}

\begin{algorithm}[h]
\footnotesize
\algorithmicrequire \; \parbox[t]{\dimexpr0.9\linewidth-\algorithmicindent}{Dataset $x = (x_1, \ldots, x_r)$, loss function $\ell(\theta, x)$,\\ Parameters: initial state $\theta_0$, learning rate $\rho$, batch size $m$, time horizon $\tau$, noise scale $\sigma$, closed and convex space $\Theta$.}\\[0.1cm]
\algorithmicensure \, Final parameter $\theta_{\tau}$.
\begin{algorithmic}[1]
    \For{$t = 1, \ldots, \tau$}
        \State \textbf{Subsampling:} Take a uniformly random subsample $I_t \subseteq \{1, \ldots, r\}$ with batch size $m$.
        \For{$i \in I_t$}
            \State \textbf{Compute gradient:} $v_t^{(i)} \leftarrow \nabla_\theta \ell(\theta_t, x_i)$
           % \State \textbf{Clip gradient:} $v_t^{(i)} \leftarrow v_t^{(i)} / \max\{1, \|v_t^{(i)}\|_2 / C\}$
        \EndFor
        \State \textbf{Average, perturb, and descend:}
        \[
        \theta_{t+1} \leftarrow \theta_t - \rho \; \Pi_{\Theta} \left( \frac{1}{m} \sum_{i \in I_t} v_t^{(i)} + Z_t \right)
        \]
    \EndFor
    \State \textbf{Output:} $\theta_{\tau}$
\end{algorithmic}
\caption{DP-SGD Algorithm}
\label{alg:noisy_sgd}
\end{algorithm}

\begin{figure*}[h!]
    \centering
    \subfloat[$n_1=10^3$]{\includegraphics[width=0.3\textwidth]{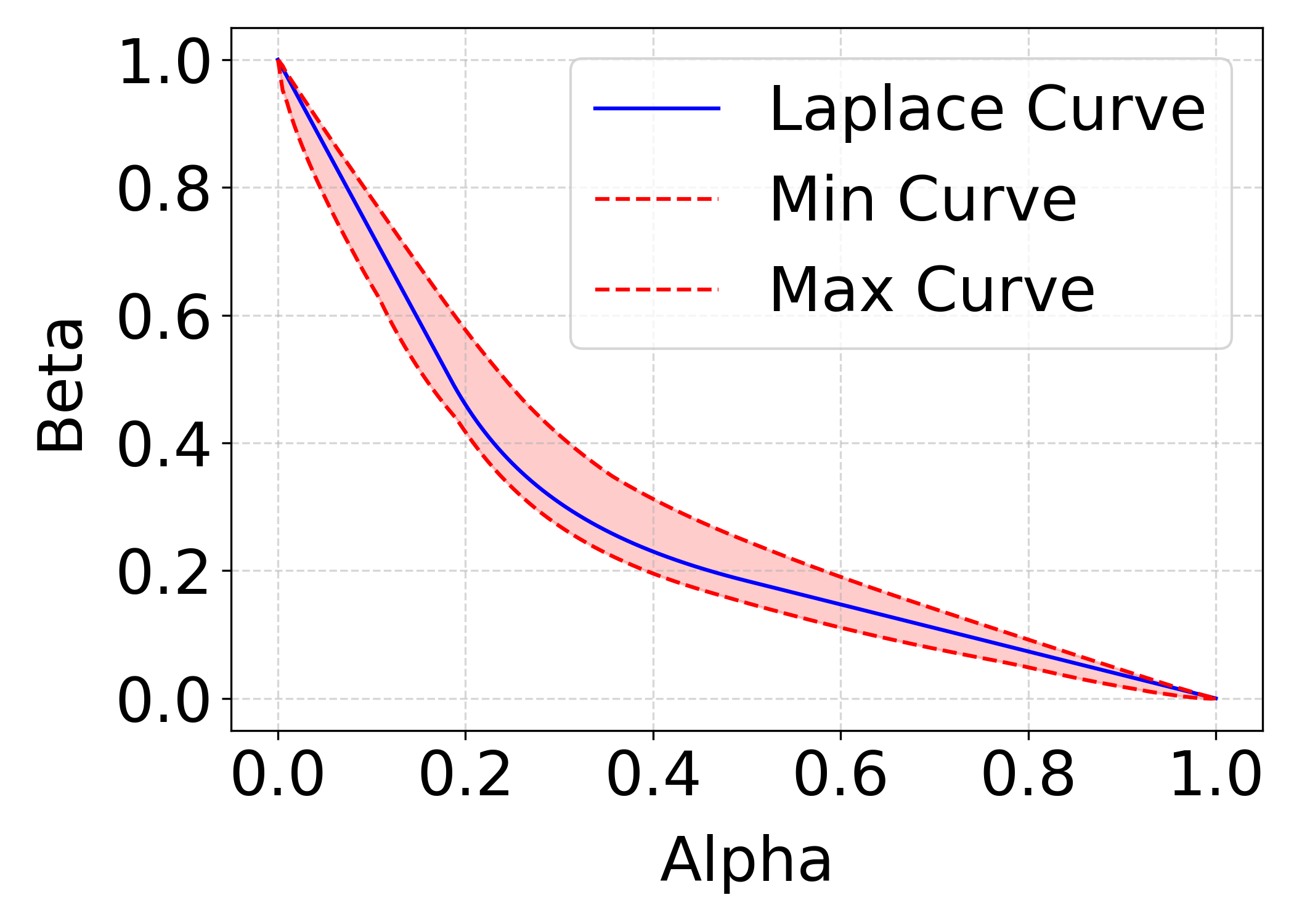}}
    \hfill
    \subfloat[$n_1=10^4$]{\includegraphics[width=0.3\textwidth]{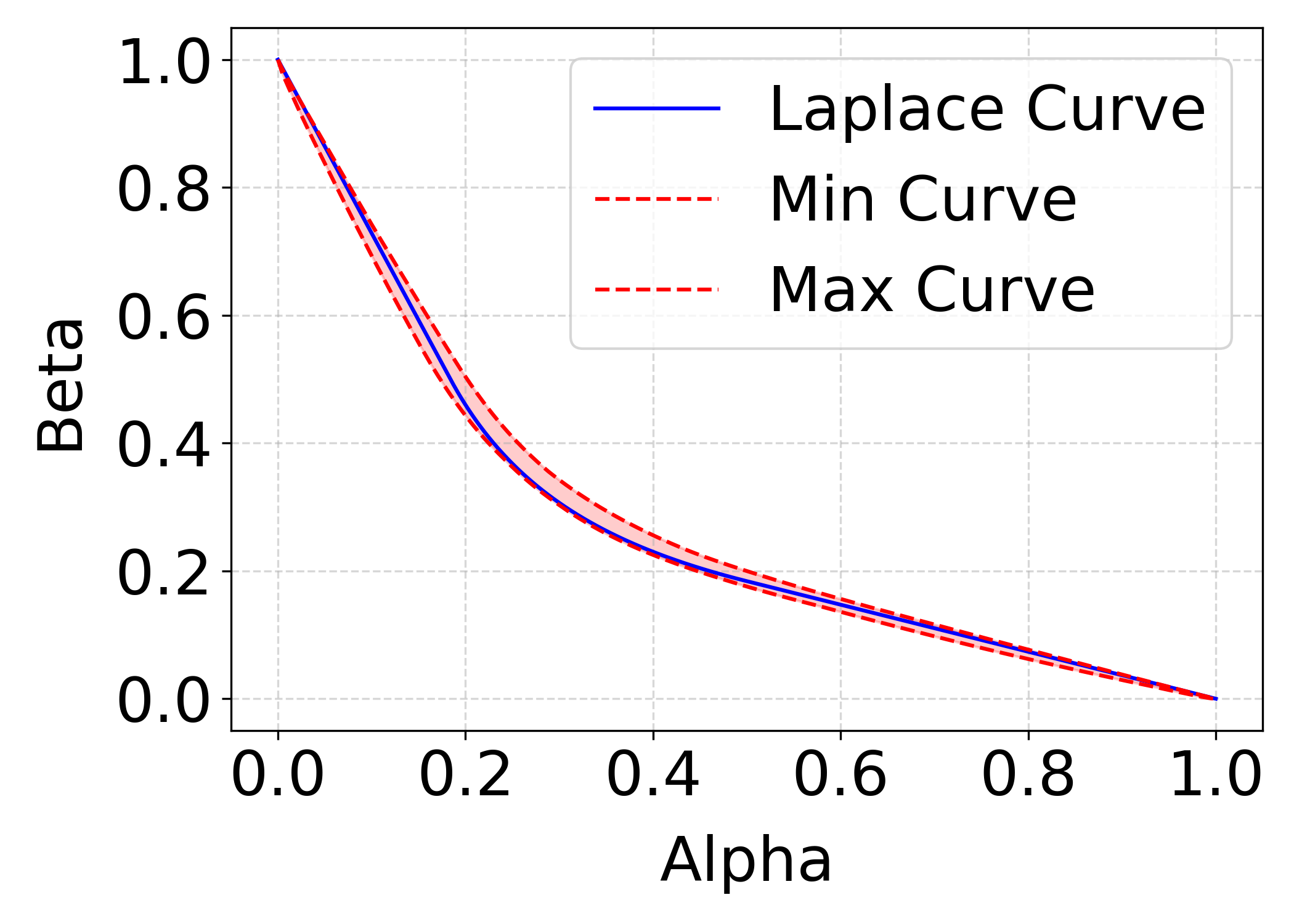}}
    \hfill
    \vspace{-0.2cm}
    \subfloat[$n_1=10^5$]{\includegraphics[width=0.3\textwidth]{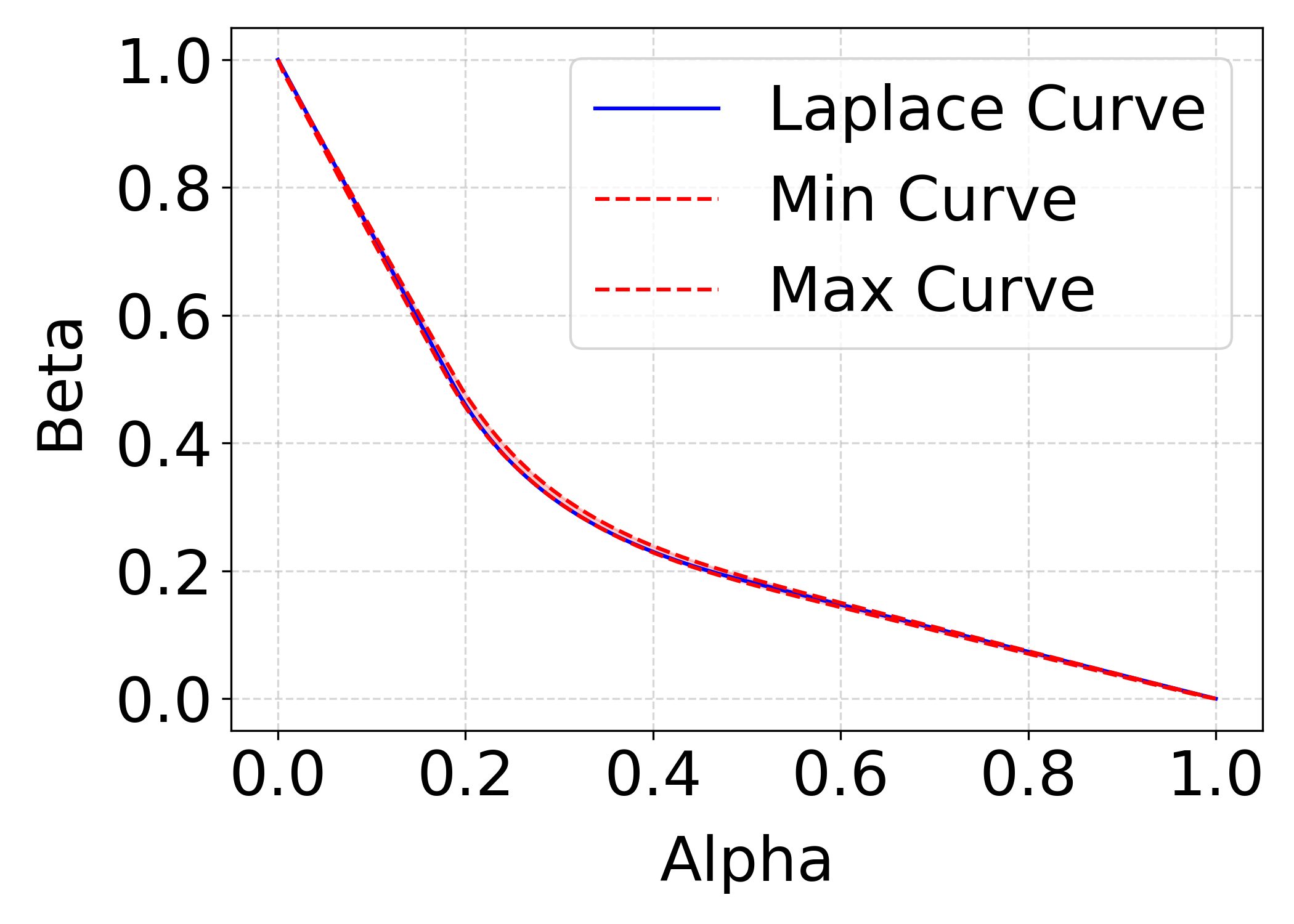}}
  \caption{Estimation of the Laplace Trade-off curve $T_{Lap}$ for varying sample sizes. Min- and Max Curve lower- and upper bound the worst point-wise deviation from the true curve $T_{Lap}$ over $1000$ simulations.}
    \label{fig:laplace}
\vspace{-0.1cm}
    \centering
    \subfloat[$n_1=10^3$]{\includegraphics[width=0.3\textwidth]{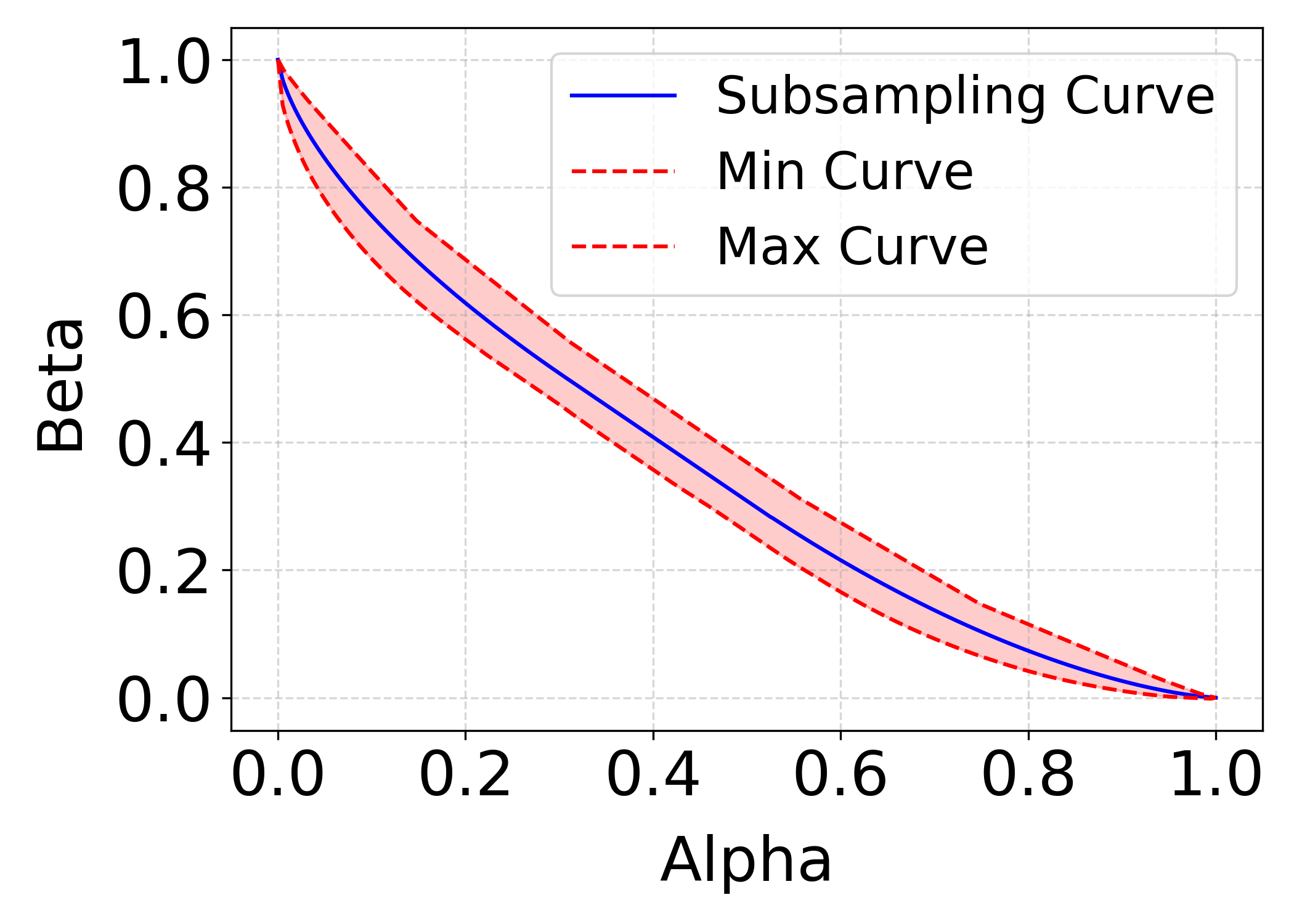}}
    \hfill
    \subfloat[$n_1=10^4$]{\includegraphics[width=0.3\textwidth]{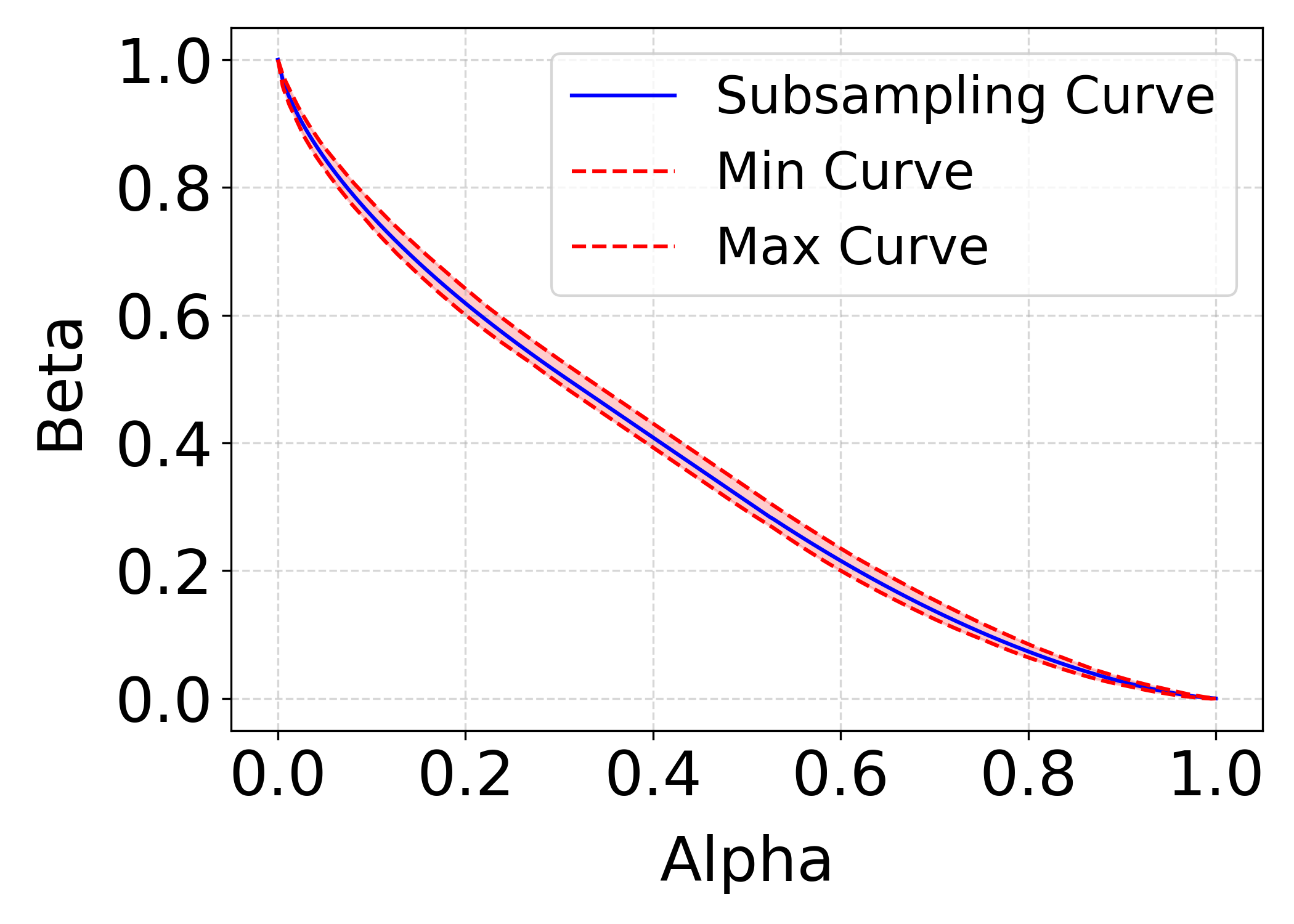}}
    \hfill
    \vspace{-0.2cm}
    \subfloat[$n_1=10^5$]{\includegraphics[width=0.3\textwidth]{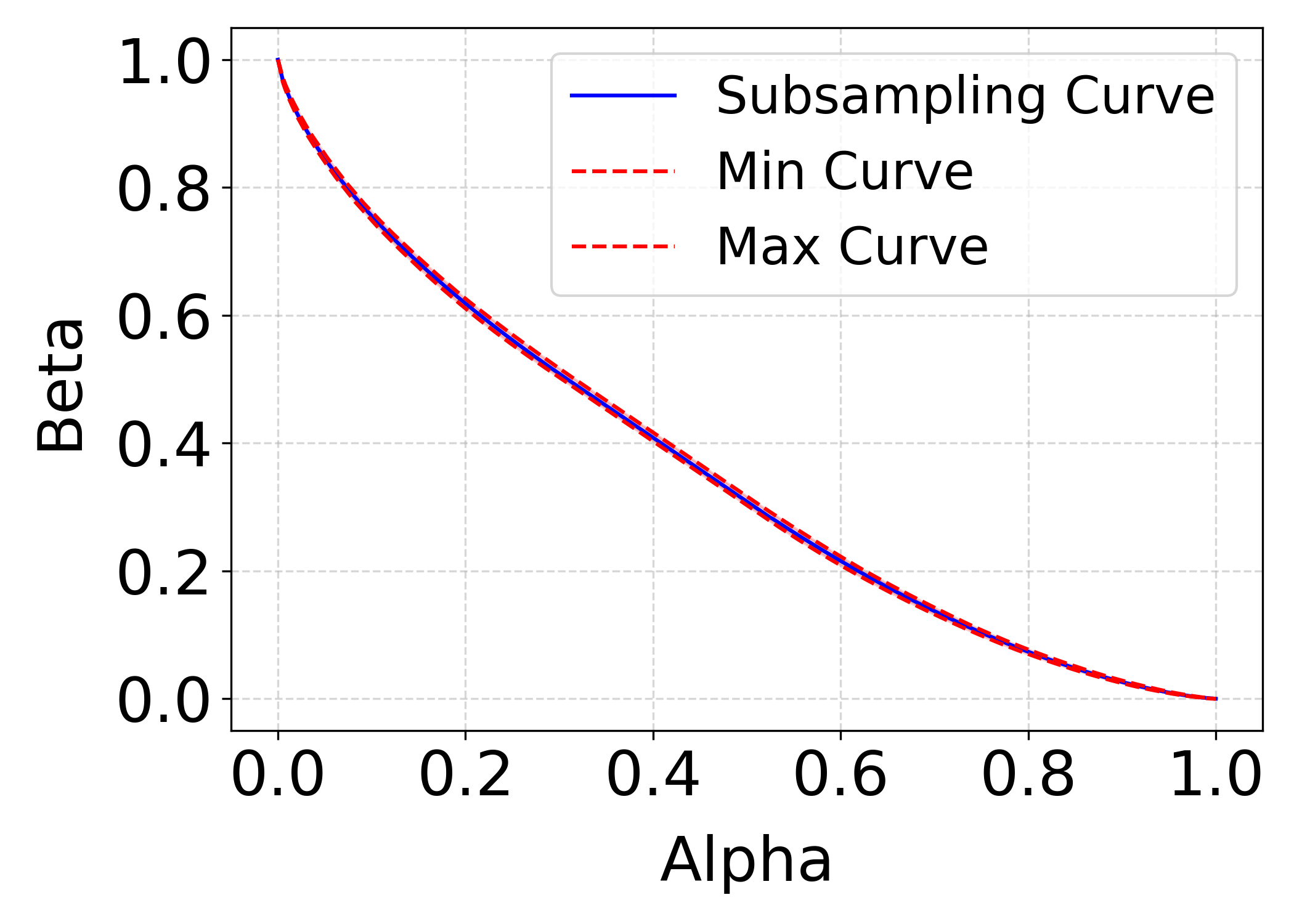}}
     \caption{Estimation of the Subsampling Trade-off curve $T_{Sub}$ with the Gaussian mechanism for $\mu=1$ and varying sample sizes. Min- and Max Curve lower- and upper bound the worst point-wise deviation from the true curve $T_{Sub}$ over $1000$ simulations.}
    \label{fig:subsampling}
\end{figure*}

\subsection{Additional Algorithms}
We test our estimation procedure on the Laplace and Subsampling algorithm, which often serve as building blocks in more sophisticated privacy mechanisms. We select the same setting for our experiments as in Section \ref{sec6} and choose $D$ and $D'$ as in \eqref{eq_databases}. \\

\noindent \textbf{Laplace mechanism.} We consider the summary statistic $S(x)= \sum_{i=1}^{10} x_i$ and the mechanism
\begin{equation*}
    M(x):= S(x)+Y~,
\end{equation*}
where $Y\sim \mathcal Lap(0, \sigma)$. The statistic $S(x)$ is privatized by the random noise $Y$ if the scale parameter $\sigma > 0$ of the Laplace distribution is chosen appropriately. We choose $\sigma = 1$ for our experiments and observe that the optimal trade-off curve is given by 
\begin{align*}
    T_{Lap}(\alpha)=\begin{cases}
        1- e \, \alpha,  &\alpha<e^{-1}/2~,\\
        e^{-1}/4 \alpha,  &e^{-1}/2\leq \alpha\leq 1/2~,\\
        e^{-1}(1-\alpha), &\alpha>1/2.
    \end{cases}
\end{align*}
We point the interested reader to \cite{Dong2022} for more details on how to derive $T_{Lap}$. \\

\noindent \textbf{Subsampling algorithm.} Random subsampling provides an effective way to enhance the privacy of a DP mechanism $M$. We only provide a rough outline here and refer for details to \cite{Dong2022}.
In simple words, we choose an integer $m$ with  $1\leq m< r$, where $r$ is the size of the database $D$. We then draw a random subsample of size $m$ from $D$, giving us the smaller database $\bar D$ of size $m$. The mechanism $M$ is then applied to $\bar D$ instead of $D$, providing users with an additional layer of privacy (if a user is not part of $\bar D$, their privacy cannot be compromised). The amplifying effect that subsampling has on privacy is visible in the optimal trade-off curve: If $M$ has the trade-off curve $T$, then $M(\bar D)$ has the trade-off curve
\begin{equation*}
    \bar T(\alpha)=  \frac{m}{r}T(\alpha)+\frac{r-m}{r}(1-\alpha),
\end{equation*}
which is strictly more private than $T$ for any $m<r$. A minor technical peculiarity of subsampling is that the resulting curve $\bar T$ is not necessarily symmetric, even if $T$ is (see \cite{Dong2022} for details on the symmetry of trade-off functions). Trade-off curves are usually considered to be symmetric and one can symmetrize $\bar T$ by applying a symmetrizing operator $\mathbf{C}$ with 
\begin{equation*}
    \mathbf{C}[T](x)=\begin{cases}
         T(x), \quad &x\in [0,x^*]\\
        x^*+ T (x^*)-x, \quad &x\in [x^*, T(x^*)]\\
         T^{-1}(x), \quad &x\in [ T(x^*),1],
    \end{cases}
\end{equation*}
where $x^*$ is the unique fix-point of $T$ with $T(x^*)=x^*$ (for more details we refer to \cite{Dong2022}). Another mathematical representation of $\mathbf{C}$ that we use in our code is 
$\mathbf{C}(T)=\min\{T,T^{-1}\}^{**}$, where the index $**$ signifies double convex conjugation. We incorporate this operation into our estimation procedure by simply applying $\mathbf{C}$ to our estimate of the trade-off function $T$. For our experiments involving subsampling, we use the Gaussian mechanism for $M$ (with $\sigma=1$) and obtain the subsampled version $M'$ by fixing the parameter $m=5$ (recall that $r=10$). \\

\noindent Similar to the experiments section, we construct figures that upper and lower bound the worst case errors for the Laplace mechanism and the Subsampling algorithm over $1000$ simulation runs. We can see again that the error of the estimator $\hat T_h$ shrinks significantly, as $n_1$ grows.

\subsection{Comparison to credible intervals}
\B{
In this work, we consider the construction of confidence bounds as common in frequentist statistics. If $\gamma$ is set to $1\%$ in Theorem \ref{theo:auditor}, this means that (on average) in $100$ audits of a correct algorithm only at most one violation will be (erroneously) detected. These kinds of guarantees are the gold standard in empirical sciences and we believe they are the guarantees real users would care about. It is worth noting that there exist other types of statistical results, including credible results from Bayesian statistics such as by  \cite{zanella2023bayesian}, who work on approximate DP. It is important to point out that Bayesian results are very different from frequentist approaches. One difference is their performance, because credible intervals do not generally provide the same bounds on false detection rates as frequentist results. We illustrate this point with a minimal simulation. The aim is to make a confidence/credible interval for the bias of a coin. 
We simulate \(n = 500\) coin flips per trial for \(k = 10^5\) trials, with varying bias (\(p\)). Frequentist confidence intervals use the sample proportion and normal approximation, while Bayesian credible intervals rely on a standard Beta prior (\(\alpha = 10, \beta = 10\)) updated with observed data. Coverage is assessed by checking if intervals contain the true \(p\). The targeted confidence/credible is $1-\gamma$ with $\gamma=0.1$ and results are displayed in Figure \ref{fig:coverage}.
As we can see, frequentist intervals hold repeated-sampling guarantees, while Bayesian credible intervals depend on priors and lack such guarantees under repeated sampling.}
\begin{figure}[h!]
\centering\includegraphics[width=0.45\textwidth]{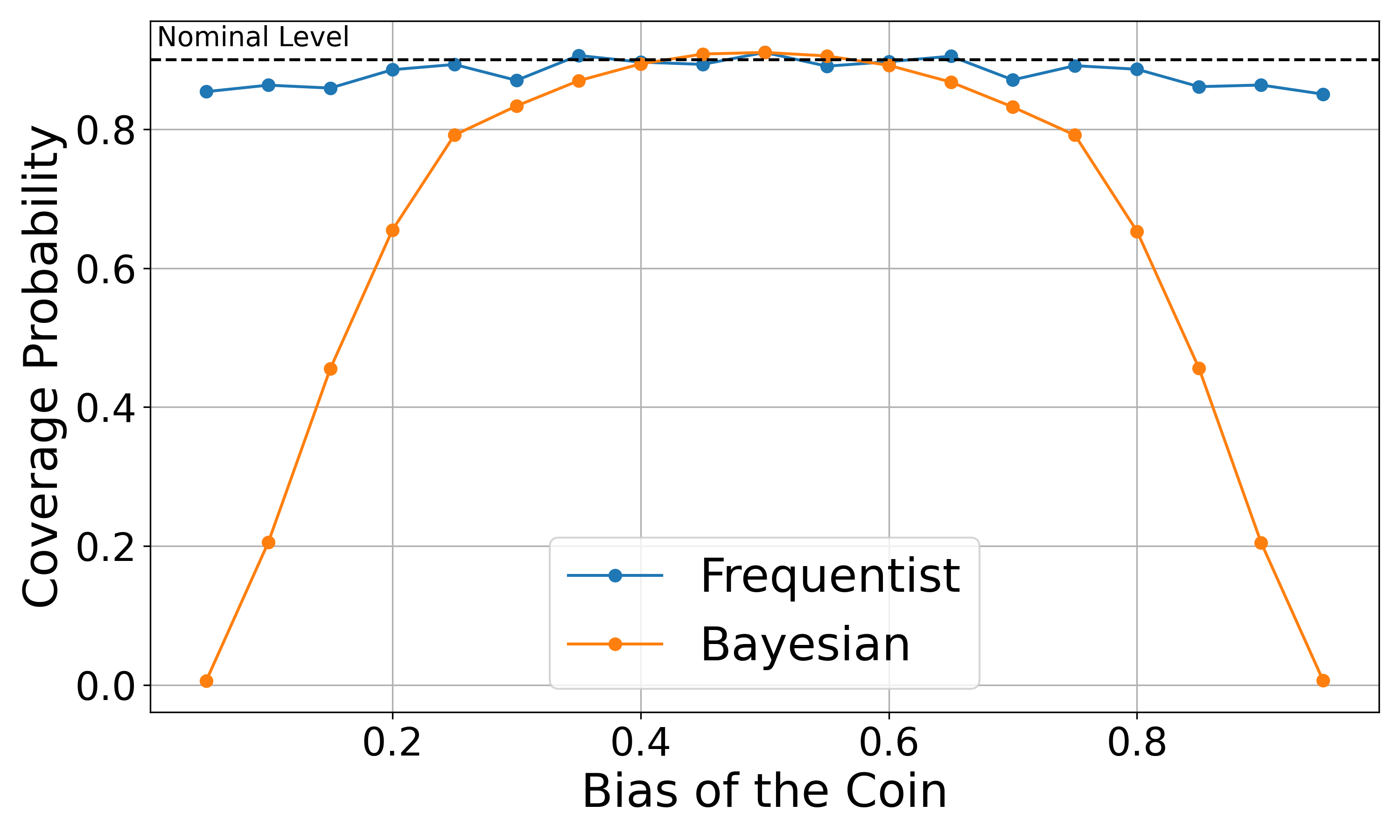}
    \caption{\B{Empirical coverage probabilities for frequentist and Bayesian intervals based on $10^5$ simulation runs.}}
    \label{fig:coverage}
\end{figure}

\subsection{Additional simulations}\label{app:add_simulations}

We present some results that complement the main findings in our experiment section. We use the same setup as described in our experiments and investigate a faulty implementation of the Gaussian mechanism. We study two things: First, the impact of the parameter $\gamma$, where we vary $\gamma$ between very small and relatively large values. As we can see, smaller values of $\gamma$ lead to larger boxes $\square_\gamma$ which make it harder for the auditor to detect violations. Secondly, we consider the impact of the sample size $n_1$ ranging from the very modest value of $10^2$ up to $10^4$. We see that the sample size has very little impact on the performance of the procedure and it already works well for fairly small samples $n_1$ ($n_2$ has a greater impact, as we have seen in our experiments).
Finally, we have also reported the runtimes for different mechanism in Table \ref{tab:running_times_KDE} and Table \ref{tab:running_times_kNN}.
\begin{table}[h!]
\centering
\begin{tabular}{|l|c|}
\hline
\textbf{Algorithm}                           & \textbf{Runtime in seconds} \\ \hline
Gaussian mechanism              &  26.3                                                    \\ \hline
Laplace mechanism            &    30.51                                                     \\ \hline
Subsampling mechanism         &   27.82                                                      \\ \hline
DP-SGD           &              61.1                                         \\ \hline
 
\end{tabular}
\caption{Average runtimes of Algorithm \ref{alg:pointwise_KDE_estimator} for $n_1=10^5$ over $10$ runs to obtain the full trade-off curve $T$.}
\label{tab:running_times_KDE}
\end{table}
\begin{table}[h!]
\centering
\begin{tabular}{|l|c|}
\hline
\textbf{Algorithm}                           & \textbf{Runtime in seconds} \\ \hline
Gaussian mechanism              &    62.63                                                  \\ \hline
Laplace mechanism            &        67.04                                                 \\ \hline
Subsampling mechanism         &     66.98                                                  \\ \hline
DP-SGD           &    114.86                                                 \\ \hline
 
\end{tabular}
\caption{Average runtimes of Algorithm \ref{alg: general BayBox estimator} for $n_2=10^6$ over $5$ runs to obtain one point of the trade-off curve $T$ with confidence region.} %\\[-8ex]} 
\label{tab:running_times_kNN}
\end{table}
\begin{figure*}[t!]
    \centering
    \subfloat[\centering $\gamma=0.001$, \textbf{Ground truth:} Violation; \textbf{Decision:} \textcolor{green}{"Violation"}{\textcolor{green}{\scalebox{1.5}{\ding{51}}}}
    % \textcolor{red}{"No Violation"}{\textcolor{red}{\scalebox{1.5}{\ding{55}}}}
    ]{\includegraphics[width=0.3\textwidth]{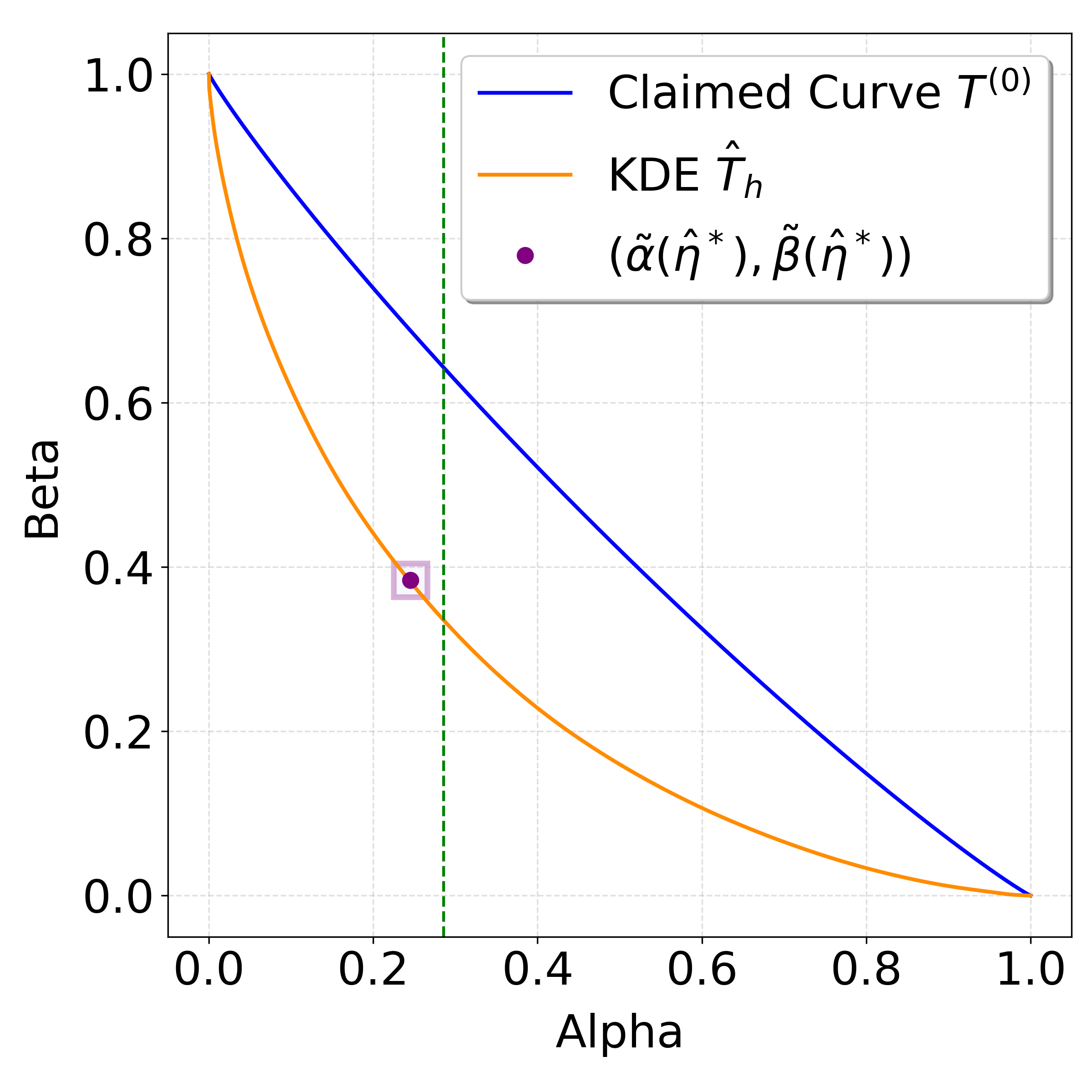}}
    \hfill
    \subfloat[\centering $\gamma=0.01$, \textbf{Ground truth:} Violation; \textbf{Decision:} \textcolor{green}{"Violation"}{\textcolor{green}{\scalebox{1.5}{\ding{51}}}}
    %\textcolor{red}{"No Violation"}{\textcolor{red}{\scalebox{1.5}{\ding{55}}}}
    ]{\includegraphics[width=0.3\textwidth]{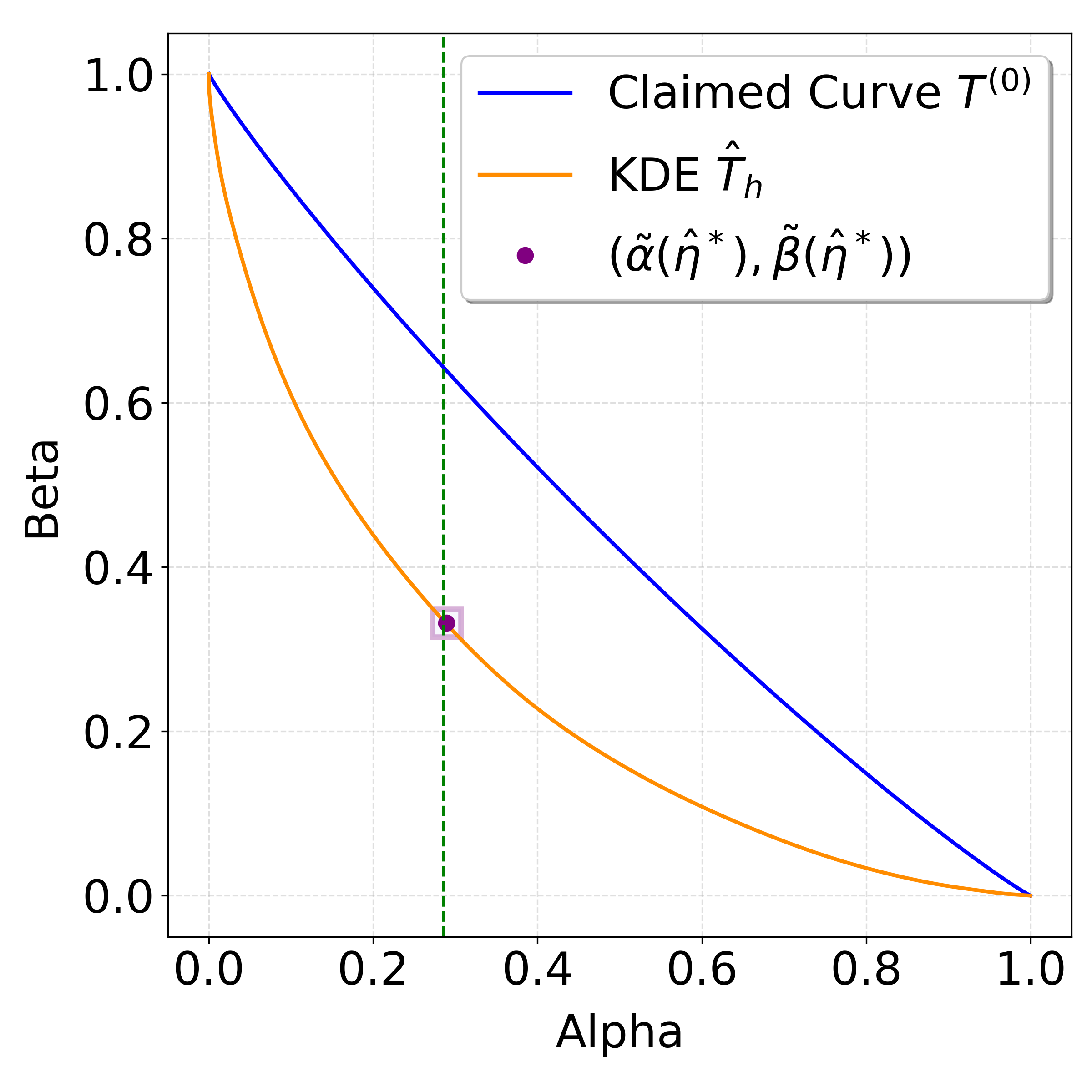}}
    \hfill
    \subfloat[\centering $\gamma=0.1$, \textbf{Ground truth:} Violation; \textbf{Decision:} \textcolor{green}{"Violation"}{\textcolor{green}{\scalebox{1.5}{\ding{51}}}}]{\includegraphics[width=0.3\textwidth]{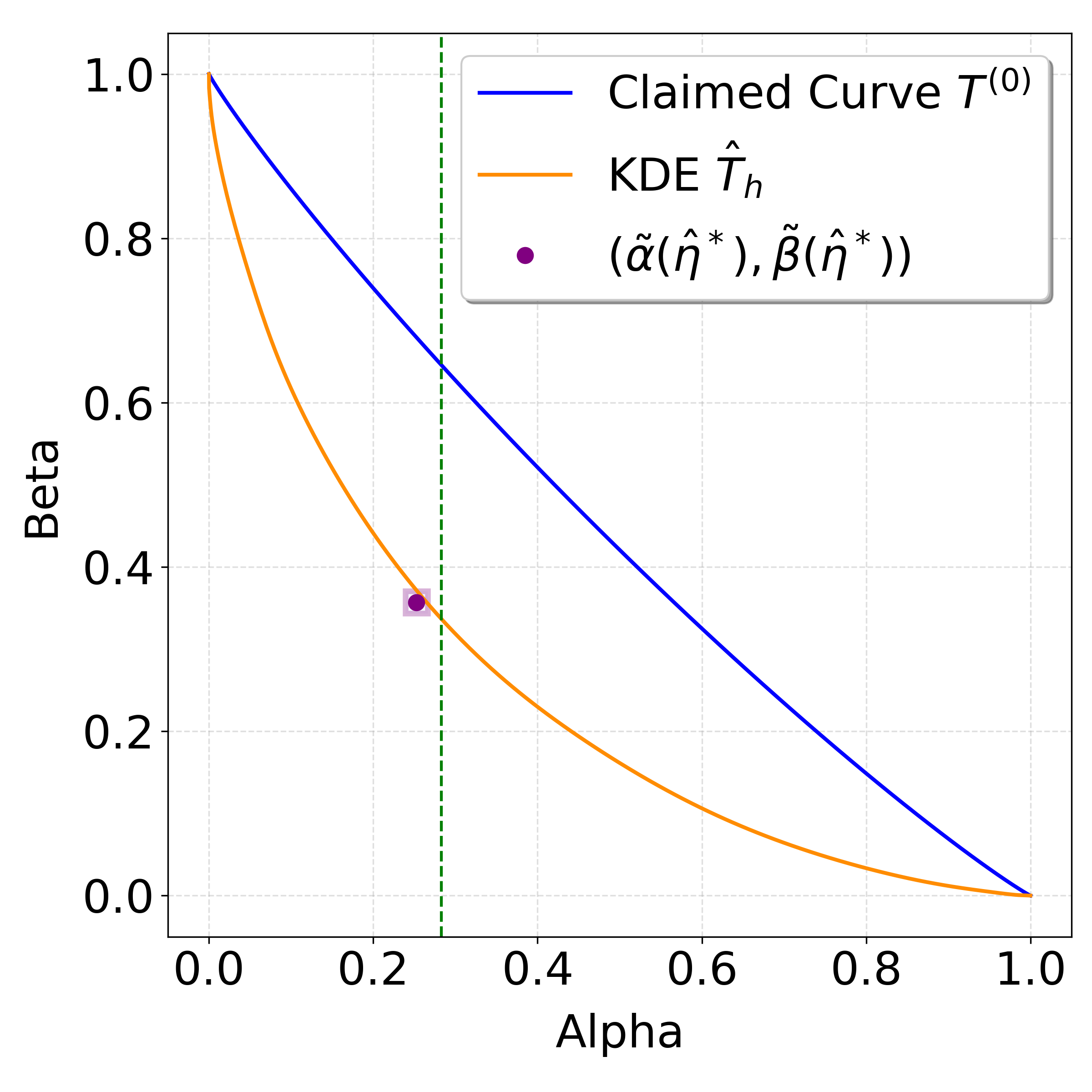}}
    \hfill
    \subfloat[\centering $n_1=10^2$,
  \textbf{Ground truth:} Violation; \textbf{Decision:} \textcolor{green}{"Violation"}{\textcolor{green}{\scalebox{1.5}{\ding{51}}}}]
  {\includegraphics[width=0.3\textwidth]{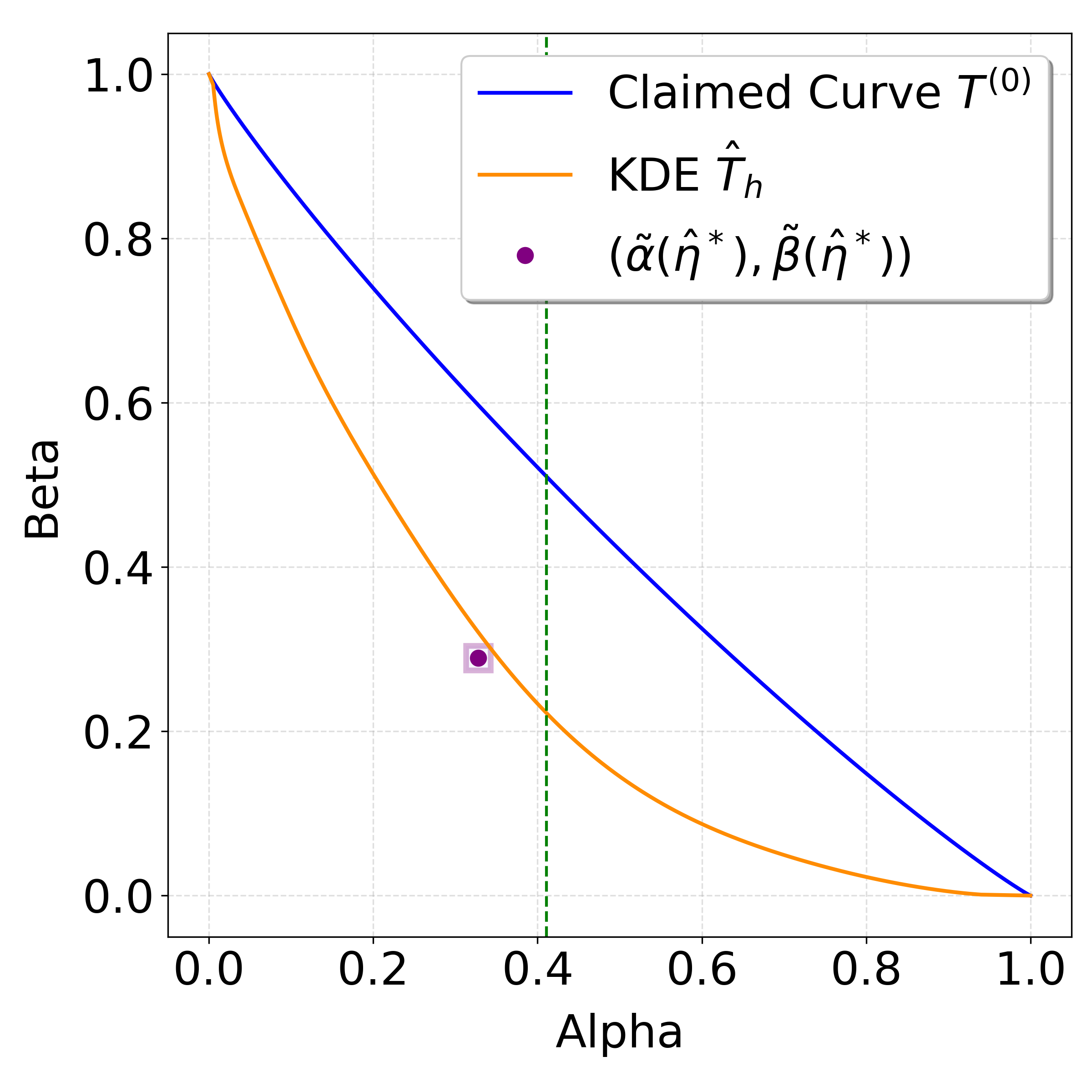}}
    \hfill
    \subfloat[\centering $n_1=10^3$, \textbf{Ground truth:} Violation; \textbf{Decision:} \textcolor{green}{"Violation"}{\textcolor{green}{\scalebox{1.5}{\ding{51}}}}]{\includegraphics[width=0.3\textwidth]{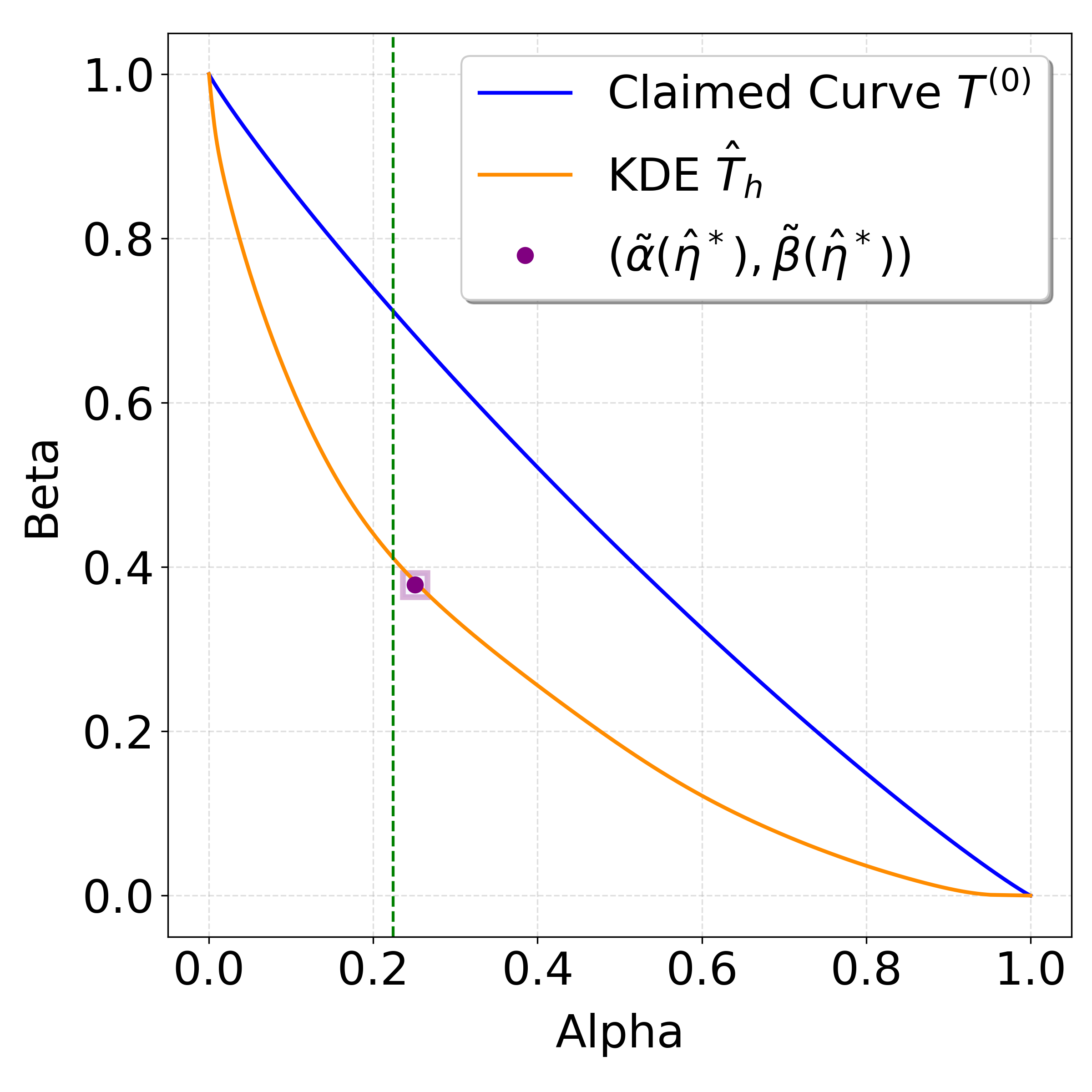}}
    \hfill
    \subfloat[\centering $n_1=10^4$, \textbf{Ground truth:} Violation; \textbf{Decision:} \textcolor{green}{"Violation"}{\textcolor{green}{\scalebox{1.5}{\ding{51}}}}]{\includegraphics[width=0.3\textwidth]{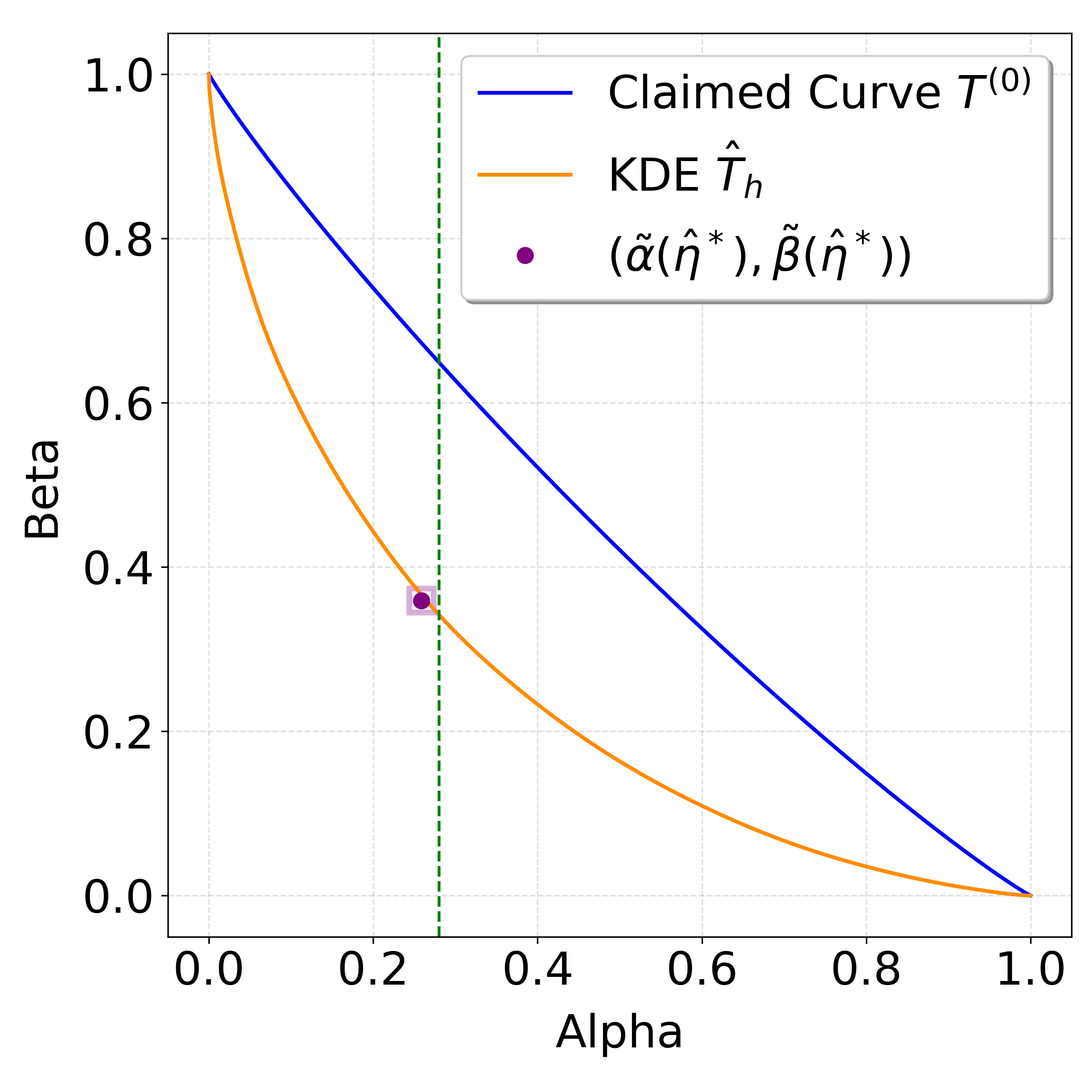}}
        \caption{\textbf{Auditing a faulty Mechanism:} Claimed Curve $\textcolor{blue}{T^{(0)}} = \textcolor{blue}{T_{Gauss}}$ with $\mu=0.2$, but in reality $\mu=1$. For (a),(b),(c) we consider $n_1=10^4$, and for (d),(e),(f) we have considered various sample sizes for the KDEs, respectively. Throughout the simulations we keep $n_2=10^4$ fixed, and confidence intervals in (d), (e), and (f) are computed with level $\gamma = 0.05$}\label{fig:additional_experiments}
       % \caption{\textbf{Auditing a faulty Mechanism:} Claimed Curve $\textcolor{blue}{T_0} = \textcolor{blue}{T_{SGD}}$ with $t_{-}=5$ iteration, but in reality data is accessed $t_{-}=10$ times. Estimate $\textcolor{blue}{T_0}$ by $\textcolor{orange}{\hat T_0}$.  With critical vertical line with intercept $(\hat\alpha(\hat\eta^*), \hat \beta(\hat \eta^*))$, $k$-NN point estimator {\textcolor{purple}{\ding{108}}} $(\tilde\alpha(\hat\eta^*), \tilde \beta(\hat\eta^*))$ and confidence region $\textcolor{purple}{\square}$. The sample size for KDE is $n_1=10^2$ and the confidence parameter is $\gamma=0.05$.}
\end{figure*}
%The derivation of this result is quite involved. 

\subsection{Proofs for Goal 1 (Estimation)}

\textbf{Consequences of Theorem \ref{theo:1}} The main result in Section \ref{sec:4} is Theorem \ref{theo:1}. Lemma \ref{lem1} can be seen as a special case, putting $\hat p=p, \hat q=q$ . Then, Assumption \ref{ass2} is met for the constant sequence $a_n=0$. It follows by this construction that $\hat T_h =T_h$, is non-random and only depends on $h$. Any choice of $h\downarrow 0$ is permissible and  Lemma \ref{lem1} follows from the Theorem. Proposition \ref{prop1} too is a direct consequence of Theorem \ref{theo:1}. To see this, we notice that
\begin{align*}
    & T_0(\hat \alpha_h(\hat \eta^*)) - T(\hat \alpha_h(\hat \eta^*)) \\
    =& T_0(\hat \alpha_h(\hat \eta^*)) - \hat T_h(\hat \alpha_h(\hat \eta^*))+o_P(1)\\
    = & \sup_{\alpha \in [0,1]}\{T_0(\alpha) - \hat T_h(\alpha)\}+o_P(1)\\
    =& \sup_{\alpha \in [0,1]}\{T_0(\alpha) -  T(\alpha)\}+o_P(1).
\end{align*}
In the first and last step, we have used the uniform convergence of Theorem \ref{theo:1}, which allows us to replace $T$ by $\hat T_h$ while only incurring an $o_P(1)$ error. In the second step, we have used the definition of $\hat \alpha_h(\hat \eta^*)$ as the maximizer of the difference between $T_0$ and $\hat T_h$. Thus Proposition \ref{prop1} follows. We now turn to the proof of the theorem. The proof is presented for densities on the real line. Extensions to $\mathbb{R}^d$ are straightforward and therefore not discussed. \\
\textbf{Preliminaries} Recall that a complete separable metric space is Polish. The real numbers, equipped with the absolute value distance is a Polish space. The continuous functions $\mathcal{C}_0$ on the real line that vanish at $\pm \infty$, i.e. that satisfy
\begin{align} \label{e:decay}
\lim_{x \to \infty} f(x) = \lim_{x \to \infty} f(-x)=0 
\end{align}
is a Polish space if equipped with the supremum norm
\[
\|f\|:= \sup_{x \in \mathbb{R}}|f(x)|.
\] 
The product of complete, separable metric spaces is complete and separable if equipped with the maximum metric, i.e. the space $\mathcal{C}_0 \times \mathcal{C}_0\times \mathbb{R}\times \mathbb{R}$ is Polish. 
Now, the vector 
\[
(\hat p, \hat q, \|\hat p-p\|_\infty/a_n, \|\hat q-q\|_\infty/a_n)
\]
lives on this space (for each $n$) and convergences to the limit $(p,q,0,0)$ in probability (see Assumption \ref{ass2}). Accordingly we can use Skorohod's theorem to find a probability space, where this convergence is a.s. 
\[
(\hat p, \hat q, \|\hat p-p\|_\infty/a_n, \|\hat q-q\|_\infty/a_n) \to (p,q,0,0) \quad a.s.
\]
It is a direct consequence that on this space it holds a.s.
\[
\|\hat p-p\| =o(a_n), \quad \|\hat q-q\| =o(a_n).
\]
In the following, we will work on this modified probability space and exploit the a.s. convergence. We will fix the outcome and regard $\hat p, \hat q$ as sequences of deterministic functions, converging to their respective limits at a rate $o(a_n)$.\\
    Next, it suffices to show the desired result pointwise for any $\alpha$. This reduction is well-known. For a sequence of  continuous, monotonically decreasing functions $(f_n)_n$ living on the unit interval $[0,1]$, pointwise convergence  to a continuous, monotonically decreasing limit $f$ on $[0,1]$ implies uniform convergence. The same argument lies at the heart of the proof of the famous Glivenko-Cantelli Theorem (see \cite{vaart:wellner:1996}). We now want to demonstrate the convergence $|\hat T(\alpha) -T(\alpha)|=o(1)$ pointwise. More precisely, we will demonstrate that for the pair $(\alpha, T(\alpha))$, there exist values of $\eta$ such that $\hat \alpha_h(\eta) \to \alpha$ and $\hat\beta_h(\eta) \to T(\alpha)$. Since the proofs of both convergence results work exactly in the same way, we restrict ourselves in this proof to show that $\hat \alpha_h(\eta) \to \alpha$. So let us consider a fixed but arbitrary value of $\alpha \in [0,1]$ and begin the proof.\\
    \textbf{Case 1:} We first consider the case where $\eta\ge 0$ (the threshold in the optimal LR test) is such that the set $\{q/p=\eta\}$ has $0$ mass. In this case, the coin toss with probability $\lambda$ can be ignored (it happens with probability $0$) and  we can define the type-I-error  $\alpha$ of the Neyman-Pearson test as 
    \[
\alpha = \int p \cdot \mathbb{I}\{q/p  > \eta\}.
\]
In this case, we want to show that 
\begin{align*}
& \int_{x \in [-h/2,h/2]} \frac{1}{h}\int_{\hat q /\hat p  > \eta +x} \hat p \\
=& \int  \int_{x \in [-h/2,h/2]}  \hat p \frac{1}{h}\,\,\mathbb{I}\{ \hat q /\hat p  > \eta +x\}=:\int \hat g \\
\to &  \int_{q/p  > \eta}  p  = \int p \cdot \mathbb{I}\{q/p  > \eta\} =:\int g .
\end{align*}
Here we have defined the functions $ g, \hat g$ in the obvious way. We will now show $\hat g$ converges pointwise to $g$. For this purpose consider the interval $[-K,K]$ for a large enough $K$, such that
\[
\int_{[-K,K]^c} p<\zeta \qquad \textnormal{and} \qquad\int_{[-K,K]^c} q<\zeta
\]
for a number $\zeta$ that we can make arbitrarily small. Given the uniform convergence of the density estimators on the interval $[-K,K]$ it holds for all $n$ sufficiently large that also
\[
\int_{[-K,K]^c} \hat p<\zeta \qquad \textnormal{and} \qquad\int_{[-K,K]^c} \hat q<\zeta.
\]
Accordingly we have 
 \[
 \bigg| \int \hat g - g\bigg| \le 2 \zeta + \bigg| \int_{[-K,K]} \hat g -g\bigg|.
 \]
 We then focus on the second term on the right and fix some argument $y \in [-K,K]$. It holds that either $ q(y)/ p(y)$ is bigger or smaller than $\eta$ (equality occurs only on a null-set and can therefore be neglected). Let us focus on the case where $q(y)/ p(y)> \eta$. If this is so, then it follows that in a small environment, say for $y' \in [y-\zeta', y+\zeta']$ we also have $q(y')/ p(y')>\eta$. For all large enough $n$ it follows that $h/2<\zeta'$. Then, it is easy to see that also $\hat q(y')/\hat p(y') >\eta$ for all $y' \in [y-\zeta', y+\zeta']$ simultaneously, for all sufficiently large $n$. If this is the case, the indicators in the definition of $\hat g, g$ become $1$ and $\hat g=\hat p$, $g=p$. 
 So, we have pointwise $\hat g(y)=\hat p(y)  \to p(y) =g(y)$. Since $\hat g$ is also bounded for all sufficiently large $n$ (since the integral over the indicator is bounded and the sequence $\hat p$ is uniformly convergent to the bounded function $p$) we obtain by the theorem of dominated convergence that 
 \[
 \bigg| \int_{[-K,K]} \hat g -g\bigg|\to 0.
 \]
 This shows that 
 \[
 \limsup_n|\hat \alpha_h(\eta) - \alpha|=\mathcal{O}(\zeta).
 \]
 Finally, letting $\zeta \downarrow 0$ in a second limit shows the desired approximation in this case.\\
 \textbf{Case 2:} Next, we consider the case where the set $\{q/p=\eta\}$ has positive mass for some $\eta>0$.\footnote{We omit the simpler case where $\eta=0$ and $L=0$ anyways.}
 This means that the coin-flip in the definition of the optimal LR test plays a role and we set the probability $\lambda $ to some value in $[0,1]$.
  We then consider as estimator the value $\hat \alpha(\eta-b h)$ for a value $b$ that we will determine below. Let us, for ease of notation, define the probability 
 \[
 L := \int_{q/p=\eta} p
 \]
and appreciate that then
\begin{align} \label{e:dec}
\alpha = \alpha'+\mathcal{O}(\zeta) + \lambda L.
\end{align}
We explain the decomposition: In equation \eqref{e:dec}, $\alpha'$ is the rejection probability of the LR test defined by the decision to reject whenever $q(y)/p(y)>\eta+\zeta''$ for some small number $\zeta''$. For all small enough values of $\zeta''$ the threshold $\eta+\zeta''$ is not a plateau value (there are only finitely many of them; see Assumption \ref{ass1}). It follows that 
\[
\alpha' = \int p \cdot \mathbb{I}\{q/p  > \eta+\zeta''\}.
\]
Next, for any fixed constant $\zeta>0$ we can choose $\zeta''$ small enough such that
\begin{align} \label{e:int1}
\int p \cdot \mathbb{I}\{\eta <q/p  \le \eta+\zeta''\} < \zeta.
\end{align}
This explains the second term on the right of equation \eqref{e:dec}. The third term corresponds to the probability of rejecting whenever $q/p=\eta$ (this probability is $L$) times the probability that the coin shows heads (reject) with probability $\lambda$.\\
Now, using these definitions, we decompose the set
\begin{align*}
&\{ \hat q /\hat p  > \eta-b h +x\}\\
=& \{ \eta +\zeta'' \ge \hat q /\hat p  > \eta-b h +x\} \cup  \{  \hat q /\hat p  > \eta +\zeta''\}. 
\end{align*}
This yields the decomposition
\begin{align} \label{e:alcon}
& \hat \alpha_h(\eta-b h)
= \hat \alpha_h(\eta+\zeta'')\\
&+ \int  \int_{x \in [-h/2,h/2]}  \hat p \,\, \frac{1}{h}\,\,\mathbb{I}\{ \eta +\zeta'' \ge \hat q /\hat p  > \eta-b h +x\} .\nonumber 
\end{align}
Now, by part 1 of this proof we have  
\[
|\hat \alpha_h(\eta+\zeta'') - \alpha'|=o(1).
\]
 Next, we study the integral on the right side of eq. \eqref{e:alcon} and for this purpose define the objects
\begin{align*}
    \tilde g := & \int_{x \in [-h/2,h/2]}  \hat p \,\,\frac{1}{h}\,\,\mathbb{I}\{ A_1\}, \\
    \tilde f := & \int_{x \in [-h/2,h/2]}  \hat p \,\,\frac{1}{h}\,\,\mathbb{I}\{ A_2\}.\\
    A_1:= &\{\eta +\zeta'' \ge \hat q /\hat p  > \eta-b h +x,q/p=\eta\}, \\
    A_2:=& \{\eta +\zeta'' \ge \hat q /\hat p  > \eta-b h +x,q/p \neq \eta\}.
\end{align*}
Now, let us consider a value $y$ where $q(y)/p(y) \neq \eta$ and for sake of argument let us focus on the (more difficult) case $q(y)/p(y) >\eta$. If $q(y)/p(y) > \eta+\zeta''$, it follows that eventually $\hat p(y)/\hat q(y) > \eta+\zeta''$ and hence $\tilde f(y)=0$. The case where $q(y)/p(y) = \eta+\zeta''$ is a null-set and hence negligible (it is not a plateau value). The case where  $q(y)/p(y) \in (\eta, \eta+\zeta'')$ implies that eventually $\hat p(y)/\hat q(y) \in (\eta, \eta+\eta'')$ and thus eventually $\tilde f(y) = \hat p(y)$ which converges pointwise to $p$. Thus, we have by dominated convergence that 
\[
\int \tilde f \to \int p \cdot \mathbb{I}\{\eta <q/p  \le \eta+\zeta''\} <\zeta.
\]
The fact that the integral is bounded by $\zeta$ was established in eq. \eqref{e:int1}. This means that for all $n$ large enough we have 
\[
\int \tilde f < \zeta.
\]
Now, let us focus on a value of $y$ where $q(y)/p(y)=\eta$. In this case it follows that $q(y), p(y)>0$ and we have 
\[
\frac{\hat q(y)}{\hat p(y)} = \frac{q(y)}{p(y)} +o(a_n) = \eta +o(a_n).
\]
Notice that we can rewrite $\tilde g$ as
\begin{align*}
    \int_{x \in [-1/2,1/2]}  \hat p \,\,\mathbb{I}\{\eta +\zeta'' \ge \hat q /\hat p  > \eta-b h +hx,q/p=\eta\}.
\end{align*}
Now, for any $x>b$ it follows that the indicator will eventually be $0$, because 
\[
\hat q /\hat p =\eta+o(a_n) << \eta + h(x-b)
\]
(because $a_n=o(h)$ by assumption in the Theorem). By similar reasoning the indicator is $1$ if $x<b$. This means that $\tilde g$ converges for any fixed $y$ with $q(y)/p(y)=\eta$ to $p(y) \cdot (1/2+b)$ and using majorized convergence yields
\[
\int \tilde g \to (1/2+b) \int_{q/p=\eta} p =  (1/2+b)L.
\]
Now, we can choose  $b=\lambda-1/2$ to get that the right side is equal to $\lambda L$. Putting these considerations together, we have shown that
\[
\limsup_n |\alpha -\hat \alpha_h(\eta-[\lambda-1/2] h)| = \mathcal{O}(\zeta). 
\]
Taking the limit $\zeta \downarrow 0$ afterwards yields the desired result. 

\subsection{Proofs for Goal 2 (Auditing)}
Before we proceed to the proofs, we state a simple but useful consequence of the Neyman-Pearson Lemma.
\begin{cor}
   \label{corollary: NP lemma}
    Let set $\cS_{\eta} = \{x: p(x)/q(x) \leq \eta \}$.
    For $\alpha \in [0,1]$, if there exists $\eta$ such that $\prdis{\rX \sim P}{\rX \in \cS_{\eta}} = \alpha$, then it holds that
    \begin{align*}
    \beta(\alpha) = 1 - \prdis{\rX \sim Q}{\rX \in \cS_{\eta}}.
    \end{align*}
% where $\cS_{\eta}$ is chosen such that $\prdis{\rX \sim P}{\rX \in \cS_{\eta}} = \alpha$.
\end{cor}
\B{
As a next step, we prove a theoretical result connecting the output of the BayBox estimator for the theoretical (in practice unknown) Bayes classifier $\phi^*$.}
%Lemma~\ref{lemma: accuracy stat of general BayBox estimator} states that, assuming the Bayes optimal classifier can be constructed, one can establish simultaneous confidence intervals for the parameters $\alpha(\eta)$ and $\beta(\eta)$ with a user-specified failure probability $\gamma$, which can be made arbitrarily small, based on the output of the BayBox estimator. In practice, however, the Bayes classifier $\phi^*$ is usually unknown and need to be approximated. Nevertheless, Lemma~\ref{lemma: accuracy stat of general BayBox estimator} is of independent interest, as it suggests that our method is, to some extent, agnostic to the choice of the classification algorithm.
\B{
\begin{lem}%[Proof is in Appendix~\ref{proof: for lemma: accuracy stat of general BayBox estimator}]
    \label{lemma: accuracy stat of general BayBox estimator}
    Let $\eta$, $(\alpha(\eta), \beta(\eta))$, $(\tilde{\alpha}(\eta), \tilde{\beta}(\eta))$, and $\phi$ be as defined in Algorithm~\ref{alg: general BayBox estimator}. Set $\phi$ to the Bayes optimal classifier $\phi^*$ for the corresponding Bayesian classification problem.
    Then, with probability $1 - \gamma$,
    \begin{align*}
    &\abs{\tilde{\alpha}(\eta) - \alpha(\eta)} \leq \sqrt{\frac{1}{2n}\ln{\frac{4}{\gamma}}} \\
    &\abs{\tilde{\beta}(\eta) - \beta(\eta)} \leq \sqrt{\frac{1}{2n}\ln{\frac{4}{\gamma}}}.
    \end{align*}
\end{lem}
}

\begin{proof}[\textbf{Proof of Lemma~\ref{lemma: accuracy stat of general BayBox estimator}}]
    \label{proof: for lemma: accuracy stat of general BayBox estimator}
    We prove the statement that $\abs{\tilde{\alpha}(\eta) - \alpha(\eta)} \leq \sqrt{\frac{1}{2n}\ln{\frac{4}{\gamma}}}$ if $\eta \geq 1$ with probability $\ge 1-\gamma/2$. The proof of the second statement follows a similar approach. We begin with a few definitions.
    Let the observation set be defined as  
    \begin{align*}
        \cO := \Supp{P} \cup \Supp{Q} \cup \{\bot\},
    \end{align*}
    i.e. the range of observation. Define the indicator function $\mathbb{I}_{\cS_{\eta}} : \cO \mapsto \bits$, which takes as input an observation $x$ from the observation set $\cO$, outputting 1 if $x \in \cS_\eta$ and $0$ otherwise. Also, recall the definition of the set $\cS_{\eta} = \{x: p(x)/q(x) \leq \eta \}$ 
    % \todo{is it maybe $<$ instead of $\leq$? maybe doesn't matter} 
    as the set of all observation $x \in \cO$ where $p(x)$ is less than or equal to $\eta q(x)$ (as before $p, q$ are the densities of distributions $P, Q$).
    
    We first show that $\mathbb{I}_{\cS_{\eta}}$ is exactly the Bayes classifier $\phi^{*}$ for the Bayesian binary classification problem $\bbcP{\MixtureD{P}{\eta}}{Q}$. We prove this by showing for every $x \in \cO$, $\phi^{*}(x) = \mathbb{I}_{\cS_{\eta}}(x)$. Therefore, consider the tuple of random variable $(\rX, \rY) \sim \bbcP{\MixtureD{P}{\eta}}{Q}$. Then, for every observation $x \in \cO \setminus \{\bot\}$, we have  
    \begin{align*}
       \phi^{*}(x) ={}& \argmax_{\bits} \{\pr{\rY = 0 | \rX = x}, \pr{\rY = 1 | \rX = x}\} \tag{by Bayes classifier $\phi^{*}$'s construction}\\
        ={}& \argmax_{\bits} \{\pr{\rY = 0, \rX = x}, \pr{\rY = 1, \rX = x}\} \tag{by Bayes Theorem}\\
        ={}& \argmax_{\bits} \{\frac{1}{\eta}p(x), q(x)\} \\
        ={}& \mathbb{I}_{\cS_{\eta}}(x) \tag{by $\mathbb{I}_{\cS_{\eta}}$'s definition}.
    \end{align*}
    For an observation $x = \bot$, it is easy to check $\phi^{*}(x) = \mathbb{I}_{\cS_{\eta}}(x) = 0,$ as $q(x) = 0.$

    Then, we also observe that  
    \begin{align}
        &\alpha(\eta) \label{equ: alpha and bayes classifier}\\
        ={}& \prdis{\rX \sim P}{\rX \in \cS_{\eta}} \tag{By Corollary~\ref{corollary: NP lemma}} \nonumber\\
        ={}& \prdis{\rX \sim P}{\mathbb{I}_{\cS_{\eta}}(\rX) = 1} \nonumber\\
        ={}& \prdis{\rX \sim P}{\phi^*(\rX) = 1} \tag{$\phi^* = \mathbb{I}_{\cS_{\eta}}$} \nonumber\\
        ={}& \Exf{\rX \sim P}{\phi^*(\rX)}  \nonumber
    \end{align}

    Recall that in algorithm~\ref{alg: general BayBox estimator}, BayBox estimatior $\bbe{\phi^*}$ computes the empirical mean of $\phi^*(\rX)$, i.e., $\tilde{\alpha}(\eta)$, as the estimate of $\alpha(\eta)$. By Hoeffding's Inequality, we finally conclude that
    \begin{align*}
           &\pr{\abs{\tilde{\alpha}(\eta) - \alpha(\eta)} > \sqrt{\frac{1}{2n}\ln{\frac{4}{\gamma}}}} \\
        ={}&\pr{\abs{\frac{1}{n}\rSum{i}{1}{n}\rZ_i - \Ex{\frac{1}{n}\rSum{i}{1}{n}\rZ_i}} > \sqrt{\frac{1}{2n}\ln{\frac{4}{\gamma}}}} \tag{$\rZ_i \defin \phi^*(\rX_i), \rX_i \overset{\text{i.i.d.}}{\sim} P$}\\
        \leq{}& \gamma/2.
    \end{align*}
\end{proof}

\begin{rem} \label{lem:bias}
    \B{For any classifier $\phi$ that is used as input of the BayBox algorithm, the output $(\tilde \alpha(\eta), \tilde \beta(\eta))$ will have a mean point $(\mathbb{E}\tilde \alpha(\eta), \mathbb{E}\tilde \beta(\eta)))$ on or above the optimal trade-off curve. The reason is that $(\tilde \alpha(\eta), \tilde \beta(\eta))$ are the empirical type-I and type-II-errors of the test that rejects whenever an output is classified as belonging to $D'$. The means $\mathbb{E}\tilde \alpha(\eta), \mathbb{E}\tilde \beta(\eta)$ correspond to the population version of the errors which by construction of the optimal trade-off curve are on or above it (no test has a better combination than the Neyman-Pearson test which demarcates the curve exactly).}
\end{rem}

%\begin{lem}
%    bla
%\end{lem}
%\begin{proof}
%    Equation~\ref{equ: alpha and bayes classifier} shows that the $\alpha$ we want to compute equals to the expectation of the Bayes classifier’s misclassification error for a classification problem. By definition of Bayes classifier, any other classifier’s misclassification error’s expectation is larger than that of Bayes’s classifier, so that their mean of misclassification error is strictly an upper bound of $\alpha$. To conclude, every randomized binary function induce an upper bound of alpha. The same conclusion holds for the derivation of $\beta$
%\end{proof}

\begin{proof}[\textbf{Proof of Theorem~\ref{thm: accuracy stat of kNN BayBox estimator}}]
    \label{proof: for thm: accuracy stat of kNN BayBox estimator} The proof of part 1) of the theorem follows in exact analogy to the proof of Lemma \ref{proof: for lemma: accuracy stat of general BayBox estimator} and we do not repeat it here. Now, we turn to the proof of part 2).
    Again we restrict ourselves to proving the statement about the type-I-errors $\abs{\tilde{\alpha}(\eta) - \alpha(\eta)} \leq \sqrt{\frac{1}{2n}\ln{\frac{4}{\gamma}}} + \sqrt{\frac{144c_d^2}{n}\ln{\frac{4}{\gamma}}}$, and the statement on type-II-errors follows by a similar approach. 
    With probability at least $1 - \gamma/2$, we have 
    \begin{align*}
           & \abs{\tilde{\alpha}(\eta) - \alpha(\eta)}\\
        ={}& \abs{\frac{1}{n}\rSum{i}{1}{n}\kNNclassifier{n}(\rX_i) - \Ex{\frac{1}{n}\rSum{i}{1}{n}\phi^*(\rX_i)}} \tag{$\rX_i \overset{\text{i.i.d.}}{\sim} P$}\\
        ={}& \abs{\frac{1}{n}\rSum{i}{1}{n}\kNNclassifier{n}(\rX_i) - \Ex{\phi^*(\rX)}} \tag{$\rX \sim P$}\\
        \leq{}& \abs{\frac{1}{n}\rSum{i}{1}{n}\kNNclassifier{n}(\rX_i) - \Ex{\kNNclassifier{n}(\rX)}} + \abs{\Ex{\kNNclassifier{n}(\rX)} - \Ex{\phi^*(\rX)}} \\
        \leq{}& \sqrt{\frac{1}{2n}\ln{\frac{4}{\gamma}}} + \abs{\Ex{\kNNclassifier{n}(\rX)} - \Ex{\phi^*(\rX)}} \tag{by Hoeffding's Inequality}\\
        ={}& \sqrt{\frac{1}{2n}\ln{\frac{4}{\gamma}}} + \abs{ \pr{\kNNclassifier{n}(\rX) = 1} - \pr{\phi^*(\rX) = 1} }\\
        ={}& \sqrt{\frac{1}{2n}\ln{\frac{4}{\gamma}}} + \abs{ \pr{\kNNclassifier{n}(\rX) \neq 0} - \pr{\phi^*(\rX) \neq 0} }\\
        \leq{}& \sqrt{\frac{1}{2n}\ln{\frac{4}{\gamma}}} + 2|R(\kNNclassifier{n}) - R(\phi^{*})|\\
        \leq{}& \sqrt{\frac{1}{2n}\ln{\frac{4}{\gamma}}} + 12\sqrt{\frac{2c_d^2}{n}\ln{\frac{4}{\gamma}}}\tag{by Theorem~\ref{thm:covergence of kNN}}.
    \end{align*}

    We note that the first equality follows the idea in the proof of Lemma~\ref{lemma: accuracy stat of general BayBox estimator}, by just replacing the Bayes classifier with the concrete $k$-NN classifier.
\end{proof}

\begin{proof}[\textbf{Proof of Theorem~\ref{theo:auditor}}] \B{To enhance the clarity of this proof, we will additionally assume that the curve $T^{(0)}$ is strictly decaying. We first need to understand the interpretation of lines 6 and 7 of the algorithm. The algorithm detects a violation, if
\[
i^* > \tilde{\alpha}(\hat{\eta}^*) + w(\gamma),
\]
where $i^*$ solves the equation $T^{(0)}(i^*) = \tilde{\beta}(\hat{\eta}^*) + w(\gamma)$.
We apply $T^{(0)}$ on both sides, which gives us the detection condition
\begin{align} \label{e:condition}
\tilde{\beta}(\hat{\eta}^*) + w(\gamma) <T^{(0)}(\tilde{\alpha}(\hat{\eta}^*) + w(\gamma)).
\end{align}
Geometrically this means that the point $(\tilde{\alpha}(\hat{\eta}^*) + w(\gamma),\tilde{\beta}(\hat{\eta}^*) + w(\gamma) )$ is below the curve $T^{(0)}$  and since $T^{(0)}$ is a trade-off curve, it follows that the entire box $\square_\gamma$ is below  $T^{(0)}$. Conversely, if the detection condition is violated, we have 
\begin{align} \label{e:condition2}
\tilde{\beta}(\hat{\eta}^*) + w(\gamma) \ge T^{(0)}(\tilde{\alpha}(\hat{\eta}^*) + w(\gamma))
\end{align}
and the upper right edge point of the box $\square_\gamma$ is on or above $T^{(0)}.$\\
Now, suppose there was no violation (part 1) of the theorem). Then, any point on or above $T$ is also on or above $T^{(0)}$. The point $(\mathbb{E} \tilde \alpha(\eta),\mathbb{E}  \tilde \beta(\eta))$ is on or above $T$ and thus on or above $T^{(0)}$. Now, according to Theorem \ref{thm: accuracy stat of kNN BayBox estimator} it holds with probability $\ge 1-\gamma$ that the following event occurs
\[
\mathcal{E}=\{(\mathbb{E} \tilde \alpha(\eta),\mathbb{E}  \tilde \beta(\eta)) \in \square_\gamma\}
\]
Conditional on that event the upper right edge point of $\square_\gamma$, namely $(\tilde{\alpha}(\hat{\eta}^*) + w(\gamma), \tilde{\beta}(\hat{\eta}^*) + w(\gamma))$ is also above $T$. It is hence above $T^{(0)}$ and satisfies condition \eqref{e:condition2} and no privacy violation is detected. \\}
Now, in part 2), suppose that there exists a privacy violation. The trade-off function is strictly convex and it is not hard to see that this implies that it equals the set $\{(\alpha(\eta), \beta(\eta): \eta \ge 0\}$ in this case (the constant $\lambda$ in the Neyman-Pearson test can be set to $0$ everywhere). We also define the maximum violation
\[
v^* = \sup_{\alpha \in [0,1]}\big[ T^{(0)}(\alpha)-T(\alpha)\big]
\]
and the set of thresholds
\[
\Psi:= \big\{\eta \ge 0:T^{(0)}(\alpha(\eta))-T(\alpha(\eta)) \ge  v^*/2\big\}.
\]
It holds by the proof of Theorem \ref{theo:1} case 1) that 
\[
\sup_\eta |\hat \alpha_h(\eta) - \alpha(\eta)|\overset{P}{\to} 0, \quad as \,\,\, n_1 \to \infty.
\]
In particular, it follows that
\[
\Pr[\hat \eta^* \in \Psi]= 1-r_{n_1},
\]
where $r_{n_1} \to 0$ as $n_1 \to \infty$.
If the above statement were false, it would follow on an event with asymptotically positive probability that 
\[
T^{(0)}(\alpha(\hat \eta^*))-T(\alpha(\hat \eta^*)) \le (1/2) v^*
\]
leading to a contradiction with Proposition \ref{prop1}. Now, we condition on the event $\{\hat \eta^* \in \Psi\}$ and pass the parameter to the BayBox estimator, which returns the estimator pair $(\tilde{\alpha}(\hat{\eta}^*), \tilde{\beta}(\hat{\eta}^*))$. Now, keeping $n_1$ fixed and letting $n_2 \to \infty$ it follows that (part 2) of Theorem \ref{thm: accuracy stat of kNN BayBox estimator})
\begin{align*}
   & \tilde{\alpha}(\hat{\eta}^*) + w(\gamma) \overset{P}{\to} \alpha(\hat{\eta}^*), \quad  \tilde{\beta}(\hat{\eta}^*) + w(\gamma) \to \beta(\hat{\eta}^*).
\end{align*}
Given the continuity of the function $T^{(0)}$ (every trade-off function is continuous) it follows that conditionally on $\Psi$
\begin{align*}
&T^{(0)}(\tilde{\alpha}(\hat{\eta}^*) + w(\gamma)) \to T^{(0)}(\alpha(\hat{\eta}^*)) \ge T(\alpha(\hat{\eta}^*) +v^*/2\\
= &\beta(\hat{\eta}^*) +\nu^*/2>  \beta(\hat{\eta}^*)
\end{align*}
and the detection condition in \eqref{e:condition} is asymptotically fulfilled as $n_2 \to \infty$. Thus, we have 
\[
\lim_{n_2 \to \infty}\Pr[A = \textnormal{"Violation"}| \{\hat \eta^* \in \Psi\}]  =1
\] 
and hence
\[
\liminf_{n_2 \to \infty}\Pr[A = \textnormal{"Violation"}]\ge 1-r_{n_1}.
\]
Taking the limit $n_1 \to \infty$ we have $r_{n_1} \to 0$ and the result follows.
\end{proof}

\end{document}